\documentclass{birkjour}
\usepackage{amsthm}
\usepackage[T1]{fontenc}
\usepackage{amsmath}
\usepackage{amssymb}
\usepackage{dsfont}
\usepackage{amsxtra}
\usepackage{hyperref}
\usepackage[utf8]{inputenc}
\usepackage[english]{babel}
\usepackage[capitalize]{cleveref}
\usepackage{microtype}

\usepackage[shortlabels]{enumitem}
\SetEnumitemKey{:arabic}{
	label=(\arabic*)}
\setlist[enumerate]{:arabic}

\usepackage{bm}
\usepackage{mathrsfs}
\usepackage{mathtools,empheq}
\usepackage{etoolbox,xcolor}
\usepackage{xparse,tikz,ragged2e}

\newtheorem{theo}{Theorem}[section]
\newtheorem{cor}[theo]{Corollary}
\newtheorem{lem}[theo]{Lemma}
\newtheorem{prop}[theo]{Proposition}
\theoremstyle{definition}
\newtheorem{de}[theo]{Definition}
\theoremstyle{remark}
\newtheorem{rem}[theo]{Remark}

\numberwithin{equation}{section}

\crefname{lem}{Lemma}{Lemmas}
\crefname{theo}{Theorem}{Theorems}
\crefrangelabelformat{equation}{(#3#1#4)-(#5#2#6)}

\DeclareMathOperator{\tr}{Tr}
\DeclareMathOperator{\supp}{supp}

\DeclarePairedDelimiter\myp()
\DeclarePairedDelimiter\myt\{\}
\DeclarePairedDelimiter\myb[]
\DeclarePairedDelimiter\abs\lvert\rvert
\DeclarePairedDelimiter\norm\lVert\rVert

\newcommand{\id}{\, \mathrm{d}}
\newcommand{\ketbra}[2]{\lvert #1 \rangle \langle #2 \rvert}

\NewDocumentCommand\normt{ s o m m }{
	\IfBooleanTF{#1}{
		\norm*{#3}
	}{% nostar
		\IfNoValueTF{#2}{
			\norm{#3}
		}{
			\norm[#2]{#3}
		}
	}
	_{#4}
}

\DeclarePairedDelimiterX\innerp[2]\langle\rangle{
	#1,#2
}

\let\Set\relax

\providecommand\given{}
\newcommand\SetSymbol[1][]{%
	\nonscript\:#1\vert
	\allowbreak
	\nonscript\:
	\mathopen{}}
\DeclarePairedDelimiterX\Set[1]\{\}{%
	\renewcommand\given{\SetSymbol[\delimsize]}
	#1
}

\newcommand\sssup[1]{
	{\scriptscriptstyle\mathrm{#1}}
}

\newcommand\nonarg{{}\cdot{}}

\NewDocumentCommand\For{ s m }{
	\IfBooleanF{#1}{\quad}
	\quad
	\text{#2}
}

\begin{document}

\title[Fermionic systems in magnetic fields]{Semi-classical limit of confined fermionic systems in homogeneous magnetic fields}
\date{\today}

\author[S. Fournais]{Søren Fournais}
\address{Department of Mathematics, Aarhus University, Ny Munkegade 118, DK-8000 Aarhus C, Denmark} 
\email{fournais@math.au.dk}

\author[P.S. Madsen]{Peter S. Madsen}
\address{CNRS \& CEREMADE, Paris-Dauphine University, PSL University, 75016 Paris, France} 
\email{madsen@mceremade.dauphine.fr}

\begin{abstract}
	We consider a system of $ N $ interacting fermions in $ \mathbb{R}^3 $ confined by an external potential and in the presence of a homogeneous magnetic field.
	The intensity of the interaction has the mean-field scaling $ 1/N $.
	With a semi-classical parameter $ \hbar \sim N^{-1/3} $, we prove convergence in the large $ N $ limit to the appropriate Magnetic Thomas-Fermi type model with various strength scalings of the magnetic field.
\end{abstract}

\date{today}
\maketitle

\setcounter{tocdepth}{1}
\tableofcontents

%%%%%%%%%%%%%%%%%%%%%%%%%%%%%%%%%%%%%%%%%%%%%%%%%%%%%
%%%%%%%%%%%%%%%%%%%%%%%%%%%%%%%%%%%%%%%%%%%%%%%%%%%%%
\section{Introduction and main results}
\label{sec:intro}
%%%%%%%%%%%%%%%%%%%%%%%%%%%%%%%%%%%%%%%%%%%%%%%%%%%%%
%%%%%%%%%%%%%%%%%%%%%%%%%%%%%%%%%%%%%%%%%%%%%%%%%%%%%
We consider a system of $ N $ fermionic particles in $ \mathbb{R}^3 $ with an exterior potential $ V $, and with the particles interacting pairwise through a potential $ w $.
The system is in the presence of a homogeneous magnetic field pointing along the $ z $-direction, i.e. of the form $ \bm{B} = \myp{0,0,b} $ for some $ b > 0 $.
That is, we can take the magnetic vector potential to be $ b A \myp{x} = \frac{b}{2} \myp{-x_2,x_1,0} $.

\subsection{The quantum mechanical model}
Given parameters $ \hbar, b > 0 $, we consider the mean-field Hamiltonian operator
\begin{equation}
\label{eq:hamiltonian}
	H_{N,\hbar,b} := \sum\limits_{j=1}^N \big( \myp{ \bm{\sigma} \cdot \myp{ -i \hbar \nabla_j + bA \myp{x_j}} }^2 + V\myp{x_j} \big) + \frac{1}{N} \sum\limits_{j < k}^N w \myp{x_j-x_k},
\end{equation}
where $ \bm{\sigma} = \myp{\sigma_1,\sigma_2,\sigma_3} $ is the vector of Pauli spin matrices
\begin{equation*}
	\sigma_1 =	\begin{pmatrix}
					0 & 1 \\
					1 & 0
				\end{pmatrix},
	\quad 
	\sigma_2 = \begin{pmatrix}
					0 & -i \\
					i & 0
	\end{pmatrix},
	\quad
	\sigma_3 = \begin{pmatrix}
					1 & 0 \\
					0 & -1
	\end{pmatrix}.
\end{equation*}
Note that the Pauli operator can also be written in the equivalent form
\begin{align}
	\myp{ \bm{\sigma} \cdot \myp{ -i \hbar \nabla + bA} }^2
	&= \myp{-i\hbar \nabla + bA}^2 \mathds{1}_{\mathbb{C}^2} + \hbar \bm{\sigma} \cdot \bm{B} \nonumber \\
	&= \myp{-i\hbar \nabla + bA}^2 \mathds{1}_{\mathbb{C}^2} + \hbar b \sigma_3.
\label{eq:paulieq}	
\end{align}
Since we are dealing with fermions, the operator $ H_{N,\hbar,b} $ must be restricted to the subspace $ \bigwedge\nolimits^N L^2 \myp{\mathbb{R}^3 ; \mathbb{C}^2} \subseteq L^2 \big( \mathbb{R}^{3N} ; \mathbb{C}^{2^N} \big) $ of anti-symmetric wave functions.
The anti-symmetry is due to the Pauli exclusion principle, stating that two identical fermionic particles cannot occupy the same quantum state.
The fact that the system is in a mean-field scaling is expressed by the prefactor $ 1/N $ in front of the interaction.
In mathematics, many-body fermionic systems in strong homogeneous magnetic fields have been considered before \cite{LieSolYng-06,LieSolYng-94,LieSolYng-95,Seiringer-01,HaiSei-01,HaiSei-01b} with Coulomb interactions (See also \cite{Ivrii-96,Ivrii-96b,Ivrii-97,Ivrii-99} and \cite{Sobolev-94} for more precise asymptotic results), and also at positive temperature \cite{HauYng-04} in the context of pressure functionals.
For references to the physics literature, see e.g. \cite{LieSolYng-06,LieSolYng-94}.
\begin{rem}[Relation between parameters]\label{rem:params}
	If the system is confined to a bounded domain, then the kinetic energy satisfies the usual (non-magnetic) Lieb-Thirring inequality
	\begin{align*}
		\MoveEqLeft[3] \innerp[\Big]{ \Psi }{ \sum\limits_{j=1}^N \myp{ \bm{\sigma} \cdot \myp{ -i \hbar \nabla_j + bA \myp{x_j}} }^2 \Psi } \\
		&\geq c \hbar^2 \int_{\mathbb{R}^3} \myp[\big]{ \rho_{\Psi}^{\myp{1}} \myp{x} }^{\frac{5}{3}} \id x - \hbar b N 
		\geq \tilde{c} \hbar^2 N^{\frac{5}{3}} - \hbar b N,
	\end{align*}
	(here the boundedness of the domain is used to get the second inequality) where $ \rho_{\Psi}^{\myp{1}} $ is the one-particle reduced position density of the normalized wave function	$ \Psi \in \bigwedge\nolimits^N L^2 \myp{\mathbb{R}^3 ; \mathbb{C}^2} $, defined in \eqref{eq:kpartdensity} below.
	This means, in case the magnetic field strength is small or vanishing (more precisely, if $ \hbar b \to 0 $), that we need to take $ \hbar $ of order 
	\begin{equation}
	\label{eq:scalingweak}
		\hbar \sim N^{-\frac{1}{3}},
	\end{equation}
	in order for the terms in \eqref{eq:hamiltonian} to be of the same order in $ N $.
	In the opposite case with a strong magnetic field (that is,	$ \hbar b \gg 1 $), one expects, by the magnetic Lieb-Thirring
	inequality (\eqref{eq:genlt} below),
	\begin{equation*}
		\innerp[\Big]{ \Psi}{ \sum\limits_{j=1}^N \myp{ \bm{\sigma} \cdot \myp{ -i \hbar \nabla_j + bA \myp{x_j}} }^2 \Psi }
		\gtrsim c \hbar^2 \frac{1}{\myp{b/\hbar}^2} \int \myp{ \rho_{\Psi}^{\myp{1}} \myp{x} }^3 \id x.
	\end{equation*}
	This inequality is not rigorous, but it is reasonable in a strong magnetic field where all particles are confined to the lowest Landau
	level.
	Assuming the inequality to hold, we get
	\begin{equation*}
		\innerp[\Big]{ \Psi}{ \sum\limits_{j=1}^N \myp{ \bm{\sigma} \cdot \myp{ -i \hbar \nabla_j + bA \myp{x_j}} }^2 \Psi } \geq \tilde{c} \frac{\hbar^4}{b^2} N^3
		= \tilde{c} \frac{\myp{\hbar^3 N}^2}{\myp{\hbar b}^2} N,
	\end{equation*}
	so in order to have energy balance we take in this case
	\begin{equation}
	\label{eq:scalingstrong}
		\hbar^3 N \sim \hbar b.
	\end{equation}
\end{rem}

Based on the observations of Remark~\ref{rem:params}, we introduce a new parameter $ \beta \geq 0 $ and define $ \hbar $ and $ b $ by
\begin{equation}
\label{eq:scaling0}
	\hbar := N^{-\frac{1}{3}} \myp{1+\beta}^{\frac{1}{5}},
	\qquad
	b := N^{\frac{1}{3}} \beta \myp{1+\beta}^{-\frac{3}{5}},
\end{equation}
or, equivalently, by the relations
\begin{equation}
\label{eq:scaling}
	\hbar b = \beta \myp{1+\beta}^{-\frac{2}{5}},
	\qquad \hbar^3 N = \myp{1+\beta}^{\frac{3}{5}},
	\qquad \frac{b}{\hbar^2 N} = \frac{\beta}{1+\beta}.
\end{equation}
This scaling convention interpolates between the two extreme cases \eqref{eq:scalingweak} ($ \beta = 0 $) and \eqref{eq:scalingstrong}
($ \beta \gg 1 $).
The notation is chosen to fit the notation in \cite{LieSolYng-94}, see also \cref{rem:mean-field} below.

In this paper we analyze the semi-classical limit of \eqref{eq:hamiltonian} as the number of particles tends to infinity and $ \hbar $ tends to zero.  
In light of \eqref{eq:scaling0} we must thus require that
\begin{equation}
\label{eq:scaling2}
	\lim_{N \to \infty} N^{- \frac{1}{3}} \beta^{\frac{1}{5}} = 0,
\end{equation}
in order to stay in the semi-classical regime, $\hbar \rightarrow 0$.
Being in the semi-classical regime is essential for our analysis since the extraction of semi-classical limiting measures depends on it.
In \cite{LieSolYng-06} limit models for the non-semi-classical regions, where \eqref{eq:scaling2} fails, are deduced.
At present, we do not know how to obtain such results with our (de Finetti type) techniques.

When $ \hbar $ and $ b $ satisfy the scaling convention \eqref{eq:scaling0}, we will instead denote the Hamiltonian \eqref{eq:hamiltonian} by $ H_{N,\beta} $, and the ground state energy of $ H_{N,\beta} $ restricted to $ \bigwedge^N L^2 \myp{\mathbb{R}^3;\mathbb{C}^2} $ will be denoted by
\begin{equation}
	E \myp{N,\beta} := \inf \sigma_{\bigwedge^N L^2 \myp{\mathbb{R}^3;\mathbb{C}^2}} \myp{H_{N,\beta}}.
\end{equation}

\begin{rem}
	\label{rem:mean-field}
	Physical systems do not usually come with a mean-field scaling, but	the Hamiltonian in question can sometimes be put in the form	\eqref{eq:hamiltonian} by rescaling appropriately.
	This is true e.g. for atoms \cite{LieSim-77,LieSim-77b,LieSolYng-94} (see also \cite{Thirring-81} where coherent states are used) and non-relativistic white dwarfs \cite{LieThi-84,LieYau-87}, with or without magnetic fields.
	In the case of an atom in a homogeneous	magnetic field of strength $ B $, the Hamiltonian is
	\begin{equation*}
		H_{N,B,Z} = \sum\limits_{j=1}^N \myp[\Big]{ \myp{ \bm{\sigma} \cdot \myp{ -i \nabla_j + B A \myp{x_j}} }^2 - \frac{Z}{\abs{x_j}} } + \sum\limits_{1 \leq j < k \leq N} \frac{1}{\abs{x_j-x_k}}.
	\end{equation*}
	Choosing parameters $ \beta := B Z^{-4/3} $ and	$ \ell := Z^{-1/3} \myp{1+\beta}^{-2/5} $, then	$ H_{N,B,Z} $ is for $ Z=N $ unitarily equivalent to
	$ Z\ell^{-1} H_{N,\beta} $, where $ H_{N,\beta} $ is given by \eqref{eq:hamiltonian} with	$ V \myp{x} = w \myp{x} = \abs{x}^{-1} $, and $ \hbar $ and $ b $ defined by \eqref{eq:scaling0}.
	The methods in the present paper cover the case of neutral systems ($ N = Z $) in the regime $ B \ll Z^3 $, cf. \cite{LieSolYng-94}.
\end{rem}
The analysis naturally splits into three cases, depending on the asymptotics of the parameter $ \beta $.
In the first case, when $ \beta \to 0 $, the presence of the magnetic field has no effect on the ground state energy of $ H_{N,\beta} $ to leading order, and the energy in the limit is described by the usual non-magnetic Thomas-Fermi theory.
In the second case, when $ \beta \to \beta_0 $ for some $ \beta_0 \in \myp{0,\infty} $, the energy in the semi-classical limit is described by a magnetic Thomas-Fermi theory, as already seen in \cite{LieSolYng-94} in the case of Coulomb interactions.
In the third case, when $ \beta $ goes to infinity, corresponding to a strong magnetic field, all particles are forced to stay in the lowest Landau band of the magnetic Laplacian, and the limit is described by a strong-field Thomas-Fermi theory.

Suitable upper bounds on the energy $ E \myp{N,\beta} $ will be provided by constructing appropriate trial states.
This is done by localizing in position space and using the Weyl asymptotics for magnetic Schrödinger operators obtained in \cite{LieSolYng-94} to contruct compactly supported Slater determinants.
Corresponding lower bounds are obtained using coherent states along with a fermionic (classical) de Finetti-Hewitt-Savage theorem.
The usefulness of classical de Finetti theorems \cite{DeFinetti-31,HewSav-55,DiaFre-80} has been known for a long time in the context of classical mechanics (e.g. \cite{BraHep-77,Spohn-81,MesSpo-82,CagLioMarPul-92,Kiessling-93}).
Recently, quantum de Finetti type theorems \cite{Stormer-69,HudMoo-75} have also been used to study mean-field problems in quantum mechanics in works by Lewin, Nam and Rougerie \cite{LewNamRou-14,LewNamRou-15b,LewNamRou-16c}, where the ground state energy of a mean-field \emph{Bose} system under rather general assumptions is shown to converge to the Hartree energy of the system.
The idea is further developed by Fournais, Lewin and Solovej in \cite{FouLewSol-18}, where it is used to treat the case of spinless fermions in weak magnetic fields, and by Lewin, Triay, and the second author of the present article to treat the case of Fermi systems at positive temperature, also in weak magnetic fields \cite{LewMadTri-19}.
See also \cite{Rougerie-15} for a thorough discussion of de Finetti theorems.
One of the main motivations for the present work is to extend the de Finetti technique to magnetic semiclassics.

We briefly remind the reader of the well-known fact that the spectrum of the Pauli operator $ \myp{ \bm{\sigma} \cdot \myp{ -i \hbar \nabla + bA \myp{x}} }^2 $ is parametrized by the Landau bands
\begin{equation}
\label{eq:landaubands}
	p^2 + 2 \hbar b j, \For{$ p \in \mathbb{R} $, $ j \in \mathbb{N}_0 $.}
\end{equation}
(See also equations \eqref{eq:HA} and \eqref{eq:uniequiv} in \cref{sec:semimeasures}).
Therefore, the phase space naturally becomes
\begin{equation}
\label{eq:phasespace}
	\Omega = \mathbb{R}^3 \times \mathbb{R} \times {\mathbb{N}}_0 \times \Set{\pm 1},
\end{equation}
where $ \mathbb{R}^3 $ is position, $ \mathbb{R} $ is momentum in the $ z $-direction (along the magnetic field), $ \mathbb{N}_0 $ is the quantized radius of the cyclotron orbit in 2D, and $ \Set{\pm 1} $ is the spin variable.
We will denote components of vectors $ \xi \in \Omega^k $ by $ \xi_{\ell} = \myp{u_{\ell},p_{\ell}, j_{\ell}, s_{\ell} } \in \Omega $.
For notational convenience, we will sometimes rearrange the variables by separating them into position and momentum components, i.e.
\begin{equation}
\label{eq:xi}
	\xi = \myp{u,p,j,s}, 
\end{equation}
with $ u \in \mathbb{R}^{3k} $, $ p \in \mathbb{R}^k $, $ j \in \mathbb{N}_0^k $, and $ s \in \Set{\pm 1}^k $.
Integration over $ \Omega^k $ will be done with respect to its natural measure, using the Lebesgue measure in the continuous variables and the counting measure in the discrete variables.

\subsection{Magnetic Thomas-Fermi theories}
We recall that the pressure 
(that is, the grand canonical minimal energy in the thermodynamic limit)
of the free Landau gas \cite[equation $ \myp{4.47} $]{LieSolYng-94}, i.e. a gas of non-interacting fermions in a homogeneous magnetic field, at chemical potential $ \nu \geq 0 $ is given by
\begin{equation}
\label{eq:landaupressure}
	P_B \myp{\nu} = \frac{B}{3 \pi^2} \myp[\Big]{ \nu^{\frac{3}{2}} + 2 \sum\limits_{j=1}^{\infty} \myb{2 j B -\nu}_-^{\frac{3}{2}} }.
\end{equation}
Here $ B > 0 $ is the magnetic field strength, and $ \gamma_- := \max \myp{-\gamma,0} $ denotes the negative part of a number $ \gamma \in \mathbb{R} $.
Clearly, $ P_B $ is a convex and continuously differentiable function with derivative
\begin{equation}
\label{eq:landaupressurederiv}
	P'_B \myp{\nu} = \frac{B}{2 \pi^2} \myp[\Big]{ \nu^{\frac{1}{2}} + 2 \sum\limits_{j=1}^{\infty} \myb{2 j B -\nu}_-^{\frac{1}{2}} }.
\end{equation}
In the mean-field scaling we will take $ B = \beta \myp{1+\beta}^{-2/5} $, and the pressure will come with an additional prefactor $ \myp{1+\beta}^{-3/5} $, cf. \eqref{eq:scaling}. Furthermore, the constant
\begin{equation}
	\label{eq:kbeta}
	k_{\beta} := \beta \myp{1+\beta}^{-\frac{2}{5}}
\end{equation}
will play a recurring role in the following.
It arises as half the distance between the Landau bands in the semi-classical limit (see \eqref{eq:scaling} and \eqref{eq:landaubands}).

\begin{de}[Magnetic Thomas-Fermi energy]
	Let $ V \in L_{\mathrm{loc}}^{5/2} \myp{\mathbb{R}^3} $	satisfy $ V \myp{x} \to \infty $ as $ \abs{x} \to \infty $, and let $ w \in L^{5/2} \myp{\mathbb{R}^3} + L_{\varepsilon}^{\infty} \myp{\mathbb{R}^3} $ be even.
	We define the magnetic Thomas-Fermi energy functional with parameter $ \beta > 0 $ by
	\begin{align}
		\mathcal{E}_{\beta}^{\sssup{MTF}} \myp{\rho} 
		= {}& \int_{\mathbb{R}^3} \tau_{\beta} \myp{\rho \myp{x}} \id x + \int_{\mathbb{R}^3} V \myp{x} \rho \myp{x} \id x \nonumber \\
		&+ \frac{1}{2} \iint_{\mathbb{R}^6} w \myp{x-y} \rho \myp{x} \rho \myp{y} \id x \id y
	\label{eq:mtffunct}
	\end{align}
	on the set
	\begin{equation*}
		\mathcal{D}^{\sssup{MTF}} = \Set[\big]{ \rho \in L^1 \myp{\mathbb{R}^3} \cap L^{\frac{5}{3}} \myp{\mathbb{R}^3} \given 0 \leq \rho, \ V \rho \in L^1 \myp{\mathbb{R}^3} },
	\end{equation*}
	where the energy density $ \tau_{\beta} $ is given by the Legendre transform of the scaled pressure
	\begin{equation}
	\label{eq:taubdef}
		\tau_{\beta} \myp{t} 
		= \sup_{\nu \geq 0} \myp[\big]{ t\nu - \myp{1+\beta}^{-\frac{3}{5}} P_{k_{\beta}} \myp{\nu} },
	\end{equation}
	with $ k_{\beta} = \beta \myp{1+\beta}^{-2/5} $ as in \eqref{eq:kbeta}.
	($ \tau_{\beta} \myp{t} $ is the canonical ground state energy of the canonical free Landau gas in the thermodynamic limit, at density $ t \geq 0 $ and in mean-field scaling).
	Furthermore, the magnetic Thomas-Fermi ground state energy is defined as the infimum
	\begin{equation}
	\label{eq:mtfenergy}
		E^{\sssup{MTF}} \myp{\beta} = \inf \Set[\big]{ \mathcal{E}_{\beta}^{\sssup{MTF}} \myp{\rho} \given \rho \in \mathcal{D}^{\sssup{MTF}}, \ \int_{\mathbb{R}^3} \rho \myp{x} \id x = 1 }.
	\end{equation}
\end{de}
Recall that the space $ L^{5/2} \myp{\mathbb{R}^3} + L_{\varepsilon}^{\infty} \myp{\mathbb{R}^3} $ consists of functions $ f $ satisfying that for each $ \varepsilon > 0 $ there exist $ f_1 \in L^{5/2} \myp{\mathbb{R}^3} $ and $ f_2 \in L^{\infty} \myp{\mathbb{R}^3} $ with $ \normt{f_2}{\infty} \leq \varepsilon $ and $ f = f_1 + f_2 $.
\begin{rem}
	\label{rem:mtfdomain}
	For any $ 0 \leq \rho \in L^1 \myp{\mathbb{R}^3} $ and $ \beta > 0 $ it is well known \cite[Proposition $ 4.2 $]{LieSolYng-94} that $ \rho \in L^{5/3} \myp{\mathbb{R}^3} $ if and only if $ \int \tau_{\beta} \myp{\rho \myp{x}} \id x $ is finite.
	In particular, we have the bound
	\begin{align}
		\int_{\mathbb{R}^3} \rho \myp{x}^{\frac{5}{3}} \id x
		\leq {}& \kappa_1 \myp{1+ \beta}^{\frac{2}{5}} \int_{\mathbb{R}^3} \tau_{\beta} \myp{\rho \myp{x}} \id x \nonumber \\
		&+ \kappa_2 \myp[\Big]{ \frac{\beta}{1+\beta} }^{\frac{2}{5}} \normt{\rho}{1}^{\frac{2}{3}} \myp[\Big]{\int_{\mathbb{R}^3} \tau_{\beta} \myp{\rho \myp{x}} \id x}^{\frac{1}{3}}
	\label{eq:rhoest2}
	\end{align}
	for some constants $ \kappa_1, \kappa_2 > 0 $.
	It follows that the domain of $ \mathcal{E}_{\beta}^{\sssup{MTF}} $ indeed is as stated above.
	Furthermore, it is not difficult to show that the functional is bounded from below on $ \mathcal{D}^{\sssup{MTF}} \cap \Set{\rho \given \int \rho \leq M} $ for each $ M > 0 $, when $ V $ and $ w $ satisfy the assumptions stated above.
\end{rem}
Similarly, we also define the \emph{strong Thomas-Fermi functional}
\begin{align*}
	\mathcal{E}^{\sssup{STF}} \myp{\rho} 
	= {}& \frac{4 \pi^4}{3} \int_{\mathbb{R}^3} \rho \myp{x}^3 \id x + \int_{\mathbb{R}^3} V \myp{x} \rho \myp{x} \id x \\
	&+ \frac{1}{2} \iint_{\mathbb{R}^6} w \myp{x-y} \rho \myp{x} \rho \myp{y} \id x \id y,
\end{align*}
along with the ordinary non-magnetic \emph{Thomas-Fermi functional}
\begin{align*}
	\mathcal{E}^{\sssup{TF}} \myp{\rho} 
	= {}& \frac{3}{5} c_{\sssup{TF}} \int_{\mathbb{R}^3} \rho \myp{x}^{\frac{5}{3}} \id x + \int_{\mathbb{R}^3} V \myp{x} \rho \myp{x} \id x \\
	&+ \frac{1}{2} \iint_{\mathbb{R}^6} w \myp{x-y} \rho \myp{x} \rho \myp{y} \id x \id y,
\end{align*}
where $ c_{\sssup{TF}} = \myp{3 \pi^2}^{2/3} $, and with corresponding ground state energies $ E^{\sssup{STF}} $ and $ E^{\sssup{TF}} $, both defined in complete analogy to \eqref{eq:mtfenergy}.

\subsubsection*{MTF theory with spin}
We also introduce a version of the magnetic Thomas-Fermi functional that keeps better track of the spin dependence.
We will mostly need this for the formulation of our main result, and most of the proofs in the paper will be done using the spin-summed version defined above.
We use a tilde ($ \sim $) to distinguish the relevant spin-dependent quantities from the corresponding spin-independent (or spin-summed)
ones.
The pressure of the free Laundau gas with complete spin polarization is
\begin{equation}
\label{eq:spinlandaupressure}
	\widetilde{P}_B \myp{\nu,s} = \frac{B}{3 \pi^2} \sum\limits_{j=0}^{\infty} \myb{ B \myp{2j+1+s} -\nu}_-^{\frac{3}{2}},
\end{equation}
where $ B > 0 $ is the magnetic field strength and $ s \in \Set{\pm 1} $ denotes the spin variable.
Note that the spin-summed pressure \eqref{eq:landaupressure} is recovered by summing the components of $ \widetilde{P}_B $, i.e. $ P_B \myp{\nu} = \widetilde{P}_B \myp{\nu,-1}+\widetilde{P}_B \myp{\nu,1} $.
As in the spin-summed case, $ \widetilde{P}_B $ is convex and continuously differentiable with derivative
\begin{equation}
\label{eq:spinlandaupressurederiv}
	\widetilde{P}_B' \myp{\nu,s} := \frac{\partial \widetilde{P}_B}{\partial \nu}\myp{\nu,s} = \frac{B}{2 \pi^2} \sum\limits_{j=0}^{\infty} \myb{ B \myp{2j+1+s} -\nu}_-^{\frac{1}{2}}.
\end{equation}
Again, the kinetic energy density is the Legendre transform of the scaled pressure
\begin{equation}
\label{eq:spintaubdef}
	\widetilde{\tau}_{\beta} \myp{t,s} 
	= \sup_{\nu \geq 0} \myp[\big]{ t\nu - \myp{1+\beta}^{-\frac{3}{5}} \widetilde{P}_{k_{\beta}} \myp{\nu,s} },
\end{equation}
with $ k_{\beta} $ given by \eqref{eq:kbeta}, and the corresponding energy functional
\begin{align}
	\widetilde{\mathcal{E}}_{\beta}^{\sssup{MTF}} \myp{\rho}
	= {}& \sum\limits_{s=\pm 1} \int_{\mathbb{R}^3} \widetilde{\tau}_{\beta} \myp{\rho \myp{x,s},s} \id x + \sum\limits_{s=\pm 1} \int_{\mathbb{R}^3} V \myp{x} \rho \myp{x,s} \id x \nonumber \\
	&+ \frac{1}{2} \sum\limits_{s_1,s_2 = \pm 1} \iint_{\mathbb{R}^6} w \myp{x-y} \rho \myp{x,s_1} \rho \myp{y,s_2} \id x \id y
\label{eq:spinmtffunct}
\end{align}
is defined on the set of densities
\begin{equation*}
	\widetilde{\mathcal{D}}^{\sssup{MTF}} 
	= \Set[\big]{ \rho \in L^1 \myp{\mathbb{R}^3 \times \Set{\pm 1}} \cap L^{\frac{5}{3}} \myp{\mathbb{R}^3 \times \Set{\pm 1}} \given 0 \leq \rho, \ V \rho \in L^1 \myp{\mathbb{R}^3 \times \Set{\pm 1}} }.
\end{equation*}
The fact that $ \widetilde{\mathcal{E}}_{\beta}^{\sssup{MTF}} $ is well defined on $ \widetilde{\mathcal{D}}^{\sssup{MTF}} $ is easily seen by showing the elementary bounds
\begin{equation*}
	2 \widetilde{\tau}_{\beta} \myp{t,-1} 
	\leq \tau_{\beta} \myp{2t} 
	\leq \widetilde{\tau}_{\beta} \myp{t,-1} + \widetilde{\tau}_{\beta} \myp{t,1} 
	\leq 2\widetilde{\tau}_{\beta} \myp{t,-1} + 4 k_{\beta} t
\end{equation*}
for each $ t \geq 0 $, and combining with the description of the domain of the spin-summed magnetic Thomas-Fermi functional in \cref{rem:mtfdomain}.
In \cref{sec:prelim} we will argue that the spin-dependent functional has the same ground state energy as the spin-independent functional,
\begin{equation*}
	\widetilde{E}^{\sssup{MTF}} \myp{\beta} = E^{\sssup{MTF}} \myp{\beta},
\end{equation*}
and that they both also coincide with the ground state energy of a Vlasov type functional on the phase space $ \Omega = \mathbb{R}^3 \times \mathbb{R} \times {\mathbb{N}}_0 \times \Set{\pm 1} $, which will be introduced in \eqref{eq:vlafunct}.

\subsection{Main results}
The main results of this paper are the asymptotics of the ground state energy of the $ N $-body Hamiltonian \eqref{eq:hamiltonian} to leading order in $ N $, along with weak convergence of approximate ground states to convex combinations of factorized states.
\begin{theo}[Convergence of energy]
\label{thm:energylowbound}
	Let $ w \in L^{5/2} \myp{\mathbb{R}^3} + L_{\varepsilon}^{\infty} \myp{\mathbb{R}^3} $ be an even function, and $ V \in L_{\mathrm{loc}}^{5/2} \myp{\mathbb{R}^3} $ with $ V \myp{x} \to \infty $ as $ \abs{x} \to \infty $.
	Let $ \myp{\beta_N} $ be a sequence of positive real numbers satisfying $ \beta_N \to \beta \in \myb{0,\infty} $ and \eqref{eq:scaling2}.
	Then we have convergence of the ground state energy per particle
	\begin{equation}
		\lim_{N \to \infty} \frac{E \myp{N,\beta_N}}{N} =	\begin{cases}
																E^{\sssup{TF}},					& \text{if } \beta = 0, \\
																E^{\sssup{MTF}} \myp{\beta},	& \text{if } 0< \beta < \infty, \\
																E^{\sssup{STF}},				& \text{if } \beta = \infty.
															\end{cases}
	\end{equation}
\end{theo}
For the next theorem we recall that the $ k $-particle position density of a function $ \Psi \in \bigwedge^N L^2 \myp{\mathbb{R}^3 ; \mathbb{C}^2} \simeq
\bigwedge^N L^2 \myp{\mathbb{R}^3 \times \Set{\pm 1} ; \mathbb{C}} $ is given by
\begin{equation}
\label{eq:kpartdensityspin}
	\widetilde{\rho}_{\Psi}^{\, \myp{k}} \myp{z_1, \dotsc, z_k}
	= \binom{N}{k} \int_{\myp{\mathbb{R}^3 \times \Set{\pm 1}}^{\myp{N-k}}} \abs{\Psi \myp{z_1,\dotsc,z_N}}^2 \id z_{k+1} \cdots \id z_N.
\end{equation}
We will also need the spin-summed densities
\begin{equation}
\label{eq:kpartdensity}
	\rho_{\Psi}^{\myp{k}} \myp{x_1, \dotsc, x_k}
	= \binom{N}{k} \sum\limits_{s \in \Set{\pm 1}^N} \int_{\mathbb{R}^{3 \myp{N-k}}} \abs{\Psi \myp{x_1,\dotsc,x_N; s}}^2 \id x_{k+1} \cdots \id x_N.
\end{equation}
For the result on convergence of states below, we would like to call attention to the fact that the phase space changes in the extreme cases $ \beta = 0 $ and $ \beta = \infty $.
When $ \beta_N \to 0 $, the distance between the Landau bands (which is $ 2 \hbar b = 2 \beta_N \myp{1+\beta_N}^{-2/5} $) also tends to zero, in which case we recover the usual phase space $ \mathbb{R}^3 \times \mathbb{R}^3 \times \Set{\pm 1} $ with position, momentum, and spin variables, respectively.
In the other extreme case, where $ \beta_N \to \infty $, the magnetic field is so strong that all particles are confined to the lowest Landau band with spin pointing downwards, so the phase space here becomes $ \mathbb{R}^3 \times \mathbb{R} $, with the single copy of $ \mathbb{R} $ being momentum along the magnetic field.
\begin{theo}[Convergence of states]
	\label{thm:convstates}
	Suppose that the assumptions of \cref{thm:energylowbound} are satisfied.
	Let $ \Psi_N \in \bigwedge^N L^2 \myp{\mathbb{R}^3 ; \mathbb{C}^2} $ be a sequence of normalized approximate ground states, i.e. satisfying $ \innerp{ \Psi_N}{ H_{N,\beta_N} \Psi_N } = E \myp{N,\beta_N} + o \myp{N} $.
	Denote by $ \mathcal{M}_{\beta} $ the set of minimizers of the corresponding classical functional describing the ground state energy in the limit, that is,
	\begin{equation*}
		\mathcal{M}_{\beta}=	\begin{cases}
									\Set{ 0 \leq \rho \in L^1 \given \int \rho = 1, \ \mathcal{E}^{\sssup{TF}} \myp{\rho} = E^{ \sssup{TF}} }, 		& \text{if } \beta = 0, \\
									\Set{ 0 \leq \rho \in L^1 \given \int \rho = 1, \ \widetilde{\mathcal{E}}_{\beta}^{\sssup{MTF}} \myp{\rho} = \widetilde{E}^{\sssup{MTF}} \myp{\beta} },																& \text{if } 0< \beta < \infty, \\
									\Set{ 0 \leq \rho \in L^1 \given \int \rho = 1, \ \mathcal{E}^{\sssup{STF}} \myp{\rho} = E^{ \sssup{STF}} },	& \text{if } \beta = \infty,
								\end{cases}
	\end{equation*}
	where $ \rho \in L^1 $ means $ \rho \in L^1 \myp{\mathbb{R}^3} $ if	$ \beta = 0 $ or $ \beta = \infty $, and $ \rho \in L^1 \myp{\mathbb{R}^3 \times \Set{\pm 1}} $ if $ 0 < \beta < \infty $.
	
	Then there exist a subsequence $ \myp{N_\ell} \subseteq \mathbb{N} $ and a Borel probability measure $ \mathscr{P} $ on $ \mathcal{M}_{\beta} $ such that for $ \varphi \in L^{5/2} \myp{\mathbb{R}^3 \times \Set{\pm 1}} + L^{\infty} \myp{\mathbb{R}^3 \times \Set{\pm 1}} $ if $ k=1 $, and for any bounded and uniformly continuous function $ \varphi $ on $ \myp{\mathbb{R}^3 \times \Set{\pm 1}}^{k} $ if $ k \geq 2 $, we have as $ \ell $ tends to infinity,
	\begin{equation}
		\frac{k!}{N_\ell^k} \sum\limits_{s \in \Set{\pm 1}^k} \int_{\mathbb{R}^{3k}} \widetilde{\rho}_{\Psi_{N_\ell}}^{ \,(k)} \myp{x,s} \varphi \myp{x,s} \id x
		\longrightarrow \int_{\mathcal{M}_{\beta}} \myp[\Big]{ \int_{\mathbb{R}^{3k}} G_{\rho,\varphi}^{\beta,k} \myp{x} \id x} \id \mathscr{P} \myp{\rho}.
	\label{eq:convstates}
	\end{equation}
	The function $ G_{\rho,\varphi}^{\beta,k} $ is given by
	\begin{equation*}
		G_{\rho,\varphi}^{\beta,k} \myp{x}
		= \begin{cases}
			\sum\limits_{s \in \Set{\pm 1}^k} 2^{-k} \rho^{\otimes k} \myp{x} \varphi \myp{x,s}, 	& \text{if } \beta = 0, \\
			\sum\limits_{s \in \Set{\pm 1}^k} \rho^{\otimes k} \myp{x,s} \varphi \myp{x,s},			& \text{if } 0< \beta < \infty, \\
			\rho^{\otimes k} \myp{x} \varphi \myp{ x,\myp{-1}^{\times k} },							& \text{if } \beta = \infty,
		\end{cases}
	\end{equation*}
	where $ \myp{-1}^{\times k} $ denotes the $ k $-dimensional vector whose entries are all equal to~$ -1 $.
	Its presence is an expression of the fact that all the particles in this regime are confined to the lowest Landau band, with all spins pointing downwards.
\end{theo}
The convergence of energy and the convergence of states for $ k = 1 $ were both previously known in the case where the interaction $ w $ is Coulomb, and $ V \in L^{5/2} \myp{\mathbb{R}^3} + L^{\infty} \myp{\mathbb{R}^3} $ with $ V $ tending to zero at infinity \cite[Theorems~5.1\nobreakdash--5.3]{LieSolYng-94}.
The convergence of states result for $ k > 1 $ and the generality of the interaction $ w $ seem to be new.
\begin{rem}
	If the interaction $ w $ is of positive type, that is, if it has non-negative Fourier transform (which is the case e.g. for Coulomb interactions), then the limiting Thomas-Fermi functional is (strictly) convex in all three cases.
	This implies that any minimizer must be unique, so the de Finetti measures in \cref{thm:convstates} are forced to be supported on a single point.
	In other words, the outer integral in \labelcref{eq:convstates} disappears, and thus in this case the $ k $-particle densities converge weakly to pure tensor products of the unique Thomas-Fermi minimizers.
\end{rem}
\begin{rem}
	Analogues of \cref{thm:energylowbound,thm:convstates} also hold if the external potential $ V $ is not confining, but for brevity we will omit this generalization.
	In this case, the convergence of energy is the same as in \cref{thm:energylowbound}, but the statement for convergence of states is slightly different.
	Instead of being supported on the set of minimizers of the classical functional, the de Finetti measure in \cref{thm:convstates} will be supported on the set of weak limits of minimizing sequences for the functional.
	In this case, the lack of compactness at infinity forces one to use a weak version of the de Finetti theorem.
	See \cite{FouLewSol-18} for details in the case where there is no strong magnetic field.
\end{rem}
\subsubsection*{Organization of the paper}
In \cref{sec:prelim} we will recall a few results and preliminary observations that will be important for the later analysis.
\cref{sec:upperbound} is devoted to proving the upper energy bounds of \cref{thm:energylowbound} through construction of appropriate trial states.
In \cref{sec:semimeasures} we will construct semi-classical measures which in \cref{sec:lowerbound} will allow us to prove the lower energy bounds of \cref{thm:energylowbound} along with \cref{thm:convstates} in the case of strong magnetic fields, i.e. when $ \beta_N \to \beta \in \left( 0,\infty \right] $.
Finally, in \cref{sec:lowerbound2} we will treat the case $ \beta_N \to 0 $ where the magnetic field is negligible, so we can in this case use the semi-classical measures constructed in \cite{FouLewSol-18} on the usual phase space $ \mathbb{R}^3 \times \mathbb{R}^3 \times \Set{\pm 1} $.

%%%%%%%%%%%%%%%%%%%%%%%%%%%%%%%%%%%%%%%%%%%%%%%%%%%%%
%%%%%%%%%%%%%%%%%%%%%%%%%%%%%%%%%%%%%%%%%%%%%%%%%%%%%
\section{Preliminary observations}
\label{sec:prelim}
%%%%%%%%%%%%%%%%%%%%%%%%%%%%%%%%%%%%%%%%%%%%%%%%%%%%%
%%%%%%%%%%%%%%%%%%%%%%%%%%%%%%%%%%%%%%%%%%%%%%%%%%%%%

We start out by recalling a few results on Pauli operators and the magnetic Thomas-Fermi functional that will be important for our analysis.
\subsection{The semi-classical approximation and Lieb-Thirring bounds}
Here we briefly recall a few useful tools obtained in \cite{LieSolYng-94}.
We denote by $ V_- \myp{x} = \max\myp{-V \myp{x},0} $ the negative part of the potential $ V $.
Supposing that $ V_- \in L^{3/2} \myp{\mathbb{R}^3} \cap L^{5/2} \myp{\mathbb{R}^3} $ and denoting by $ e_j \myp{\hbar,b,V} $, $ j \geq 1 $, the negative eigenvalues for the operator
\begin{equation}
	H \myp{\hbar,b} = \myp{\bm{\sigma} \cdot \myp{-i \hbar \nabla + b A \myp{x}}}^2 + V\myp{x},
\end{equation}	
then we have the magnetic Lieb-Thirring inequality \cite[Theorem $ 2.1 $]{LieSolYng-94}
\begin{equation}
\label{eq:genlt}
	\sum\limits_{j=1}^{\infty} \abs{e_j \myp{\hbar,b,V}}
	\leq L_1 \frac{b}{\hbar^2} \int_{\mathbb{R}^3} V_-\myp{x}^{\frac{3}{2}} \id x + L_2 \frac{1}{\hbar^3} \int_{\mathbb{R}^3} V_-\myp{x}^{\frac{5}{2}} \id x,
\end{equation}
where for each $ 0 < \delta < 1 $ one can choose $ L_1 = \frac{4}{3} \myp{\pi \myp{1-\delta}}^{-1} $ and $ L_2 = 8 \sqrt{6} \myp{5 \pi \delta^2}^{-1} $.
The inequality can also be stated in terms of the $ 1 $-particle position density $ \rho_{\Psi}^{\myp{1}} $ of a many-body state $ \Psi \in \bigwedge^N L^2 \myp{\mathbb{R}^3; \mathbb{C}^2} $.
Namely, letting $ F_B $ denote the Legendre transform of the function $ v \mapsto L_1 B v^{3/2} + L_2 v^{5/2} $, that is,
\begin{equation}
\label{eq:fbdef}
	F_B \myp{t} = \sup_{v \geq 0} \myp{ tv - L_1 B v^{\frac{3}{2}} - L_2 v^{\frac{5}{2}} },
\end{equation}
then we have the lower bound \cite[Corollary $ 2.2 $]{LieSolYng-94}
\begin{equation}
\label{eq:genltcor}
	\innerp[\Big]{ \Psi}{ \sum\limits_{j=1}^N \myp{\bm{\sigma} \cdot \myp{-i \hbar \nabla_j + b A \myp{x_j}}}^2 \Psi }
	\geq \hbar^2 \int_{\mathbb{R}^3} F_{\frac{b}{\hbar}} \myp{ \rho_{\Psi}^{\myp{1}} \myp{x} } \id x
\end{equation}
on the kinetic energy of the state $ \Psi $.
We also have Weyl asymptotics for the Pauli operator \cite[Theorem $ 3.1 $]{LieSolYng-94}
\begin{equation}
\label{eq:pauliweyl}
	\lim_{\hbar \to 0} \frac{\sum\nolimits_j e_j \myp{\hbar,b,V}}{E_{\mathrm{scl}} \myp{\hbar,b,V}} = 1,
\end{equation}
uniformly in the magnetic field strength $ b $, where $ E_{\mathrm{scl}} \myp{\hbar,b,V} $ is the semi-classical expression for the sum of negative eigenvalues
\begin{equation}
\label{eq:paulisemi-classics}
	E_{\mathrm{scl}} \myp{\hbar,b,V} = - \frac{1}{\hbar^3} \int_{\mathbb{R}^3} P_{\hbar b} \myp{V_- \myp{x}} \id x,
\end{equation}
with $ P_{\hbar b} $ given in \eqref{eq:landaupressure}.
In our case, with the scaling relations \eqref{eq:scaling}, the Weyl asymptotics take the following form:
\begin{cor}
\label{cor:pauliweil}
	Suppose $ V_- \in L^{3/2} \myp{\mathbb{R}^3} \cap L^{5/2} \myp{\mathbb{R}^3} $, let $ \myp{\beta_N} $ be a sequence of positive real numbers satisfying $ \beta_N \to \beta \in \myb{0,\infty} $ and \eqref{eq:scaling2}, and define $ \hbar $ and $ b $ by \eqref{eq:scaling0}.
	Then the Weyl asymptotics \eqref{eq:pauliweyl} take the form
	\begin{equation}
		\lim_{N \to \infty} \frac{1}{N} \sum\limits_j e_j \myp{\hbar,b,V} 
		=	\begin{cases}
				- \frac{2}{15 \pi^2} \int V_- \myp{x}^{\frac{5}{2}} \id x,					& \text{if } \beta = 0, \\
				- \myp{1+\beta}^{-\frac{3}{5}} \int P_{k_{\beta}} \myp{V_- \myp{x}} \id x,	& \text{if } 0< \beta < \infty, \\
				- \frac{1}{3 \pi^2} \int V_- \myp{x}^{\frac{3}{2}} \id x,					& \text{if } \beta = \infty.
			\end{cases}
	\end{equation}
\end{cor}
The details of the proof, which mainly consists of applying the dominated convergence theorem to \eqref{eq:pauliweyl}, will be omitted.
\begin{rem}
\label{rem:pauliweyl}
	The Weyl asymptotics in \eqref{eq:pauliweyl} and \cref{cor:pauliweil} also hold true if the Pauli operator $ \myp{\bm{\sigma} \cdot \myp{-i \hbar \nabla + b A \myp{x}}}^2 $ is replaced by the Pauli operator in a cube $ C_R = \myp[\big]{-\frac{R}{2}, \frac{R}{2}}^3 $ with Dirichlet boundary conditions, denoted by $ \myp{\bm{\sigma} \cdot \myp{-i \hbar \nabla + b A \myp{x}}}^2_{C_R} $, (and $ V $ is replaced by a potential defined on $ C_R $).
	For further details on this, see e.g. \cite{MadThese}.
\end{rem}
Applying the generalized Lieb-Thirring inequality \cref{eq:genlt,eq:genltcor} yields the following important estimates.
The proof is a step-by-step imitation of the proof of \cite[Lemma $ 3.4 $]{FouLewSol-18}.
\begin{lem}
\label{lem:ltbounds}
	If $ V_-, w_- \in L^{5/2} \myp{\mathbb{R}^3} + L^{\infty} \myp{\mathbb{R}^3} $ and $ \beta_N > 0 $, then
	\begin{equation}
	\label{eq:energybound}
		H_{N,\beta_N} \geq \sum\limits_{j=1}^N \myp[\Big]{ \frac{1}{2} \myp{ \bm{\sigma} \cdot \myp{-i\hbar \nabla_j + bA\myp{x_j}} }^2 + V_+ \myp{x_j} } - CN \myp[\Big]{\frac{b+1}{\hbar^2 N} +1 }
	\end{equation}
	and for any normalized fermionic wave function $ \Psi \in \bigwedge^N L^2 \myp{\mathbb{R}^3 ; \mathbb{C}^2} $,
	\begin{align}
	\label{eq:kinbound}
		\MoveEqLeft[3] \innerp[\Big]{ \Psi}{ \sum\limits_{j=1}^N \myp[\big]{ \myp{ \bm{\sigma} \cdot \myp{ -i \hbar \nabla_j + bA \myp{x_j}} }^2 + V_+ \myp{x_j} } \Psi }
		+ \hbar^2 \int F_{\frac{b}{\hbar}} \myp{ \rho_{\Psi}^{\myp{1}} \myp{x} } \id x \nonumber \\
		&\leq 2 \innerp{\Psi}{ H_{N, \beta_N} \Psi }+ CN \myp[\Big]{\frac{b+1}{\hbar^2 N} +1},
	\end{align}
	with $ F_{b/\hbar} $ given by \eqref{eq:fbdef}.
	Furthermore, if $ \Psi_N \in \bigwedge^N L^2 \myp{\mathbb{R}^3; \mathbb{C}^2} $ is a sequence satisfying $ \innerp{ \Psi_N}{ H_{N,\beta_N} \Psi_N } \leq CN $, then for any $ f = f_1 + f_2 \in L^{3/2} \myp{\mathbb{R}^3} \cap L^{5/2} \myp{\mathbb{R}^3} + L^{\infty} \myp{\mathbb{R}^3} $, we have that
	\begin{align}
		\MoveEqLeft[3]	\frac{1}{N} \int_{\mathbb{R}^3} f \myp{x} \rho_{\Psi_N}^{\myp{1}} \myp{x} \id x + \frac{1}{N^2} \iint_{\mathbb{R}^6} f \myp{x-y} \rho_{\Psi_N}^{\myp{2}} \myp{x,y} \id x \id y \nonumber \\
		&\leq \tilde{C} \myp[\Big]{\frac{b+1}{\hbar^2 N} + \frac{1}{\hbar^3 N} + 1} \myp{ \normt{f_1}{\frac{3}{2}} + \normt{f_1}{\frac{5}{2}} + \normt{f_2}{\infty}}.
	\label{eq:ltdensity}
	\end{align}
\end{lem}
\begin{proof}
	The argument goes along the same lines as the proof of \cite[Lemma $ 3.4 $]{FouLewSol-18}.
	We write $ V_- = V_1 + V_2 $ and $ w_- = w_1+w_2 $ with $ V_1,w_1 \in L^{5/2} \myp{\mathbb{R}^3} \cap L^{3/2} \myp{\mathbb{R}^3} $ and $ V_2,w_2 \in L^{\infty} \myp{\mathbb{R}^3} $.
	We clearly have that
	\begin{equation}
	\label{eq:vneg2}
		\innerp[\Big]{ \Psi}{ \sum\limits_{j=1}^N -V_2 \myp{x_j} \Psi }
		\geq - \normt{V_2}{\infty} N
	\end{equation}
	for any normalized wave function $ \Psi $.
	Briefly denoting by $ H_{V_1} $ the operator $ H_{V_1} := \frac{1}{4} \myp{ \bm{\sigma} \cdot \myp{-i \hbar \nabla +b A}}^2 - V_1 $ and applying the magnetic Lieb-Thirring inequality \eqref{eq:genlt}, we obtain
	\begin{align}
		\innerp[\Big]{ \Psi}{ \sum\limits_{j=1}^N H_{V_1} \Psi }
		&\geq \text{Tr} \myp{H_{V_1}}_- \nonumber \\
		&\geq -2 L_1 \frac{b}{\hbar^2} \int_{\mathbb{R}^3} V_1 \myp{x}^{\frac{3}{2}} \id x -8 L_2 \frac{1}{\hbar^3} \int_{\mathbb{R}^3} V_1 \myp{x}^{\frac{5}{2}} \id x.
	\label{eq:vneg1}
	\end{align}
	Note that by symmetry we have
	\begin{align*}
		\innerp[\Big]{\Psi}{ \frac{1}{N} \sum\limits_{1 \leq k < \ell \leq N} w_- \myp{x_k-x_{\ell}} \Psi }
		&= \frac{N-1}{2} \innerp[\Big]{ \Psi}{ w_- \myp{x_1-x_2} \Psi } \nonumber \\
		&= \frac{1}{2} \innerp[\Big]{ \Psi}{ \sum\limits_{j=2}^N w_- \myp{x_1-x_j} \Psi },
	\end{align*}
	so applying what we have just shown to the last $ N-1 $ variables, we get
	\begin{align*}
		\MoveEqLeft[3]	\innerp[\Big]{ \Psi}{ \myp[\bigg]{ \sum\limits_{j=1}^N \frac{1}{4} \myp{ \bm{\sigma} \cdot \myp{-i\hbar \nabla_j + bA\myp{x_j}} }^2 - \frac{1}{2} \sum\limits_{j=2}^N w_1 \myp{x_1 - x_j} } \Psi } \\
		&\geq \innerp[\Big]{ \Psi}{ \myp[\bigg]{ \sum\limits_{j=2}^N \frac{1}{4} \myp{ \bm{\sigma} \cdot \myp{-i\hbar \nabla_j + bA\myp{x_j}} }^2 - \frac{1}{2} w_1 \myp{x_1 - x_j} } \Psi } \\
		&\geq -C_1 \frac{b}{\hbar^2} \int_{\mathbb{R}^3} w_1 \myp{x}^{\frac{3}{2}} \id x -C_2 \frac{1}{\hbar^3} \int_{\mathbb{R}^3} w_1 \myp{x}^{\frac{5}{2}} \id x.
	\end{align*}
	Hence we see that
	\begin{align}
		\MoveEqLeft[3]	\sum\limits_{j=1}^N \frac{1}{4} \myp{ \bm{\sigma} \cdot \myp{-i\hbar \nabla_j + bA\myp{x_j}} }^2 + \frac{1}{N} \sum\limits_{1 \leq k < \ell \leq N} w \myp{x_k-x_{\ell}} \nonumber \\
		&\geq -C_1 \frac{b}{\hbar^2} \int_{\mathbb{R}^3} w_1 \myp{x}^{\frac{3}{2}} \id x -C_2 \frac{1}{\hbar^3} \int_{\mathbb{R}^3} w_1 \myp{x}^{\frac{5}{2}} \id x - \frac{N-1}{2} \normt{w_2}{\infty}.
	\label{eq:wbound}
	\end{align}
	Combining \labelcref{eq:vneg2,eq:vneg1,eq:wbound} yields \eqref{eq:energybound}.
	We obtain \eqref{eq:kinbound} directly from \eqref{eq:energybound} by applying the Lieb-Thirring inequality \eqref{eq:genltcor}.
	
	Let us turn our attention towards the proof of \eqref{eq:ltdensity}.
	Note that it suffices to prove the estimate for non-negative functions $ f $.
	We will prove the one-body part of the estimate first.
	Clearly, since $ \Psi_N $ is normalized,
	\begin{equation}
	\label{eq:onebodyest2}
		\frac{1}{N} \int_{\mathbb{R}^3} f_2 \myp{x} \rho_{\Psi_N}^{\myp{1}} \myp{x} \id x
		\leq \normt{f_2}{\infty},
	\end{equation}
	so we may consider only $ f_1 \in L^{3/2} \myp{\mathbb{R}^3} \cap L^{5/2} \myp{\mathbb{R}^3} $.
	For any $ v \geq 0 $ we have by definition,
	\begin{equation*}
		\rho_{\Psi_N}^{\myp{1}} \myp{x} v - L_1 \frac{b}{\hbar} v^{\frac{3}{2}} - L_2 v^{\frac{5}{2}} 
		\leq F_{\frac{b}{\hbar}} \myp{ \rho_{\Psi_N}^{\myp{1}} \myp{x} },
	\end{equation*}
	so we replace $ v $ by $ \frac{1}{\varepsilon} f_1 \myp{x} $, integrate and apply \eqref{eq:kinbound} to obtain
	\begin{align*}
		\MoveEqLeft[3]	\int_{\mathbb{R}^3} f_1 \myp{x} \rho_{\Psi_N}^{\myp{1}} \myp{x} - \frac{L_1}{\varepsilon^{\frac{1}{2}}} \frac{b}{\hbar} f_1 \myp{x}^{\frac{3}{2}} - \frac{L_2}{\varepsilon^{\frac{3}{2}}} f_1 \myp{x}^{\frac{5}{2}} \id x \\
		&\leq \varepsilon \int_{\mathbb{R}^3} F_{\frac{b}{\hbar}} \myp{ \rho_{\Psi_N}^{\myp{1}} \myp{x} } \id x \leq \frac{CN}{\hbar^2} \myp[\Big]{\frac{b+1}{\hbar^2 N} +3} \varepsilon,
	\end{align*}
	implying the bound
	\begin{equation*}
		\int_{\mathbb{R}^3} f_1 \myp{x} \rho_{\Psi_N}^{\myp{1}} \myp{x} \id x
		\leq \frac{CN}{\hbar^2} \myp[\Big]{\frac{b+1}{\hbar^2 N} +3} \varepsilon + L_1 \frac{b}{\hbar \varepsilon^{\frac{1}{2}}} \normt{f_1}{\frac{3}{2}}^{\frac{3}{2}} + L_2 \frac{1}{\varepsilon^{\frac{3}{2}}} \normt{f_1}{\frac{5}{2}}^{\frac{5}{2}}.
	\end{equation*}
	With the choice $ \varepsilon = \myp{ \normt{f_1}{\frac{3}{2}} + \normt{f_1}{\frac{5}{2}} } \hbar^2 $, we get
	\begin{align}
		\frac{1}{N} \int_{\mathbb{R}^3} f_1 \myp{x} \rho_{\Psi_N}^{\myp{1}} \myp{x} \id x 
		\leq \widetilde{C} \myp[\Big]{\frac{b+1}{\hbar^2 N} + \frac{1}{\hbar^3 N} + 1} \myp[\big]{ \normt{f_1}{\frac{3}{2}} + \normt{f_1}{\frac{5}{2}} }
	\label{eq:onebodyest}
	\end{align}
	for some constant $ \widetilde{C} > 0 $, showing the one-body part of \eqref{eq:ltdensity}.
	To obtain the two-body estimate, we apply \eqref{eq:wbound} with $ w $ replaced by $ \frac{1}{\varepsilon} f_1 $ and use the anti-symmetry of $ \Psi_N $ to get
	\begin{align*}
		\MoveEqLeft[3]	-C_1 \frac{b}{\hbar^2 \varepsilon^{\frac{3}{2}}} \int_{\mathbb{R}^3}  f_1 \myp{x}^{\frac{3}{2}} \id x -C_2 \frac{1}{\hbar^3 \varepsilon^{\frac{5}{2}}} \int_{\mathbb{R}^3} f_1 \myp{x}^{\frac{5}{2}} \id x \\
		&\leq \innerp[\Big]{ \Psi_N }{ \myp[\bigg]{ \sum\limits_{j=1}^N \frac{1}{4} \myp{ \bm{\sigma} \cdot \myp{-i\hbar \nabla_j + bA\myp{x_j}} }^2 - \frac{1}{2} \sum\limits_{j=2}^N \frac{1}{\varepsilon} f_1 \myp{x_1-x_j} } \Psi_N } \\
		&\leq CN \myp[\Big]{\frac{b+1}{\hbar^2 N} +3} - \frac{1}{\varepsilon N} \iint_{\mathbb{R}^6} f_1 \myp{x-y} \rho_{\Psi_N}^{\myp{2}} \myp{x,y} \id x \id y,
	\end{align*}
	where the second inequality holds by \eqref{eq:kinbound}.
	Now we simply take $ \varepsilon = \normt{f_1}{\frac{3}{2}} + \normt{f_1}{\frac{5}{2}} $ and rearrange to obtain
	\begin{equation*}
		\frac{1}{N^2} \iint_{\mathbb{R}^6} f_1 \myp{x-y} \rho_{\Psi_N}^{\myp{2}} \myp{x,y} \id x \id y
		\leq \widetilde{C} \myp[\Big]{\frac{b+1}{\hbar^2 N} + \frac{1}{\hbar^3 N} + 1} \myp[\big]{ \normt{f_1}{\frac{3}{2}} + \normt{f_1}{\frac{5}{2}} }.
	\end{equation*}
	Combining this with the fact that
	\begin{equation*}
		\frac{1}{N^2} \iint_{\mathbb{R}^6} f_2 \myp{x-y} \rho_{\Psi_N}^{\myp{2}} \myp{x,y} \id x \id y
		\leq \normt{f_2}{\infty} \frac{1}{N^2} \binom{N}{2} \normt{\Psi_N}{2}^2
		\leq \frac{1}{2} \normt{f_2}{\infty},
	\end{equation*}
	we get the two-body estimate in \eqref{eq:ltdensity}, finishing the proof.
\end{proof}
\begin{cor}[to LT inequality]
	If $ \Psi \in \bigwedge^N L^2 \myp{\mathbb{R}^3; \mathbb{C}^2} $ is an $ N $-particle state with finite kinetic energy, then $ \rho_{\Psi}^{\myp{1}} \in L^{5/3} \myp{\mathbb{R}^3} $.
\end{cor}
\begin{proof}
	Note first that for any $ M> 0 $, we have by Markov's inequality,
	\begin{align*}
		\int \myp{ \rho_{\Psi}^{\myp{1}} \myp{x} \mathds{1}_{ \Set{ \rho_{\Psi}^{\myp{1}} \leq M } } \myp{x} }^{\frac{5}{3}} \id x
		&\leq \int_{\Set{\rho_{\Psi}^{\myp{1}} \leq 1 }} \myp{ \rho_{\Psi}^{\myp{1}} \myp{x} }^{\frac{5}{3}} \id x + M^{\frac{5}{3}} \abs{ \Set{ \rho_{\Psi}^{\myp{1}} \geq 1 } } \\
		&\leq \myp[\big]{1+M^{\frac{5}{3}} } \int_{\mathbb{R}^3} \rho_{\Psi}^{\myp{1}} \myp{x} \id x,
	\end{align*}
	so $ \rho_{\Psi}^{\myp{1}} \mathds{1}_{ \Set{ \rho_{\Psi}^{\myp{1}} \leq M
	} } \in L^{5/3} \myp{\mathbb{R}^3} $.
	Since $ \Psi $ has finite kinetic energy, it follows from \eqref{eq:genltcor} that
	\begin{equation*}
		\int_{\mathbb{R}^3} \rho_{\Psi}^{\myp{1}} \myp{x} f \myp{x} - L_1 \frac{b}{\hbar} f \myp{x} ^\frac{3}{2} - L_2 f \myp{x}^{\frac{5}{2}} \id x \leq C
	\end{equation*}
	for any $ 0 \leq f \in L^{3/2} \myp{\mathbb{R}^3} \cap L^{5/2} \myp{\mathbb{R}^3} $, where $ C $ is independent of $ f $, and $ L_2 $ can be chosen such that $ 0 < L_2 < 1 $ by the Lieb-Thirring inequality \eqref{eq:genlt}.
	Choosing $ f = \myp{ \rho_{\Psi}^{\myp{1}} \mathds{1}_{ \Set{ \rho_{\Psi}^{\myp{1}} \leq M } } }^{2/3} $, we get
	\begin{equation*}
		\myp{1-L_2} \int_{\Set{\rho_{\Psi}^{\myp{1}} \leq M } } \rho_{\Psi}^{\myp{1}} \myp{x}^{\frac{5}{3}} \id x
		\leq C + L_1 \frac{b}{\hbar} \int_{\Set{\rho_{\Psi}^{\myp{1}} \leq M } } \rho_{\Psi}^{\myp{1}} \myp{x} \id x.
	\end{equation*}
	Taking $ M $ to infinity finishes the proof.
\end{proof}
\begin{lem}
\label{lem:densitybound}
	Suppose that $ \Psi_N \in \bigwedge^N L^2 \myp{\mathbb{R}^3; \mathbb{C}^2} $ is a sequence satisfying the kinetic energy bound
	\begin{equation*}
		\innerp[\Big]{ \Psi_N}{ \bigg( \sum\limits_{j=1}^N \myp{ \bm{\sigma} \cdot \myp{ -i \hbar \nabla_j + bA \myp{x_j}} }^2 \bigg) \Psi_N } 
		\leq \widetilde{C} N,
	\end{equation*}
	where $ \hbar $ and $ b $ satisfy the scaling relations	\eqref{eq:scaling0}.
	Then there exists a $ C > 0 $ such that $ \normt{ \rho_{\Psi_N}^{\myp{1}} }{\frac{5}{3}} \leq C N $ for all $ N $.
	In particular, if $ \frac{1}{N}\rho_{\Psi_N}^{\myp{1}} \rightharpoonup \rho $ weakly as functionals on $ C_c \myp{\mathbb{R}^3} $, then $ \rho \in L^1 \myp{\mathbb{R}^3} \cap L^{5/3} \myp{\mathbb{R}^3} $ and for any test function $ \varphi \in L^{5/2} \myp{\mathbb{R}^3} + L_{\varepsilon}^{\infty} \myp{\mathbb{R}^3} $ we have $ \int \frac{1}{N}\rho_{\Psi_N}^{\myp{1}} \varphi \to \int \rho \varphi $.
\end{lem}
\begin{proof}
	For the duration of the proof we will denote $ \rho_N = \rho_{\Psi_N}^{\myp{1}} $.
	For any $ 0 \leq f \in L^{3/2} \myp{\mathbb{R}^3} \cap L^{5/2} \myp{\mathbb{R}^3} $ we have by \eqref{eq:genltcor} that
	\begin{equation*}
		\hbar^2 \int_{\mathbb{R}^3} \rho_N \myp{x} f \myp{x} - L_1 \frac{b}{\hbar} f \myp{x} ^\frac{3}{2} - L_2 f \myp{x}^{\frac{5}{2}} \id x
		\leq \widetilde{C} N.
	\end{equation*}
	Hence, noting by \eqref{eq:genlt} that we can take $ L_2 < 1 $, and choosing $ f = \varepsilon^2 \rho_N^{2/3} $ with $ 0 < \varepsilon < 1 $, we obtain
	\begin{equation*}
		\varepsilon^2 \myp{1-L_2 \varepsilon^3} \int_{\mathbb{R}^3} \rho_N \myp{x}^{\frac{5}{3}} \id x
		\leq \widetilde{C} \frac{N}{\hbar^2} + \varepsilon^3 L_1 \frac{b}{\hbar} \int_{\mathbb{R}^3} \rho_N \myp{x} \id x.
	\end{equation*}
	Inserting the definitions of $ \hbar $ and $ b $ \eqref{eq:scaling0} yields
	\begin{equation*}
		\normt{\rho_N}{\frac{5}{3}}^{\frac{5}{3}}
		\leq \frac{\widetilde{C} N}{\varepsilon^2 \myp{1-L_2}} \myp[\Big]{\frac{1}{\hbar^2} + \varepsilon^3 \frac{b}{\hbar}}
		= C' N^{\frac{5}{3}} \myp[\bigg]{ \frac{1}{ \varepsilon^2 \myp{1+\beta_N}^{\frac{2}{5}} } + \frac{\varepsilon \beta_N}{\myp{1+\beta_N}^{\frac{4}{5}} } },
	\end{equation*}
	so simply choosing $ \varepsilon = \myp{1 + \beta_N}^{-1/5} $ gives the desired bound.
	
	The last part of the lemma follows easily from standard methods in functional analysis, and the details will be omitted.
\end{proof}

\subsection{Energy functionals on phase space}
Instead of working with a functional on position densities, it will in some situations be much more convenient to use a functional defined on densities on phase space.
Hence we introduce the Vlasov energy functional, and note its connection to the magnetic Thomas-Fermi functionals.
\begin{de}[Magnetic Vlasov functional]
	We put $ \Omega = \mathbb{R}^3 \times \mathbb{R} \times \mathbb{N}_0 \times \Set{\pm 1} $, and for $ \beta > 0 $,
	\begin{equation*}
		\mathcal{D}^{\sssup{Vla}} = \Set{m \in L^1 \myp{\Omega} \given 0 \leq m \leq 1, \ V\rho_m, \myp{w \ast \rho_m} \rho_m \in L^1 \myp{\mathbb{R}^3}},
	\end{equation*}
	where
	\begin{equation*}
		\rho_m \myp{x} = \frac{1}{\myp{2 \pi}^2} \frac{\beta}{1+ \beta} \sum\limits_{s= \pm 1} \sum\limits_{j=0}^{\infty} \int_{\mathbb{R}} m \myp{x,p,j,s} \id p.
	\end{equation*}
	Letting $ k_{\beta} = \beta \myp{1+\beta}^{-2/5} $ as in \eqref{eq:kbeta}, we define a functional
	\begin{align}
		\mathcal{E}_{\beta}^{\sssup{Vla}} \myp{m}
		= {} & \frac{1}{\myp{2\pi}^2} \frac{\beta}{1+\beta} \sum\limits_{\substack{j\geq 0 \\ s = \pm 1}} \int_{\mathbb{R}} \int_{\mathbb{R}^3} \myp{p^2 + k_{\beta} \myp{2j+1+s }} m \myp{x,p,j,s} \id x \id p \nonumber \\
		&+ \int_{\mathbb{R}^3} V \myp{x} \rho_m \myp{x} \id x + \frac{1}{2} \iint_{\mathbb{R}^6} w \myp{x-y} \rho_m \myp{x} \rho_m \myp{y} \id x \id y.
	\label{eq:vlafunct}
	\end{align}
	Furthermore, we define the Vlasov ground state energy
	\begin{equation*}
		E^{\sssup{Vla}} \myp{\beta} = \inf \Set[\big]{\mathcal{E}_{\beta}^{\sssup{Vla}} \myp{m} \given m \in \mathcal{D}^{\sssup{Vla}}, \frac{1}{\myp{2 \pi}^2} \frac{\beta}{1+\beta} \int_{\Omega} m \myp{\xi} \id \xi = 1 }.
	\end{equation*}
\end{de}
The factors $ \beta \myp{1+\beta}^{-1} $ turn up naturally from the pressure of the free Landau gas \eqref{eq:landaupressure} equipped with the mean-field scaling \eqref{eq:kbeta}, also recalling the additional prefactor $ \myp{1+\beta}^{-3/5} $ in front of the pressure.
\begin{lem}
	\label{lem:functionalminima}
	Suppose that $ \rho \in \mathcal{D}^{\sssup{MTF}} $ and define a measure on $ \Omega = \mathbb{R}^3 \times \mathbb{R} \times \mathbb{N}_0 \times \Set{\pm 1} $ with density
	\begin{equation}
	\label{eq:mrhodef}
		m_{\rho} \myp{x,p,j,s} = \mathds{1}_{\Set{p^2 + k_{\beta} \myp{2j+1+s} \leq r \myp{x} } },
	\end{equation}
	where for (almost) each $ x \in \mathbb{R}^3 $, $ r \myp{x} $ is the unique solution to the equation
	\begin{equation}
	\label{eq:sxdef}
		\rho \myp{x} = \frac{1}{2 \pi^2} \frac{\beta}{1+\beta} \myp[\Big]{ {r \myp{x}}^{\frac{1}{2}} + 2\sum\limits_{j=1}^{\infty} \myb{2 k_{\beta} j -r \myp{x}}_-^{\frac{1}{2}} }
		= \myp{1+\beta}^{-\frac{3}{5}} P_{k_{\beta}}' \myp{r \myp{x}}.
	\end{equation}
	Then $ m_{\rho} \in \mathcal{D}^{\sssup{Vla}} $ and satisfies for almost all $ x \in \mathbb{R}^3 $
	\begin{equation*}
		\rho_{m_{\rho}} \myp{x}
		:= \frac{1}{\myp{2 \pi}^2} \frac{\beta}{1+\beta} \sum\limits_{s = \pm 1} \sum\limits_{j=0}^{\infty} \int_{\mathbb{R}} m_{\rho} \myp{x,p,j,s} \id p
		= \rho \myp{x},
	\end{equation*}
	\begin{equation}
	\label{eq:vlamtf}
		\frac{1}{\myp{2 \pi}^2} \frac{\beta}{1+\beta} \sum\limits_{s = \pm 1} \sum\limits_{j=0}^{\infty} \int_{\mathbb{R}} \myp{p^2 + k_{\beta} \myp{2j+ 1+s}} m_{\rho} \myp{x,p,j,s} \id p
		= \tau_{\beta} \myp{\rho \myp{x}},
	\end{equation}
	and $ \mathcal{E}_{\beta}^{\sssup{MTF}} \myp{\rho} = \mathcal{E}_{\beta}^{\sssup{Vla}} \myp{m_{\rho}} $.
	
	On the other hand, if $ m \in \mathcal{D}^{\sssup{Vla}} $, then	$ \rho_m \in \mathcal{D}^{\sssup{MTF}} $ and $ \mathcal{E}_{\beta}^{\sssup{MTF}} \myp{\rho_m} \leq \mathcal{E}_{\beta}^{\sssup{Vla}} \myp{m} $.
	In particular,
	\begin{equation*}
		E^{\sssup{MTF}} \myp{\beta} = E^{\sssup{Vla}} \myp{\beta}.
	\end{equation*}
\end{lem}
\begin{rem}
	The assertions of the lemma also hold true when the magnetic Thomas-Fermi functional is replaced by the spin dependent functional $ \widetilde{\mathcal{E}}_{\beta}^{\sssup{MTF}} $.
	In particular,
	\begin{equation*}
		\widetilde{E}^{\sssup{MTF}} \myp{\beta} 
		= E^{\sssup{Vla}} \myp{\beta}
		= E^{\sssup{MTF}} \myp{\beta}.
	\end{equation*}
	The proof is exactly the same as in the spin-summed case, except that in this case, the equation
	\begin{equation*}
		\widetilde{\rho} \myp{x,s} = \myp{1+\beta}^{-\frac{3}{5}} \widetilde{P}_{k_{\beta}}' \myp{\widetilde{r} \myp{x,s},s}
	\end{equation*}
	does not uniquely define $ \widetilde{r} $ everywhere, because $ \widetilde{P}_{k_{\beta}}' \myp{\nu,1} = 0 $ for $ 0 \leq \nu \leq 2 k_{\beta} $.
	However, this can easily be remedied by instead defining
	\begin{equation*}
		\widetilde{r} \myp{x,s} = \max \big\{ r \geq 0 \mid \widetilde{\tau}_{\beta} \myp{\rho \myp{x,s},s} 
		= \rho \myp{x,s} r - \myp{1+\beta}^{-\frac{3}{5}} \widetilde{P}_{k_{\beta}} \myp{r ,s} \big\},
	\end{equation*}
	but we omit the details.
\end{rem}
\begin{proof}
	The idea is to fix a position density and minimize the Vlasov problem for each fixed position $ x \in \mathbb{R}^3 $.
	For any $ \nu \geq 0 $, we calculate the measure of the set
	\begin{align}
		\MoveEqLeft[3]	\abs{ \Set{ \myp{p,j,s} \in \mathbb{R} \times \mathbb{N}_0 \times \Set{\pm 1} \given p^2 + k_{\beta} \myp{2j+1+s } \leq \nu } } \nonumber \\
		&= \sum\limits_{j=0}^{\infty} \abs{ \Set{p^2 \leq \nu-2 k_{\beta} \myp{j+1}} } + \sum\limits_{j=0}^{\infty} \abs{ \Set{p^2 \leq \nu-2 k_{\beta} j} } \nonumber \\
		&= 2 \nu^{\frac{1}{2}} + 4 \sum\limits_{j=1}^{\infty} \myb{2 k_{\beta} j -\nu}_-^{\frac{1}{2}}
		= \frac{\myp{2 \pi}^2}{k_{\beta}} P_{k_{\beta}}' \myp{\nu}.
	\label{eq:pressure}
	\end{align}
	Supposing that $ \rho \in \mathcal{D}^{\sssup{MTF}} $, then for each $ x \in \mathbb{R}^3 $ we may choose $ r \myp{x} \geq 0 $ to be the unique solution of \eqref{eq:sxdef} and define $ m_{\rho} $ as in \eqref{eq:mrhodef}.
	The calculation above then clearly shows
	\begin{align*}
		\MoveEqLeft[3]	\frac{1}{\myp{2 \pi}^2} \frac{\beta}{1+\beta} \sum\limits_{s = \pm 1} \sum\limits_{j=0}^{\infty} \int_{\mathbb{R}} m_{\rho} \myp{x,p,j,s} \id p \\
		&= \frac{1}{\myp{2 \pi}^2} \frac{\beta}{1+\beta} \sum\limits_{s = \pm 1} \sum\limits_{j=0}^{\infty} \int_{\mathbb{R}} \mathds{1}_{\Set{p^2 + k_{\beta} \myp{2j + 1+s } \leq r \myp{x} }} \id p 
		= \rho \myp{x}.
	\end{align*}
	To see that \eqref{eq:vlamtf} holds, note that the supremum in \eqref{eq:taubdef} is attained exactly at the point $ r \myp{x} \geq 0 $, that is,
	\begin{equation}
	\label{eq:taueq}
		\tau_{\beta} \myp{\rho \myp{x}} 
		= \rho \myp{x} r \myp{x} - \myp{1+\beta}^{-\frac{3}{5}} P_{k_{\beta}} \myp{r \myp{x}}.
	\end{equation}
	Furthermore, using $ \myb{2 k_{\beta} j - r \myp{x}}_-^{3/2} = \myp{r \myp{x} - 2 k_{\beta} j} \myb{2 k_{\beta} j - r \myp{x}}_-^{1/2} $ along with the definition of the pressure $ P_{k_{\beta}} $, we also have
	\begin{equation*}
		2 \sum\limits_{j=1}^{\infty} 2 k_{\beta} j \myb{2 k_{\beta} j - r \myp{x}}_-^{\frac{1}{2}}
		= \frac{2 \pi^2}{k_{\beta}} P_{k_{\beta}}' \myp{r \myp{x}} r \myp{x} - \frac{3 \pi^2}{k_{\beta}} P_{k_{\beta}} \myp{r \myp{x}}.
	\end{equation*}
	With this in mind, we calculate, using the definition of $ m_{\rho} $,
	\begin{align}
		\MoveEqLeft[3]	\sum\limits_{s = \pm 1} \sum\limits_{j=0}^{\infty} \int_{\mathbb{R}} \myp{p^2 + k_{\beta} \myp{2j+ 1+s }} m_{\rho} \myp{x,p,j,s} \id p \nonumber \\
		&= \int_{\mathbb{R}} p^2 \mathds{1}_{ \Set{p^2 \leq r \myp{x} }} \id p + 2 \sum\limits_{j=1}^{\infty} \int_{\mathbb{R}} \myp{p^2 + 2 k_{\beta} j} \mathds{1}_{ \Set{p^2 + 2 k_{\beta} j \leq r \myp{x} } } \id p \nonumber \\
		&= \frac{2 \pi^2}{k_{\beta}} P_{k_{\beta}} \myp{r \myp{x}} + 4 \sum\limits_{j=1}^{\infty} 2 k_{\beta} j \myb{2 k_{\beta} j - r \myp{x}}_-^{\frac{1}{2}} \nonumber \\
		&= \frac{\myp{2 \pi}^2}{k_{\beta}} \myp{ P_{k_{\beta}}' \myp{r \myp{x}} r \myp{x} - P_{k_{\beta}} \myp{r \myp{x}}}
		\label{eq:pressure2} \\
		&= \myp{2 \pi}^2 \frac{1+\beta}{\beta} \tau_{\beta} \myp{\rho{\myp{x}}} \nonumber,
	\end{align}
	showing \eqref{eq:vlamtf}.
	This implies that $ m_{\rho} \in \mathcal{D}^{\sssup{Vla}} $ and $ \mathcal{E}_{\beta}^{\sssup{MTF}} \myp{\rho} = \mathcal{E}_{\beta}^{\sssup{Vla}} \myp{m_{\rho}} $, and hence	$ E^{\sssup{MTF}} \myp{\beta} \geq E^{\sssup{Vla}} \myp{\beta} $.
	
	On the other hand, for any $ m \in \mathcal{D}^{\sssup{Vla}} $ we may consider $ \rho_m \in \mathcal{D}^{\sssup{MTF}} $ and construct as above the measure $ m_{\rho_m} \in \mathcal{D}^{\sssup{Vla}} $.
	Then for each $ x \in \mathbb{R}^3 $, $ m_{\rho_m} \myp{x, \nonarg} $ is by construction a minimizer of the functional
	\begin{equation*}
		\mathcal{E} \myp{\widetilde{m}} 
		:= \frac{1}{\myp{2 \pi}^2} \frac{\beta}{ 1+\beta} \sum\limits_{s = \pm 1} \sum\limits_{j=0}^{\infty} \int_{\mathbb{R}} \myp{p^2 + k_{\beta} \myp{2j + 1+s }} \widetilde{m} \myp{p,j,s} \id p
	\end{equation*}
	defined on the set of densities $ \widetilde{m} \in L^1 \myp{\mathbb{R} \times \mathbb{N}_0 \times \Set{\pm 1}} $ satisfying $ 0\leq \widetilde{m} \leq 1 $ and
	\begin{equation*}
		\frac{1}{\myp{2 \pi}^2} \frac{\beta}{ 1+\beta} \sum\limits_{s = \pm 1} \sum\limits_{j=0}^{\infty} \int_{\mathbb{R}} \widetilde{m} \myp{p,j,s} \id p = \rho_m \myp{x}.
	\end{equation*}
	(This is the bathtub principle \cite[Theorem $ 1.14 $]{LieLos-01}).
	Hence for almost every $ x \in \mathbb{R}^3 $,
	\begin{align*}
		\MoveEqLeft[3]	\sum\limits_{s = \pm 1} \sum\limits_{j=0}^{\infty} \int_{\mathbb{R}} \myp{p^2 + k_{\beta} \myp{2j+ 1+s }} m_{\rho_m} \myp{x,p,j,s} \id p \\
		&\leq \sum\limits_{s = \pm 1} \sum\limits_{j=0}^{\infty} \int_{\mathbb{R}} \myp{p^2 + k_{\beta} \myp{2j+ 1+s }} m \myp{x,p,j,s} \id p,
	\end{align*}
	implying $ \mathcal{E}_{\beta}^{\sssup{MTF}} \myp{\rho_m} = \mathcal{E}_{\beta}^{\sssup{Vla}} \myp{m_{\rho_m}} \leq \mathcal{E}_{\beta}^{\sssup{Vla}} \myp{m} $.
	We conclude that $ E^{\sssup{MTF}} \myp{\beta} \leq E^{\sssup{Vla}} \myp{\beta} $, so we have the desired result.
\end{proof}

To handle the extreme cases where $ \beta_N \to \infty $ or $ \beta_N \to 0 $, we need to introduce a couple of extra Vlasov type functionals.
\begin{de}[Strong field Vlasov energy]
	Define a functional by
	\begin{align*}
		\mathcal{E}_{\infty}^{\sssup{Vla}} \myp{m} = {}& \frac{1}{\myp{2 \pi}^2} \int_{\mathbb{R}} \int_{\mathbb{R}^3} p^2 m\myp{x,p} \id x \id p + \int_{\mathbb{R}^3} V \myp{x} \rho_m \myp{x} \id x \\
		&+ \frac{1}{2} \iint_{\mathbb{R}^6} w \myp{x-y} \rho_m \myp{x} \rho_m \myp{y} \id x \id y
	\end{align*}
	on the set
	\begin{equation*}
		\mathcal{D}_{\infty}^{\sssup{Vla}} = \Set{m \in L^1 \myp{\mathbb{R}^3 \times \mathbb{R}} \given 0 \leq m \leq 1, \ V\rho_m, \myp{w \ast \rho_m} \rho_m \in L^1 \myp{\mathbb{R}^3}},
	\end{equation*}
	where
	\begin{equation*}
		\rho_m \myp{x} := \frac{1}{\myp{2 \pi}^2} \int_{\mathbb{R}} m \myp{x,p} \id p.
	\end{equation*}
	The functional has ground state energy
	\begin{equation*}
		E^{\sssup{Vla}} \myp{\infty} = \inf \Set[\big]{\mathcal{E}_{\infty}^{\mathrm{Vla}} \myp{m} \given m \in \mathcal{D}_{\infty}^{\sssup{Vla}}, \frac{1}{\myp{2 \pi}^2} \iint_{\mathbb{R}^3 \times \mathbb{R}} m \myp{x,p} \id x \id p = 1 }.
	\end{equation*}
\end{de}
\begin{lem}
	\label{lem:functionalminimastrong}
	Suppose that $ \rho $ is in the domain of the strong Thomas-Fermi functional $ \mathcal{E}^{\sssup{STF}} $ and define a measure on $ \mathbb{R}^3 \times \mathbb{R} $ with density $ m_{\rho} \myp{x,p} = \mathds{1}_{\Set{p^2 \leq 4 \pi^4 \rho \myp{x}^2} } $.
	Then $ m_{\rho} \in \mathcal{D}_{\infty}^{\sssup{Vla}} $ and satisfies
	\begin{equation*}
		\rho_{m_{\rho}} \myp{x}
		:= \frac{1}{\myp{2 \pi}^2} \int_{\mathbb{R}} m_{\rho} \myp{x,p} \id p
		= \rho \myp{x},
	\end{equation*}
	\begin{equation}
		\frac{1}{\myp{2 \pi}^2} \int_{\mathbb{R}} p^2 m_{\rho} \myp{x,p} \id p
		= \frac{4 \pi^4}{3} \rho\myp{x}^3,
	\end{equation}
	and $ \mathcal{E}^{\sssup{STF}} \myp{\rho} = \mathcal{E}_{\infty}^{\sssup{Vla}} \myp{m_{\rho}} $.
	On the other hand, if $ m \in \mathcal{D}_{\infty}^{\sssup{Vla}} $, then $ \mathcal{E}^{\sssup{STF}} \myp{\rho_m} \leq \mathcal{E}_{\infty}^{\sssup{Vla}} \myp{m} $.
	In particular, $ E^{\sssup{STF}} = E^{\sssup{Vla}} \myp{\infty} $.
\end{lem}
This result is proved in exactly the same way as \cref{lem:functionalminima}, using the bathtub principle.
\begin{de}[Weak field Vlasov energy]
	Let $ b \geq 0 $ and define a functional by
	\begin{align*}
		\mathcal{E}_0^{\sssup{Vla}} \myp{m} = {}& \frac{1}{\myp{2 \pi}^3} \sum_{s= \pm 1} \iint_{\mathbb{R}^6} \myp{p+bA\myp{x}}^2 m\myp{x,p,s} \id x \id p \\
		&+ \int_{\mathbb{R}^3} V \myp{x} \rho_m \myp{x} \id x + \frac{1}{2} \iint_{\mathbb{R}^6} w \myp{x-y} \rho_m \myp{x} \rho_m \myp{y} \id x \id y
	\end{align*}
	on the set
	\begin{equation*}
		\mathcal{D}_0^{\sssup{Vla}} = \Set{m \in L^1 \myp{\mathbb{R}^6 \times \Set{\pm 1} } \given 0 \leq m \leq 1, \ V\rho_m, \myp{w \ast \rho_m} \rho_m \in L^1 \myp{\mathbb{R}^3}},
	\end{equation*}
	where
	\begin{equation*}
		\rho_m \myp{x} := \frac{1}{\myp{2 \pi}^3} \sum_{s= \pm 1} \int_{\mathbb{R}^3} m \myp{x,p,s} \id p.
	\end{equation*}
	The functional has ground state energy
	\begin{equation*}
		E^{\sssup{Vla}} \myp{0} = \inf \Set[\big]{\mathcal{E}_0^{\sssup{Vla}} \myp{m} \given m \in \mathcal{D}_0^{\sssup{Vla}}, \frac{1}{\myp{2 \pi}^3} \sum_{s= \pm 1} \iint_{\mathbb{R}^6} m \myp{x,p,s} \id x \id p = 1 }.
	\end{equation*}
\end{de}
\begin{lem}
	Suppose that $ \rho $ is in the domain of the usual Thomas-Fermi functional $ \mathcal{E}^{\sssup{TF}} $ and define a measure $ m_{\rho} $ on $ \mathbb{R}^3 \times \mathbb{R}^3 \times \Set{\pm 1} $ with density $ m_{\rho} \myp{x,p,s} = \mathds{1}_{ \Set{ \myp{p+bA\myp{x}}^2 \leq c_{\sssup{TF}} \rho \myp{x}^{2/3} } } $, where $ c_{\sssup{TF}} = \myp{ 3 \pi^2 }^{2/3} $.
	Then $ m_{\rho} \in \mathcal{D}_0^{\sssup{Vla}} $ and satisfies
	\begin{equation*}
		\rho_{m_{\rho}} \myp{x}
		:= \frac{1}{\myp{2 \pi}^3} \sum_{s=\pm 1} \int_{\mathbb{R}^3} m_{\rho} \myp{x,p,s} \id p
		= \rho \myp{x},
	\end{equation*}
	\begin{equation}
		\frac{1}{\myp{2 \pi}^3} \sum_{s= \pm 1} \int_{\mathbb{R}^3} \myp{p+bA\myp{x}}^2 m_{\rho} \myp{x,p,s} \id p
		= \frac{3}{5} c_{\sssup{TF}} \rho\myp{x}^{\frac{5}{3}},
	\end{equation}
	and $ \mathcal{E}^{\sssup{TF}} \myp{\rho} = \mathcal{E}_0^{\sssup{Vla}} \myp{m_{\rho}} $.
	On the other hand, if $ m \in \mathcal{D}_0^{\sssup{Vla}} $, then $ \mathcal{E}^{\sssup{TF}} \myp{\rho_m} \leq \mathcal{E}_0^{\sssup{Vla}} \myp{m} $.
	In particular, $ E^{\sssup{TF}} = E^{\sssup{Vla}} \myp{0} $.
\end{lem}
\begin{rem}
	\label{rem:vlamin}
	For any fixed density $ 0 \leq \rho \in L^1 \myp{\mathbb{R}^3} $, it follows from the uniqueness statement in \cite[Theorem $ 1.14 $]{LieLos-01} that for each fixed $ x \in \mathbb{R}^3 $, the measure $ m_{\rho} \myp{x,\nonarg} $ on $ \mathbb{R}^3 \times \Set{\pm 1} $ constructed above is the \emph{unique} minimizer of the functional
	\begin{equation*}
		m \mapsto \frac{1}{\myp{2 \pi}^3} \sum_{s= \pm 1} \int_{\mathbb{R}^3} \myp{p+b A\myp{x}}^2 m \myp{p,s} \id p
	\end{equation*}
	under the constraints $ 0 \leq m \leq 1 $ and $ \sum_{s= \pm 1} \int_{\mathbb{R}^3} m \myp{p,s} \id p = \myp{2 \pi}^3 \rho \myp{x} $.
	In particular, if $ m_0 $ is a minimizer of the Vlasov functional $ \mathcal{E}_0^{\sssup{Vla}} $, then the uniquess statement implies that
	\begin{equation*}
		m_0 \myp{x,p,s} = \mathds{1}_{ \Set{ \myp{p+bA\myp{x}}^2 \leq c_{\sssup{TF}} \rho_{m_0} \myp{x}^{2/3} } },
	\end{equation*}
	so the minimizers of $ \mathcal{E}_0^{\sssup{Vla}} $ are independent of the spin variable.
\end{rem}

%%%%%%%%%%%%%%%%%%%%%%%%%%%%%%%%%%%%%%%%%%%%%%%%%%%%%
%%%%%%%%%%%%%%%%%%%%%%%%%%%%%%%%%%%%%%%%%%%%%%%%%%%%%
\section{Upper energy bounds}
\label{sec:upperbound}
%%%%%%%%%%%%%%%%%%%%%%%%%%%%%%%%%%%%%%%%%%%%%%%%%%%%%
%%%%%%%%%%%%%%%%%%%%%%%%%%%%%%%%%%%%%%%%%%%%%%%%%%%%%
This section is devoted to proving the upper bounds in \cref{thm:energylowbound}, i.e.
\begin{prop}
\label{prop:energyupbound}
	With the assumptions in \cref{thm:energylowbound}, we have
	\begin{equation}
		\limsup_{N \to \infty} \frac{E \myp{N,\beta_N}}{N} 
		\leq	\begin{cases} 
					E^{\sssup{TF}}, 				&\text{if} \ \beta_N \to 0, \\
					E^{\sssup{MTF}} \myp{\beta},	&\text{if} \ \beta_N \to \beta \in \myp{0,\infty}, \\
					E^{\sssup{STF}},				&\text{if} \ \beta_N \to \infty.
				\end{cases}
	\end{equation}
\end{prop}
We will prove \cref{prop:energyupbound} by constructing an appropriate trial state for the variational problem.
Let $ \myp{f_j}_{j=1}^N $ be functions in the magnetic Sobolev space $ H_{\hbar^{-1}b A}^1 \myp{\mathbb{R}^3 ; \mathbb{C}^2} $, orthonormal in $ L^2 \myp{\mathbb{R}^3 ; \mathbb{C}^2} $.
Consider the corresponding Hartree-Fock state (abbrv. HF state) $ \Psi \in \bigwedge^N L^2 \myp{\mathbb{R}^3 ; \mathbb{C}^2} $ defined by
\begin{align*}
	\Psi \myp{x, s} = \frac{1}{\sqrt{N!}} \det \myb{f_i \myp{x_j,s_j}}
	&= \frac{1}{\sqrt{N!}} \sum\limits_{\sigma \in S_N} \mathrm{sgn} \myp{\sigma} \prod\limits_{j=1}^N f_{\sigma\myp{j}} \myp{x_j,s_j},
\end{align*}
where $ S_N $ is the symmetric group of $ N $ elements.
The function $ \Psi $ is normalized in $ L^2 \myp[\big]{ \mathbb{R}^{3N}; \mathbb{C}^{2^N} } $ and its one-particle density matrix is $ \gamma_{\Psi} = \sum\nolimits_{j=1}^N \ketbra{f_j}{f_j} $, so $ \tr \myb{\gamma_{\Psi}} = N $ and $ \gamma_{\Psi} $ is the orthogonal projection onto the subspace in $ L^2 \myp{\mathbb{R}^3; \mathbb{C}^2} $ spanned by the $ f_j $'s.
Furthermore, $ \gamma_{\Psi} $ has integral kernel
\begin{equation*}
	\gamma_{\Psi} \myp{x_1,s_1 ; x_2,s_2} = \sum\limits_{j=1}^N f_j \myp{x_1,s_1} \overline{f_j \myp{x_2,s_2}}, 
\end{equation*}
and the one-particle position density is $ \rho_{\Psi}^{\myp{1}} \myp{x} = \sum\nolimits_{s=\pm 1} \sum\nolimits_{j=1}^N \abs{f_j \myp{x,s}}^2 $.
Note that $ \normt{\gamma_{\Psi}}{2}^2 = \tr \myb{\gamma_{\Psi}} = N $ since the $ f_j $'s are orthonormal.  One easily calculates the expectation of the energy in the state $ \Psi $ to be
\begin{align}
	\innerp{ \Psi}{ H_{N,\beta_N} \Psi }
	= {}& \tr \myb{ \myp{\bm{\sigma} \cdot \myp{-i \hbar \nabla + b A}}^2 \gamma_{\Psi} } + \int_{\mathbb{R}^3} V \myp{x} \rho_{\Psi}^{\myp{1}} \myp{x} \id x \nonumber \\
	&+ \frac{1}{2N} \iint_{\mathbb{R}^6} w \myp{x-y} \rho_{\Psi}^{\myp{1}} \myp{x} \rho_{\Psi}^{\myp{1}} \myp{y} \id x \id y \nonumber \\
	&- \frac{1}{2N} \sum\limits_{s_1,s_2 = \pm 1} \iint_{\mathbb{R}^6} w \myp{x-y} \abs{\gamma_{\Psi} \myp{x,s_1;y,s_2}}^2 \id x \id y.
\label{eq:hartree-fock-energy}
\end{align}
We proceed to derive a bound on the exchange term involving $ \abs{\gamma_{\Psi}}^2 $.
The bound in the lemma below should hold for general fermionic states, but this will not be needed, so for simplicity we only prove it for Slater determinants.
\begin{lem}[Bound on the exchange term]
\label{lem:interactionbound}
	Let $ \Psi $ be an $ N $-body Slater determinant as above, and $ w \in L^{5/2} \myp{\mathbb{R}^3} + L^{\infty} \myp{\mathbb{R}^3} $.
	There is a constant $ C > 0 $ such that for each $ N \geq 1 $,
	\begin{align}
		\MoveEqLeft[3]	\frac{1}{2N} \sum\limits_{s_1,s_2 = \pm 1} \iint_{\mathbb{R}^6} \abs{w \myp{x-y}} \abs{\gamma_{\Psi} \myp{x,s_1;y,s_2}}^2 \id x \id y \nonumber \\
		&\leq C N^{\frac{2}{3}} \myp[\Big]{ \frac{1}{N} \tr \myb{ \myp{ \bm{\sigma} \cdot \myp{-i \hbar \nabla + b A}}^2 \gamma_{\Psi} } + \beta_N \myp{1+\beta_N}^{-\frac{4}{5}} +1 }.
	\label{eq:interactionbound}
	\end{align}
\end{lem}
\begin{proof}
	We mimic the proof of the analogous bound in \cite[Proposition $ 3.1 $]{FouLewSol-18}.
	Writing $ w = w_1 + w_2 $ with $ w_1 \in L^{3/2} \myp{\mathbb{R}^3} \cap L^{5/2} \myp{\mathbb{R}^3} $ and $ w_2 \in L^{\infty} \myp{\mathbb{R}^3} $.
	Using that $ \gamma_{\Psi} $ is a projection, we note that the contribution from $ w_2 $ is bounded by
	\begin{equation}
	\label{eq:w2bound}
		\frac{1}{2N} \sum\limits_{s_1,s_2 = \pm 1} \iint_{\mathbb{R}^6} \abs{w_2 \myp{x-y}} \abs{\gamma_{\Psi} \myp{x,s_1;y,s_2}}^2 \id x \id y
		\leq \frac{\normt{w_2}{\infty}}{2},
	\end{equation}
	so we concentrate on controlling the contribution from $ w_1 $.
	
	Now, for any function $ f $ in the magnetic Sobolev space $ H_{\hbar^{-1} bA}^1 \myp{\mathbb{R}^3} $, the diamagnetic inequality implies that $ \abs{f} \in H^1 \myp{\mathbb{R}^3} $.
	Defining $ f_{\varepsilon} \myp{x} := \varepsilon^{1/2} f \myp{\varepsilon x} $ for $ \varepsilon > 0 $, we have $ \normt{f_{\varepsilon}}{6} = \normt{f}{6} $, $ \normt{f_{\varepsilon}}{2}^2 = \varepsilon^{-2} \normt{f}{2}^2 $, and $ \normt{ \nabla \abs{f_{\varepsilon}} }{2}^2 = \normt{\nabla \abs{f}}{2}^2 $, so by the diamagnetic and Gagliardo-Nirenberg-Sobolev inequalities,
	\begin{align*}
		\normt{f}{6}^2 
		&= \normt{f_{\varepsilon}}{6}^2 
		\leq C \myp{ \normt{ \nabla\abs{f_{\varepsilon}} }{2}^2 + \normt{f_{\varepsilon}}{2}^2 }
		= C \myp{ \normt{ \nabla\abs{f} }{2}^2 + \varepsilon^{-2} \normt{f}{2}^2 } \\
		&\leq C \myp{ \hbar^{-2} \normt{ \myp{-i\hbar \nabla+ b A} f }{2}^2 + \varepsilon^{-2} \normt{f}{2}^2 }.
	\end{align*}
	Combining this with the Hölder inequality, we obtain
	\begin{align*}
		\int_{\mathbb{R}^3} \abs{w_1 \myp{x}} \abs{f \myp{x}}^2 \id x 
		\leq C \normt{w_1}{\frac{3}{2}} \myp{ \hbar^{-2} \normt{ \myp{-i\hbar \nabla+ b A} f }{2}^2 + \varepsilon^{-2} \normt{f}{2}^2 }.
	\end{align*}
	We will apply this to the function $ \gamma_{\Psi} \myp{\nonarg, s_1 ; y,s_2} $ for fixed $ y $, so we calculate, again using that $ \gamma_{\Psi} $ is a projection,
	\begin{equation*}
		\sum\limits_{s_1,s_2 = \pm 1} \int_{\mathbb{R}^3} \normt{ \myp{-i \hbar \nabla_x + b A \myp{x}} \gamma_{\Psi} \myp{\nonarg, s_1 ; y,s_2}}{2}^2 \id y
		= \tr \myb{ \myp{-i \hbar \nabla + b A}^2 \gamma_{\Psi} }.
	\end{equation*}
	Recalling \eqref{eq:paulieq} and noting that $ \tr \myb{\sigma_3 \gamma_{\Psi}} \geq -N $, we combine the bounds above to obtain
	\begin{align*}
		\MoveEqLeft[3]	\sum\limits_{s_1,s_2 = \pm 1} \iint_{\mathbb{R}^6} \abs{w_1 \myp{x-y}} \abs{\gamma_{\Psi} \myp{x,s_1;y,s_2}}^2 \id x \id y \nonumber \\
		&\leq C \normt{w_1}{\frac{3}{2}} \myp[\Big]{ \frac{1}{\hbar^2} \tr \myb{ \myp{ \bm{\sigma} \cdot \myp{-i \hbar \nabla + b A}}^2 \gamma_{\Psi} } - \frac{1}{\hbar^2} \tr \myb{ \hbar b \sigma_3 \gamma_{\Psi}} + \varepsilon^{-2} N } \nonumber \\
		&\leq C \normt{w_1}{\frac{3}{2}} \myp[\Big]{ \frac{1}{\hbar^2} \tr \myb{ \myp{ \bm{\sigma} \cdot \myp{-i \hbar \nabla + b A}}^2 \gamma_{\Psi} } + \frac{b}{\hbar} N + \varepsilon^{-2} N }.
	\end{align*}
	Now, choosing $ \varepsilon^2 = \hbar/b $ and recalling the definitions of $ \hbar $ and $ b $ \eqref{eq:scaling0}, we get
	\begin{align*}
		\MoveEqLeft[3]	\frac{1}{2N} \sum\limits_{s_1,s_2 = \pm 1} \iint_{\mathbb{R}^6} \abs{w_1 \myp{x-y}} \abs{\gamma_{\Psi} \myp{x,s_1;y,s_2}}^2 \id x \id y \\
		&\leq C \normt{w_1}{\frac{3}{2}} N^{\frac{2}{3}} \myp[\Big]{ \frac{1}{N} \tr \myb{ \myp{ \bm{\sigma} \cdot \myp{-i \hbar \nabla + b A}}^2 \gamma_{\Psi} } + \beta_N \myp{1+\beta_N}^{-\frac{4}{5}} },
	\end{align*}
	so combining with \eqref{eq:w2bound}, we obtain \eqref{eq:interactionbound}.
\end{proof}
Continuing \eqref{eq:hartree-fock-energy}, recalling \eqref{eq:scaling2} (the assumption $ N^{-1/3} \beta_N^{1/5} \to 0 $), and applying the min-max principle, we get the bound
\begin{align}
	\MoveEqLeft[3]	\limsup_{N \to \infty} \frac{E\myp{N,\beta_N}}{N} \nonumber \\
	\leq {}& \limsup_{N\to \infty} \inf\limits_{\substack{\Psi \ \text{HF-} \\ \text{state} }} \myt[\Big]{ \frac{1+C N^{-\frac{1}{3}}}{N} \tr \myb{ \myp{\bm{\sigma} \cdot \myp{-i \hbar \nabla + b A}}^2 \gamma_{\Psi} } \nonumber \\
	&+ \frac{1}{N} \int_{\mathbb{R}^3} V \myp{x} \rho_{\Psi}^{\myp{1}} \myp{x} \id x + \frac{1}{2N^2} \iint_{\mathbb{R}^6} w \myp{x-y} \rho_{\Psi}^{\myp{1}} \myp{x} \rho_{\Psi}^{\myp{1}} \myp{y} \id x \id y }.
\label{eq:uppartialbound}
\end{align}
We proceed to construct an appropriate trial state for this variational problem.
For $ R > 0 $ we denote by $ \myp{\bm{\sigma} \cdot \myp{-i \hbar \nabla + b A}}^2_{C_R} $ the Pauli operator in the cube $ C_R = \myp{-R/2, R/2}^3 $ with Dirichlet boundary conditions.
\begin{lem}
\label{lem:trialstate}
	Suppose that $ \beta_N \to \beta \in \myp{0,\infty} $ and let $ 0 \leq \rho \in C_c \myp{\mathbb{R}^3} $ be any function with $ \int_{\mathbb{R}^3} \rho \myp{x} \id x = 1 $.
	Define $ r \myp{x} $ to be the solution to the equation
	\begin{equation*}
		\rho \myp{x} = \myp{1+\beta}^{-\frac{3}{5}} P'_{k_{\beta}} \myp{r \myp{x}},
	\end{equation*}
	where $ P_{k_{\beta}} $ is the pressure of the free Landau gas \eqref{eq:landaupressure} and $ k_{\beta} = \beta \myp{1+\beta}^{-2/5} $,	cf. \cref{lem:functionalminima}. 	
	Furthermore, fix $ R > 0 $ large enough such that $ \supp \rho \subseteq C_R $.
	Then the sequence of density matrices $ \gamma_N $ given by the spectral projections
	\begin{equation}
	\label{eq:gamman}
		\gamma_N := \mathds{1}_{\left( -\infty, 0 \right]} \myp{ \myp{\bm{\sigma} \cdot \myp{-i \hbar \nabla + b A}}^2_{C_R} - r \myp{x} },
	\end{equation}
	satisfies
	\begin{equation}
	\label{eq:kinlimit}
		\lim_{N \to \infty} \frac{1}{N} \tr \myb{ \myp{\bm{\sigma} \cdot \myp{-i \hbar \nabla + b A}}^2 \gamma_N }
		= \int_{\mathbb{R}^3} \tau_{\beta} \myp{\rho \myp{x}} \id x
	\end{equation}
	and
	\begin{equation}
	\label{eq:traceconv}
		\lim_{N \to \infty} \frac{1}{N} \tr \myb{\gamma_N} = \int_{\mathbb{R}^3} \rho \myp{x} \id x = 1.
	\end{equation}
	Moreover, the densities $ \frac{1}{N}\rho_{\gamma_N} $ converge to $ \rho $ weakly in $ L^1 \myp{\mathbb{R}^3} $ and $ L^{5/3} \myp{\mathbb{R}^3} $, and the same conclusions also hold if $ \gamma_N $ is replaced by the projection $ \tilde{\gamma}_N $ onto the $ N $ lowest eigenvectors of the operator $ \myp{\bm{\sigma} \cdot \myp{-i \hbar \nabla + b A}}^2_{C_R} - r \myp{x} $.
\end{lem}
\begin{proof}
	For the duration of the proof, we will employ the notation $ \smash{T_{C_R}^{\beta_N}} = \myp{\bm{\sigma} \cdot \myp{-i \hbar \nabla + b A}}^2_{C_R} $.
	By domain inclusions it is not difficult to see that in the sense of quadratic forms $ \tr \myb{ \smash{T_{C_R}^{\beta_N}} \gamma_N } = \tr \myb{ \myp{\bm{\sigma} \cdot \myp{-i \hbar \nabla + b A}}^2 \gamma_N } $, and that the same equality holds when $ \gamma_N $ is replaced by $ \tilde{\gamma}_N $.
	Thus, it is sufficient to show \eqref{eq:kinlimit} using $ \smash{T_{C_R}^{\beta_N} } $ instead of the Pauli operator on the whole space.
	
	Note also that the quadratic form domain of $ \smash{T_{C_R}^{\beta_N}} - r \myp{x} $, $ H_0^1 \myp{C_R ; \mathbb{C}^2} $, is compactly embedded in $ L^2 \myp{C_R ; \mathbb{C}^2} $, so that $ \smash{T_{C_R}^{\beta_N}} - r \myp{x} $ has compact resolvent, and hence it has purely discrete spectrum.
	This implies that $ \gamma_N $ is a projection onto a finite-dimensional subspace of $ L^2 \myp{C_R} $, and hence $ \rho_{{\gamma}_N} $ is an $ L^1 $-function.
	
	Using the Weyl asymptotics from \cref{cor:pauliweil,rem:pauliweyl} and recalling \eqref{eq:taueq}, we obtain in the semi-classical
	limit
	\begin{align}
		\lim_{N \to \infty} \frac{1}{N} \tr \myb{ \myp{ T_{C_R}^{\beta_N} - r \myp{x} } \gamma_N }
		&= - \myp{1+\beta}^{-\frac{3}{5}} \int_{C_R} P_{k_{\beta}} \myp{r \myp{x}} \id x \nonumber \\
		&= \int_{C_R} \tau_{\beta} \myp{\rho \myp{x}} \id x - \int_{C_R} \rho \myp{x} r \myp{x} \id x.
	\label{eq:dirichletenergy}
	\end{align}
	Let now $ g \in L^{\infty} \myp{C_R} $ be real and non-negative. By a Feynman-Hellmann type argument, we shall see that
	\begin{equation}
	\label{eq:rhogammalim}
		\lim_{N \to \infty} \frac{1}{N} \tr \myb{g \myp{x} \gamma_N} = \int_{C_R} g\myp{x} \rho \myp{x} \id x.
	\end{equation}
	To this end, note first for any real $ \delta $ that any function in the range of $ \gamma_N $ is also in the domain of $ \smash{T_{C_R}^{\beta_N}} - r \myp{x} + \delta g \myp{x} $.
	Hence we have by the variational principle
	\begin{equation*}
		\tr \myb{ \myp{ T_{C_R}^{\beta_N} - r \myp{x} + \delta g \myp{x} } \gamma_N }
		\geq \tr \myp{ T_{C_R}^{\beta_N} - r \myp{x} + \delta g \myp{x} }_-,
	\end{equation*}
	so that
	\begin{align}
		\delta \tr \myb{g \myp{x} \gamma_N} 
		&\geq \tr \myp{ T_{C_R}^{\beta_N} - r \myp{x} + \delta g \myp{x} }_- - \tr \myp{ T_{C_R}^{\beta_N} - r \myp{x} }_-.
	\label{eq:tracevariation}
	\end{align}
	Hence for $ \delta < 0 $, we get by \cref{cor:pauliweil},
	\begin{align*}
		\MoveEqLeft[3]	\limsup_{N \to \infty} \frac{1}{N} \tr \myb{g \myp{x} \gamma_N} \\
		&\leq \frac{1}{\delta} \lim_{N \to \infty} \frac{1}{N} \myp[\big]{ \tr \myp{ T_{C_R}^{\beta_N} - r \myp{x} + \delta g \myp{x} }_- - \tr \myp{ T_{C_R}^{\beta_N} - r \myp{x} }_- } \\
		&= \myp{1+\beta}^{-\frac{3}{5}} \int_{C_R} \frac{P_{k_\beta} \myp{r \myp{x} - \delta g \myp{x}} - P_{k_\beta} \myp{r \myp{x}}}{- \delta g \myp{x}} g \myp{x} \id x.
	\end{align*}
	Since $ g $ is non-negative and $ P_{k_\beta} $ is convex and increasing, the integrand above decreases pointwise to $ P'_{k_\beta} \myp{r \myp{x}} $ as $ \delta \to 0_- $, on the set where $ g \myp{x} \neq 0 $.
	This implies by the monotone convergence theorem and definition of $ r $ that
	\begin{align*}
		\limsup_{N \to \infty} \frac{1}{N} \tr \myb{g \myp{x} \gamma_N}
		&\leq \myp{1+\beta}^{-\frac{3}{5}} \int_{C_R} P'_{k_\beta} \myp{r \myp{x}} g \myp{x} \id x \\
		&= \int_{C_R} \rho \myp{x} g \myp{x} \id x.
	\end{align*}
	In the same way we get from \eqref{eq:tracevariation} for positive $ \delta $, that
	\begin{align*}
		\MoveEqLeft[3]	\liminf_{N \to \infty} \frac{1}{N} \tr \myb{g \myp{x} \gamma_N} \\
		&\geq \myp{1+\beta}^{-\frac{3}{5}} \int_{C_R} \frac{P_{k_\beta} \myp{\myb{r \myp{x} - \delta g \myp{x}}_+} - P_{k_\beta} \myp{r \myp{x}}}{- \delta g \myp{x}} g \myp{x} \id x \\
		&\xrightarrow{\delta \to 0_+} \myp{1+\beta}^{-\frac{3}{5}} \int_{C_R} P'_{k_\beta} \myp{r \myp{x}} g \myp{x} \id x 
		= \int_{C_R} \rho \myp{x} g \myp{x} \id x,
	\end{align*}
	since the fraction in the integral this time increases to $ P'_{k_\beta} \myp{r \myp{x}} $ on the set where $ g \myp{x} \neq 0 $.
	It follows that \eqref{eq:rhogammalim} holds,
	and by extension that for arbitrary $ g \in L^{\infty} \myp{\mathbb{R}^3} $, we have
	\begin{equation*}
		\int_{\mathbb{R}^3} g \myp{x} \frac{\rho_{\gamma_N} \myp{x}}{N} \id x 
		= \frac{1}{N} \tr \myb{ \myp{g \mathds{1}_{C_R}} \myp{x} \gamma_N}
		\to \int_{\mathbb{R}^3} g \myp{x} \rho \myp{x} \id x,
	\end{equation*}
	as $ N $ tends to infinity, so $ \frac{1}{N} \rho_{\gamma_N} \rightharpoonup \rho $ weakly in $ L^1 \myp{\mathbb{R}^3} $, as advertised.
	Taking $ g = \mathds{1}_{C_R} $ in \eqref{eq:rhogammalim} yields \eqref{eq:traceconv}, implying that $ \smash{T_{C_R}^{\beta_N}} -r \myp{x} $ has $ N + o \myp{N} $ negative eigenvalues.
	Noting that $ r $ is bounded by construction, we can take $ g = r $ in \eqref{eq:rhogammalim} and combine with \eqref{eq:dirichletenergy} to obtain
	\begin{equation}
	\label{eq:kindirgammaconv}
		\lim_{N \to \infty} \frac{1}{N} \tr \myb{ T_{C_R}^{\beta_N} \gamma_N } 
		= \int_{C_R} \tau_{\beta} \myp{\rho \myp{x}} \id x.
	\end{equation}
	Finally, applying \cref{lem:densitybound} to get weak convergence in $ L^{5/3} \myp{\mathbb{R}^3} $, we have proven the lemma for $ \gamma_N $.
	
	We want to see that the assertions of the lemma also hold for $ \tilde{\gamma}_N $.
	The fact that the dimension of the range of $ \gamma_N $ is $ N + o \myp{N} $ immediately implies that $ \normt{\rho_{{\gamma}_N} - \rho_{\tilde{\gamma}_N}}{1} = \tr \myb{ \abs{\gamma_N - \tilde{\gamma}_N}} = o \myp{N} $.
	Hence for any $ g \in L^{\infty} \myp{\mathbb{R}^3} $ we have
	\begin{equation}
	\label{eq:tildeweakconv}
		\tr \myb{g \myp{x} \myp{\gamma_N - \tilde{\gamma}_N}} = o \myp{N}.
	\end{equation}	
	In other words, $ \frac{1}{N}\rho_{\tilde{\gamma}_N} $ has the same weak limit in $ L^1 \myp{\mathbb{R}^3} $ as $ \frac{1}{N}\rho_{\gamma_N} $.
	Note by \eqref{eq:landaupressurederiv} that $ P'_{k_{\beta}} $ is continuous and increasing, which along with continuity of $ \rho $ implies that $ r $ is continuous.
	Also, it is clear that $ \mathrm{supp} \, r = \mathrm{supp} \, \rho \subseteq C_R $, so we get by uniform continuity that for each $ \delta > 0 $ there is some $ \varepsilon \in \left(0,\delta \right] $ such that for all $ x \in C_R $, we have
	\begin{equation}
	\label{eq:punifbound}
		\myp{1+\beta}^{-\frac{3}{5}} \abs{ P_{k_\beta}' \myp{ \myb{r \myp{x} \pm \varepsilon}_+} - P_{k_\beta}' \myp{r \myp{x}}} \leq \delta.
	\end{equation}
	Use this $ \varepsilon $ to define
	\begin{equation*}
		\gamma_{N,\pm \varepsilon} 
		:= \mathds{1}_{\left(-\infty,\pm \varepsilon \right]} \myp{ T_{C_R}^{\beta_N} -r \myp{x} } 
		= \mathds{1}_{\left(-\infty,0 \right]} \myp{ T_{C_R}^{\beta_N} - \myp{r \myp{x} \pm \varepsilon} }.
	\end{equation*}
	Redoing the argument used to prove \eqref{eq:rhogammalim}, we obtain
	\begin{equation*}
		\lim_{N \to \infty} \frac{1}{N} \tr \myb{\gamma_{N,\pm \varepsilon}} 
		= \myp{1+\beta}^{-\frac{3}{5}} \int_{C_R} P_{k_\beta}' \myp{ \myb{r \myp{x} \pm \varepsilon}_+} \id x,
	\end{equation*}
	and since $ P_{k_\beta}' $ strictly increasing on $ \left[ 0,\infty \right) $, we have
	\begin{equation*}
		\eta_{\pm} := \pm \myp{1+\beta}^{-\frac{3}{5}} \int_{C_R} P_{k_\beta}' \myp{ \myb{r \myp{x} \pm \varepsilon}_+} - P_{k_\beta}' \myp{r \myp{x}} \id x
		> 0
	\end{equation*}
	as long as $ \varepsilon $ is small enough, implying
	\begin{equation*}
		\lim_{N \to \infty} \frac{1}{N} \tr \myb{\gamma_{N, -\varepsilon}}
		\leq 1 - \eta_-
		\leq 1 + \eta_+
		\leq \lim_{N \to \infty} \frac{1}{N} \tr \myb{\gamma_{N, \varepsilon}}.
	\end{equation*}
	These bounds yield for $ N $ large enough that $ \tr \myb{\gamma_{N, -\varepsilon}} \leq N \leq \tr \myb{\gamma_{N, \varepsilon}} $, so
	\begin{equation}
		\gamma_{N,-\varepsilon} \leq \tilde{\gamma}_N \leq \gamma_{N,\varepsilon}.
	\end{equation}
	Similarly, using \eqref{eq:punifbound} along with the fact that $ P_{k_\beta}' $ is increasing, we have
	\begin{equation*}
		1 - \delta R^3
		\leq \lim_{N \to \infty} \frac{1}{N} \tr \myb{\gamma_{N, -\varepsilon}}
		\leq \lim_{N \to \infty} \frac{1}{N} \tr \myb{\gamma_{N, \varepsilon}}
		\leq 1 + \delta R^3.
	\end{equation*}
	Now, for each $ N \geq 1 $, there are two cases; either $ \tilde{\gamma}_N $ is a subprojection of $ \gamma_N $, or the converse is true.
	In case $ \tilde{\gamma}_N \leq \gamma_N $, we have $ \gamma_N - \tilde{\gamma}_N \geq \gamma_N - \gamma_{N,-\varepsilon} $, where the latter is the spectral projection of $ \smash{T_{C_R}^{\beta_N}} - r \myp{x} $ corresponding to the interval $ \left( -\varepsilon ,0 \right] $.
	Hence we have
	\begin{align*}
		0 &\geq \tr \myb{ \myp{ T_{C_R}^{\beta_N} - r \myp{x} } \myp{\gamma_N-\tilde{\gamma}_N} }
		\geq \tr \myb{ \myp{ T_{C_R}^{\beta_N} - r \myp{x} } \myp{\gamma_N-\gamma_{N,-\varepsilon}} } \\
		&\geq - \varepsilon \tr \myb{\gamma_N - \gamma_{N,-\varepsilon}}
		\geq - \varepsilon \myp{\delta R^3 N + o \myp{N}}.
	\end{align*}
	The other case, where $ \gamma_N \leq \tilde{\gamma}_N $, is handled similarly.
	Here we get the bound
	\begin{align*}
		0 &\leq -\tr \myb{ \myp{ T_{C_R}^{\beta_N} - r \myp{x} } \myp{\gamma_N-\tilde{\gamma}_N} }
		\leq \tr \myb{ \myb{ T_{C_R}^{\beta_N} - r \myp{x} } \myp{\gamma_{N,\varepsilon}-\gamma_N} } \\
		&\leq \varepsilon \tr \myb{\gamma_{N,\varepsilon} - \gamma_N} 
		\leq \varepsilon \myp{\delta R^3 N + o \myp{N}}.
	\end{align*}
	In either case,
	\begin{equation*}
		\tr \myb{ \myp{ T_{C_R}^{\beta_N} - r \myp{x} } \myp{\gamma_N-\tilde{\gamma}_N} } = o \myp{N},
	\end{equation*}
	so combining with \eqref{eq:tildeweakconv}, we obtain $ \tr \myb{ \smash{T_{C_R}^{\beta_N}} \myp{\gamma_N-\tilde{\gamma}_N} } = o \myp{N} $, meaning that \eqref{eq:kindirgammaconv} also holds for $ \tilde{\gamma}_N $, finishing the proof.
\end{proof}
In the regimes where either $ \beta_N \to 0 $ or $ \beta_N \to \infty $, we modify the proof above to obtain similar results:
\begin{lem}
	\label{lem:trialstate2}
	Let $ 0 \leq \rho \in C_c \myp{\mathbb{R}^3} $ be any function with $ \int_{\mathbb{R}^3} \rho \myp{x} \id x = 1 $, and fix $ R > 0 $ large enough such that $ \supp \rho \subseteq C_R $.
	\begin{enumerate}
		\item If $ \beta_N \to 0 $, then the sequence of density matrices $ \gamma_N $ given by the spectral projections
		\begin{equation*}
			\gamma_N := \mathds{1}_{\left( -\infty, 0 \right]} \myp[\big]{ \myp{\bm{\sigma} \cdot \myp{-i \hbar \nabla + b A}}^2_{C_R} - c_{\sssup{TF}} \rho \myp{x}^{\frac{2}{3}} }
		\end{equation*}
		satisfies
		\begin{equation*}
			\lim_{N \to \infty} \frac{1}{N} \tr \myb{ \myp{\bm{\sigma} \cdot \myp{-i \hbar \nabla + b A}}^2 \gamma_N }
			= \frac{3}{5} c_{\sssup{TF}} \int_{\mathbb{R}^3} \rho \myp{x}^{\frac{5}{3}} \id x.
		\end{equation*}
		
		\item If $ \beta_N \to \infty $ and \eqref{eq:scaling2} holds, then the sequence of density matrices $ \gamma_N $ given by the spectral projections
		\begin{equation*}
			\gamma_N := \mathds{1}_{\left( -\infty, 0 \right]} \myp{ \myp{\bm{\sigma} \cdot \myp{-i \hbar \nabla + b A}}^2_{C_R} - 4 \pi^4 \rho\myp{x}^2 }
		\end{equation*}
		satisfies
		\begin{equation*}
			\lim_{N \to \infty} \frac{1}{N} \tr \myb{ \myp{\bm{\sigma} \cdot \myp{-i \hbar \nabla + b A}}^2 \gamma_N }
			= \frac{4 \pi^4}{3} \int_{\mathbb{R}^3} \rho \myp{x}^3 \id x.
		\end{equation*}
	\end{enumerate}
	Moreover, in both cases we have
	\begin{equation*}
		\lim_{N \to \infty} \frac{1}{N} \tr \myb{\gamma_N} = \int_{\mathbb{R}^3} \rho \myp{x} \id x = 1,
	\end{equation*}
	and the densities $ \frac{1}{N}\rho_{\gamma_N} $ converge to $ \rho $ weakly in $ L^1 \myp{\mathbb{R}^3} $ and in
	$ L^{5/3} \myp{\mathbb{R}^3} $.
	The same conclusions also hold if $ \gamma_N $ is replaced by the projection $ \tilde{\gamma}_N $ onto the $ N $ lowest eigenvectors of the operator used to define $ \gamma_N $.
\end{lem}
\begin{proof}
	The proof of \cref{lem:trialstate} also holds mutatis mutandis for this lemma, and we omit the details.
\end{proof}
Using the trial states constructed above we can now show the upper bound on the energy.
\begin{proof}[Proof of \cref{prop:energyupbound}.]
	Let $ 0\leq \rho \in C_c \myp{\mathbb{R}^3} $ with $ \int \rho \myp{x} \id x = 1 $, and take $ \tilde{\gamma}_N $ as in either \cref{lem:trialstate} or \cref{lem:trialstate2}, depending on the sequence $ \myp{\beta_N} $.
	Since $ V \in L_{\mathrm{loc}}^{5/2} \myp{\mathbb{R}^3} $ and $ \rho_{\tilde{\gamma}_N} $ is supported inside the box $ C_R $, we get by weak convergence of $ \frac{1}{N} \rho_{\tilde{\gamma}_N} $ that
	\begin{equation*}
		\frac{1}{N} \int_{\mathbb{R}^3} V \myp{x} \rho_{\tilde{\gamma}_N} \myp{x} \id x 
		\longrightarrow \int_{C_R} V \myp{x} \rho \myp{x} \id x
	\end{equation*}
	as $ N $ tends to infinity.
	By the Stone-Weierstrass theorem, we may approximate $ w \myp{x-y} $ in $ L^{5/2} \myp{C_R^2} $ by a function of the form $ w_0 = \sum_{j=1}^k g_j \otimes h_j $ with $ g_j, h_j \in C \myp{C_R} $.
	Denoting $ D_w \myp{\rho} := \iint w \myp{x-y} \rho \myp{x} \rho \myp{y} \id x \id y $, we have
	\begin{equation*}
		D_w \myp{N^{-1} \rho_{\tilde{\gamma}_N}}
		= D_{w_0} \myp{N^{-1} \rho_{\tilde{\gamma}_N}} + D_{w-w_0} \myp{N^{-1} \rho_{\tilde{\gamma}_N}},
	\end{equation*}
	where by the Hölder inequality,
	\begin{equation*}
		\abs{D_{w-w_0} \myp{N^{-1} \rho_{\tilde{\gamma}_N}}}
		\leq \normt{\myp{w-w_0} \myp{x-y}}{L^{5/2} \myp{C_R^2}} \normt{N^{-1} \rho_{\tilde{\gamma}_N}}{L^{5/3} \myp{C_R}}^2.
	\end{equation*}
	Here, the right hand side can be made arbitrarily small, since $ N^{-1} \rho_{\tilde{\gamma}_N} $ is bounded in $ L^{5/3} \myp{\mathbb{R}^3} $, by \cref{lem:densitybound}.
	By the weak convergence of $ N^{-1} \rho_{\tilde{\gamma}_N} $ and the explicit form of $ w_0 $,
	\begin{equation*}
		D_{w_0} \myp{N^{-1} \rho_{\tilde{\gamma}_N}} 
		\longrightarrow D_{w_0} \myp{\rho},
	\end{equation*}
	where $ D_{w_0} \myp{\rho} $ is arbitrarily close to $ D_w \myp{\rho} $ (again using the Hölder inequality), so we conclude that
	\begin{equation*}
		\frac{1}{N^2} \iint_{\mathbb{R}^6} w \myp{x-y} \rho_{\tilde{\gamma}_N} \myp{x} \rho_{\tilde{\gamma}_N} \myp{y} \id x \id y 
		\longrightarrow \iint_{C_R^2} w \myp{x-y} \rho \myp{x} \rho \myp{y} \id x \id y.
	\end{equation*}
	Hence, continuing from \eqref{eq:uppartialbound} with $ \gamma_{\Psi} = \tilde{\gamma}_N $, we find (for instance in the case where $ \beta_N \to \beta \in \myp{0,\infty} $)
	\begin{equation*}
		\limsup_{N \to \infty} \frac{E\myp{N,\beta_N}}{N} 
		\leq \mathcal{E}_{\beta}^{\sssup{MTF}} \myp{\rho}.
	\end{equation*}
	If $ \beta_N \to 0 $, or $ \beta_N \to \infty $, we obtain analogous bounds by appealing to \cref{lem:trialstate2}.
	This concludes the proof since the Thomas-Fermi ground state energy can be obtained by minimizing over compactly supported, continuous functions, and $ \rho \in C_c \myp{\mathbb{R}^3} $ is arbitrary.
\end{proof}

%%%%%%%%%%%%%%%%%%%%%%%%%%%%%%%%%%%%%%%%%%%%%%%%%%%%%
%%%%%%%%%%%%%%%%%%%%%%%%%%%%%%%%%%%%%%%%%%%%%%%%%%%%%
\section{Semi-classical measures}
\label{sec:semimeasures}
%%%%%%%%%%%%%%%%%%%%%%%%%%%%%%%%%%%%%%%%%%%%%%%%%%%%%
%%%%%%%%%%%%%%%%%%%%%%%%%%%%%%%%%%%%%%%%%%%%%%%%%%%%%
Having established the upper bound on the energy, we turn our attention towards proving the lower bound.
In order to do this, we will construct semi-classical measures using coherent states, and see that these measures have some very nice properties in the limit as the number of particles tends to infinity.
Afterwards, a de Finetti theorem may be applied to yield general information about the structure in the limit.
The constructions in this section are only useful for dealing with the case where $ \beta_N \to \beta > 0 $.
In the case where $ \beta_N \to 0 $, it is more convenient to use the same semi-classical measures as in \cite{FouLewSol-18}.
This case is treated in \cref{sec:lowerbound2}.

The first step is to diagonalize the three-dimensional magnetic Laplacian, i.e., we consider
\begin{equation}
\label{eq:HA}
	H_A = \myp{-i\nabla + A}^2 
	= H_{A^{\perp}} - \partial_{x_3}^2,
\end{equation}
where $ A^{\perp} \myp{x_1,x_2} = \frac{1}{2} \myp{-x_2,x_1} $, and $ H_{A^{\perp}} := \myp{ -i \nabla + A^{\perp} }^2 $ acts on $ L^2 \myp{\mathbb{R}^2} $.
Letting $ \mathcal{F}_2 $ denote the partial Fourier transform in the second variable on $ L^2 \myp{\mathbb{R}^2} $, and $ T $ the unitary operator on $ L^2 \myp{\mathbb{R}^2} $ defined by $ \myp{T \varphi} \myp{x_1, \xi} = \varphi \myp{x_1+\xi, \xi} $, an elementary calculation shows that
\begin{equation}
\label{eq:uniequiv}
	H_{A^{\perp}} e^{i \frac{1}{2} x_1 x_2} \mathcal{F}_2^{-1} T
	= e^{i \frac{1}{2} x_1 x_2} \mathcal{F}_2^{-1} T \myp[\Big]{ \myp[\Big]{ -\frac{ \id^2}{\id x_1^2} + x_1^2 } \otimes \mathds{1}_{L^2 \myp{\mathbb{R}}} }.
\end{equation}
It is very well known that the harmonic oscillator admits an orthonormal basis of eigenfunctions $ \myp{f_j}_{j \geq 0} $ for $ L^2 \myp{\mathbb{R}} $, with $ \myp{ -\frac{ \id ^2}{\id x^2} + x^2 } f_j = \myp{2j+1} f_j $ and $ f_0 \myp{x} = \smash{\pi^{-\frac{1}{4}} e^{- \frac{1}{2}x^2}} $.
In particular, equation \eqref{eq:uniequiv} means for any $ j \geq 0 $ and any normalized Schwartz function $ v $ on $ \mathbb{R} $, that $ e^{i \frac{1}{2} x_1 x_2} \mathcal{F}_2^{-1} T \myp{f_j \otimes v} $ is a normalized eigenfunction for $ H_{A^{\perp}} $ with corresponding eigenvalue $ 2 j+1 $.

Suppose that $ \varphi $ is an eigenfunction for $ H_{A^{\perp}} $ corresponding to $ 2j+1 $.
If we scale the magnetic field and instead consider $ H_{BA^{\perp}} = \myp{-i \nabla + B A^{\perp}}^2 $, and denote $ x \times y = x_1 y_2 - x_2 y_1 $ for $ x, y \in \mathbb{R}^2 $, we see for any fixed $ y \in \mathbb{R}^2 $ that $ \widetilde{\varphi}_{y,B} \myp{x} := \sqrt{B} e^{- i \frac{B}{2} y \times x} u_j \myp{ \sqrt{B} \myp{x-y} } $ is an eigenfunction for $ H_{BA^{\perp}} $ corresponding to the eigenvalue $ B \myp{2j +1} $.

\subsection{Coherent states}
Throughout this subsection, $ \hbar $ and $ b $ will denote arbitrary positive numbers, that is, the scaling relations \eqref{eq:scaling}
will not be needed.
For $ f \in L^2 \myp{\mathbb{R}^3} $ we denote by $ f^{\hbar} $ the function
\begin{equation*}
	f^{\hbar} \myp{y} = \hbar^{-\frac{3}{4}} f \myp{ \hbar^{-\frac{1}{2}} y }.
\end{equation*}
\begin{de}
\label{def:coherentstates}
	We fix a normalized $ f \in L^2 \myp{\mathbb{R}^3} $, and for each $ j \in \mathbb{N}_0 $ we fix a normalized eigenfunction $ \varphi_j $ in the $ j $'th Landau level of $ H_{A^{\perp}} $.
	For fixed $ x \in \mathbb{R}^2 $, $ u \in \mathbb{R}^3 $, $ p \in \mathbb{R} $, and $ \hbar, b > 0 $, we define functions $ \varphi_{x,j}^{\hbar,b} $ on $ \mathbb{R}^2 $ and $ f_{x,u,p,j}^{\hbar, b} $ on $ \mathbb{R}^3 $ by
	\begin{equation}
	\label{eq:phidef}
		\varphi_{x,j}^{\hbar,b} \myp{y_{\perp}} = \hbar^{-\frac{1}{2}} b^{\frac{1}{2}} e^{-i \frac{b}{2 \hbar} x \times y_{\perp}}
		\varphi_j \myp{ \hbar^{-\frac{1}{2}} b^{\frac{1}{2}} \myp{y_{\perp}-x} },
	\end{equation}
	and
	\begin{equation}
	\label{eq:fdef}
		f_{x,u,p,j}^{\hbar, b} \myp{y} = \varphi_{x,j}^{\hbar,b} \myp{y_{\perp}} f^{\hbar} \myp{y-u} e^{i \frac{p y_3}{\hbar}},
	\end{equation}
	where $ y_{\perp} = \myp{y_1,y_2} $ denotes the part of $ y $ orthogonal to the magnetic field.
\end{de}
Note that $ \smash{\varphi_{x,j}^{\hbar,b}} $ by construction is an eigenfunction for the operator $ H_{\hbar^{-1}bA^{\perp}} $ corresponding to the eigenvalue $ \hbar^{-1} b \myp{2j+1} $.
The purpose of $ \varphi_{x,j}^{\hbar,b} $ in \eqref{eq:fdef} is to control momentum in the coordinates perpendicular to the magnetic field (fixing the radius of the 2D cyclotron orbit), which also explains why only a one-dimensional momentum variable is present in the exponential factor (cf. the usual coherent states \eqref{eq:coherentstate}).
The presence of $ f^{\hbar} $ in \eqref{eq:fdef} is to localize in position space around $ y $ on a length scale of $ \sqrt{\hbar} $.
Later we will put further assumptions on the function $ f $, but for now it can be any normalized $ L^2 $-function.

We will use \eqref{eq:phidef} and \eqref{eq:fdef} to build Landau level projections and some resolutions of the identity.
Recall that for any normalized function $ v \in L^2 \myp{\mathbb{R}} $ we have a resolution of the identity
\begin{equation}
\label{eq:vres}
	\frac{1}{2\pi} \int_{\mathbb{R}^2} \ketbra{v_{x,p}}{v_{x,p}} \id x \id p = \mathds{1}_{L^2 \myp{\mathbb{R}}},
\end{equation}
where $ v_{x,p} \myp{y} = v \myp{y-x} e^{ipy} $, $ x,p \in \mathbb{R} $.
We will use shortly that if $ u \in L^2 \myp{\mathbb{R}} $ is any other function, then
\begin{align}
	\MoveEqLeft[3]	\frac{1}{2 \pi} \int_{\mathbb{R}^2} \innerp{ \psi }{ v_{x,p} }
	\innerp{ u_{x,p} }{ \psi } \id x \id p \nonumber \\
	&= \frac{1}{2 \pi} \int_{\mathbb{R}^2} \innerp{ u }{ \psi_{-x,-p} } \innerp{ \psi_{-x,-p} }{ v } \id x \id p
	= \innerp{ u}{v } \normt{\psi}{}^2.
\label{eq:vres2}
\end{align}
\begin{lem}
	\label{lem:phires}
	Let $ \Pi_j^{\myp{2}} $ denote the projection onto the $ j $'th Landau level of the operator $ H_{\hbar^{-1}bA^{\perp}} $.
	We have that
	\begin{equation}
	\label{eq:phires}
		\frac{b}{2 \pi \hbar} \sum\limits_{j=0}^{\infty} \int_{\mathbb{R}^2} \ketbra{\varphi_{x,j}^{\hbar,b}}{\varphi_{x,j}^{\hbar,b}} \id x
		= \mathds{1}_{L^2 \myp{\mathbb{R}^2}},
	\end{equation}
	and
	\begin{equation}
	\label{eq:landauproj}
		\frac{b}{2 \pi \hbar} \int_{\mathbb{R}^2} \ketbra{\varphi_{x,j}^{\hbar,b}}{\varphi_{x,j}^{\hbar,b}} \id x 
		= \Pi_j^{\myp{2}}.
	\end{equation}
\end{lem}
Keeping the notation from \cite{LieSolYng-94}, we use the superscript $ \myp{2} $ on the Landau projection to emphasize that we are working on a two-dimensional space (perpendicular to the magnetic field direction).
\begin{proof}
	For $ \varphi \in L^2 \myp{\mathbb{R}^2} $ we denote $ \widetilde{\varphi}_x \myp{y} = e^{-\frac{i}{2} x \times y} \varphi \myp{y-x} $, and recalling the isomorphism \eqref{eq:uniequiv}, we furthermore define a unitary operator $ U := T^{\ast} \mathcal{F}_2 e^{-i \frac{1}{2} \myp{\nonarg}_1 \myp{\nonarg}_2} $.
	Utilizing the usual properties of the Fourier transform, we have for any function $ \varphi $ that
	\begin{align*}
		\MoveEqLeft[3]	\mathcal{F}_2 \myb[\big]{ e^{-\frac{i}{2} \myp{\nonarg}_1 \myp{\nonarg}_2} \widetilde{\varphi}_x } \myp{y_1,\xi} \\
		&= e^{-\frac{i}{2}x_1x_2} \mathcal{F}_2 \myb[\big]{ e^{-\frac{i}{2} \myp{\myp{\nonarg}_1+x_1} \myp{\myp{\nonarg}_2-x_2}} \varphi \myp{\nonarg-x} } \myp{y_1,\xi} \\
		&= e^{-\frac{i}{2}x_1x_2} e^{-i x_2 \xi} \mathcal{F}_2 \myb[\big]{ e^{-\frac{i}{2} \myp{\myp{\nonarg}_1+x_1} \myp{\nonarg}_2} \varphi \myp{\myp{\nonarg}_1-x_1, \myp{\nonarg}_2} } \myp{y_1,\xi} \\
		&= e^{-\frac{i}{2}x_1x_2} e^{-i x_2 \xi} \mathcal{F}_2 \myb[\big]{ e^{-\frac{i}{2} \myp{\nonarg}_1 \myp{\nonarg}_2} \varphi } \myp{y_1-x_1,\xi+x_1},
	\end{align*}
	and so
	\begin{equation}
	\label{eq:ucomm}
		\myp{U \widetilde{\varphi}_x}  \myp{y_1,\xi}
		= e^{-\frac{i}{2}x_1x_2} e^{-i x_2 \xi} \myp{ U \varphi } \myp{y_1,\xi+x_1}.
	\end{equation}
	Introducing the parameter $ \alpha := b/\hbar $ and denoting by $ V_{\alpha} $ the unitary operator given by $ \myp{V_\alpha \psi} \myp{y} := \sqrt{\alpha} \psi \myp{ \sqrt{\alpha} y } $, then $ \varphi_{x,j}^{\hbar,b} \myp{y} = \myp{V_\alpha \varphi_{\sqrt{\alpha} x,j}^{1,1} } \myp{y} $.
	Since $ U \varphi_j $ is an eigenvector corresponding to the $ j $'th eigenvalue of $ \myp{-\frac{\id^2}{\id x_1^2} + x_1^2 } \otimes \mathds{1}_{L^2 \myp{\mathbb{R}}} $, we can write
	\begin{equation*}
		U \varphi_j = \sum\limits_{k=1}^{\infty} c_k f_j \otimes v_k,
	\end{equation*}
	where $ \myp{v_k} $ is any orthonormal basis of $ L^2 \myp{\mathbb{R}} $, and $ \sum_k \abs{c_k}^2 = 1 $.
	Combining this with \eqref{eq:ucomm} and using \eqref{eq:vres2}, we obtain
	\begin{align*}
		\MoveEqLeft[3]	\frac{b}{2 \pi \hbar} \int_{\mathbb{R}^2} \abs[\big]{ \innerp[\big]{\varphi_{-x,j}^{\hbar,b} }{ \psi } }^2 \id x \\
		&= \frac{\alpha}{2 \pi} \int_{\mathbb{R}^2} \abs[\big]{ \innerp[\big]{ U\varphi_j \myp{\myp{\nonarg}_1,\myp{\nonarg}_2-\sqrt{\alpha} x_1} e^{i \sqrt{\alpha}x_2 \myp{\nonarg}_2} }{ U V_{\alpha}^{\ast} \psi } }^2 \id x \nonumber \\
		&= \frac{1}{2 \pi} \int_{\mathbb{R}^2} \sum\limits_{\ell,k} \overline{c_{\ell}} c_k \innerp[\big]{ f_j \otimes \myp{v_{\ell}}_{x_1,x_2} }{ U V_{\alpha}^{\ast} \psi } \innerp[\big]{ U V_{\alpha}^{\ast} \psi }{ f_j \otimes \myp{v_k}_{x_1,x_2} } \id x \nonumber \\
		&= \sum\limits_{\ell,k} \overline{c_{\ell}} c_k \innerp[\big]{ U V_{\alpha}^{\ast} \psi }{ \myp{ \ketbra{f_j}{f_j} \otimes \innerp{ v_{\ell} }{ v_k } \mathds{1}_{L^2 \myp{\mathbb{R}}}} U V_{\alpha}^{\ast} \psi } \nonumber \\
		&= \innerp[\big]{ U V_{\alpha}^{\ast} \psi }{ \myp{ \ketbra{f_j}{f_j} \otimes \mathds{1}_{L^2 \myp{\mathbb{R}}} } U V_{\alpha}^{\ast} \psi }.
	\label{eq:phiproj}
	\end{align*}
	This actually shows \eqref{eq:landauproj}, since $ \Pi_j^{\myp{2}} = V_{\alpha}U^{\ast} \myp{\ketbra{f_j}{f_j} \otimes \mathds{1}_{L^2 \myp{\mathbb{R}}} } UV_{\alpha}^{\ast} $ by the unitary equivalence \eqref{eq:uniequiv}.
	Summing over all $ j $, we also get
	\begin{equation*}
	\label{eq:phires1}
		\frac{b}{2 \pi \hbar} \sum\limits_{j=0}^{\infty} \int_{\mathbb{R}^2} \abs[\big]{ \innerp[\big]{ \varphi_{x,j}^{\hbar,b} }{ \psi }}^2 \id x
		= \normt{U V_{\alpha}^{\ast} \psi}{}^2
		= \normt{\psi}{}^2,
	\end{equation*}
	concluding the proof.
\end{proof}
\begin{de}
	Using the functions from \eqref{eq:fdef}, we define operators on $ L^2 \myp{\mathbb{R}^3} $ by
	\begin{equation}
	\label{eq:coherentdef}
		P_{u,p,j}^{\hbar,b} := \int_{\mathbb{R}^2} \ketbra{f_{x,u,p,j}^{\hbar,b}}{f_{x,u,p,j}^{\hbar,b}} \id x.
	\end{equation}
\end{de}
Applying the lemma above and using \eqref{eq:vres}, it is easy to show the following
\begin{lem}
	\label{lem:pres}
	The $ P_{u,p,j}^{\hbar,b} $ yield a resolution of the identity on $ L^2 \myp{\mathbb{R}^3} $, i.e.,
	\begin{equation*}
		\frac{b}{\myp{2 \pi \hbar}^2} \sum\limits_{j=0}^{\infty} \int_{\mathbb{R}} \int_{\mathbb{R}^3} P_{u,p,j}^{\hbar,b} \id u \id p = \mathds{1}_{L^2 \myp{\mathbb{R}^3}}.
	\end{equation*}
	Furthermore, $ P_{u,p,j}^{\hbar,b} $ is a trace class operator with $ \mathrm{Tr} \myp{ P_{u,p,j}^{\hbar,b} } = 1 $.
\end{lem}
\begin{proof}
	Recall that for $ y \in \mathbb{R}^3 $ we denote by $ y_{\perp} = \myp{y_1,y_2} \in \mathbb{R}^2 $ the coordinates of $ y $ orthogonal to the magnetic field.
	Let $ \psi \in L^2 \myp{\mathbb{R}^3} $ and define an auxiliary function
	\begin{align}
		g_{u,p}^{\hbar} \myp{y_{\perp}}
		:=& \innerp[\big]{ f^{\hbar} \myp{y_\perp-u_{\perp}, \nonarg-u_3} e^{i \frac{p  }{\hbar} \myp{\nonarg}} }{ \psi \myp{y_{\perp}, \nonarg} }_{L^2 \myp{\mathbb{R}}} \nonumber \\
		=&\sqrt{2 \pi} \mathcal{F}_3 \myb[\big]{ \overline{f^{\hbar} \myp{\nonarg-u}} \psi} \myp[\big]{y_{\perp},\frac{p}{\hbar}},
	\label{eq:fourierg}
	\end{align}
	with $ \mathcal{F}_3 $ being the partial Fourier transform in the third variable.
	Using \cref{lem:phires}, we calculate
	\begin{align}
	\label{eq:gproj}
		\innerp{ \psi}{ P_{u,p,j}^{\hbar,b} \psi }
		&= \int_{\mathbb{R}^2} \abs[\big]{ \innerp[\big]{ f_{x,u,p,j}^{\hbar,b} }{ \psi } }^2 \id x \nonumber \\
		&= \int_{\mathbb{R}^2} \abs[\big]{ \innerp[\big]{ \varphi_{x,j}^{\hbar,b} }{ g_{u,p}^{\hbar} } }^2 \id x 
		= \frac{2 \pi \hbar}{b} \innerp[\big]{ g_{u,p}^{\hbar} }{ \Pi_j^{\myp{2}} g_{u,p}^{\hbar} },
	\end{align}
	implying that
	\begin{align*}
		\MoveEqLeft[3] \frac{b}{\myp{2 \pi \hbar}^2} \sum\limits_{j=0}^{\infty} \int_{\mathbb{R}} \int_{\mathbb{R}^3} \innerp[\big]{ \psi }{ P_{u,p,j}^{\hbar,b} \psi } \id u \id p \\
		&= \frac{1}{2 \pi \hbar} \sum\limits_{j=0}^{\infty} \int_{\mathbb{R}} \int_{\mathbb{R}^3} \innerp[\big]{ g_{u,p}^{\hbar} }{ \Pi_j^{\myp{2}} g_{u,p}^{\hbar} } \id u \id p \\
		&= \frac{1}{\hbar} \int_{\mathbb{R}} \int_{\mathbb{R}^3} \int_{\mathbb{R}^2} \abs[\big]{\mathcal{F}_3 \myb[\big]{ \overline{f^{\hbar} \myp{\nonarg-u}} \psi} \myp[\big]{y_{\perp},\frac{p}{\hbar}}}^2 \id y_{\perp} \id u \id p \\
		&= \int_{\mathbb{R}^3} \int_{\mathbb{R}^3} \abs{ f^{\hbar} \myp{y-u} \psi \myp{y} }^2 \id y \id u = \innerp{ \psi }{ \psi }.
	\end{align*}
	
	To calculate the trace of $ P_{u,p,j}^{\hbar,b} $, we take an arbitrary orthonormal basis $ \myp{\psi_{\ell}} $ of $ L^2 \myp{\mathbb{R}^3} $ and use the definition of the coherent states \eqref{eq:phidef} and \eqref{eq:fdef}
	\begin{align*}
		\mathrm{Tr} \myp{ P_{u,p,j}^{\hbar,b} } 
		&= \sum\limits_{\ell=1}^{\infty} \innerp[\big]{ \psi_{\ell} }{ P_{u,p,j}^{\hbar,b} \psi_{\ell} }
		= \int_{\mathbb{R}^2} \normt[\big]{ f_{x,u,p,j}^{\hbar,b} }{2}^2 \id x \\
		&= \int_{\mathbb{R}^2} \int_{\mathbb{R}^3} \frac{b}{\hbar} \abs{ \varphi_j \myp{ \hbar^{-\frac{1}{2}} b^{\frac{1}{2}} \myp{y_{\perp}-x} } }^2 \abs{ f^{\hbar} \myp{y-u} }^2 \id y \id x = 1.
		\qedhere
	\end{align*}
\end{proof}

\subsection{Semi-classical measures on phase space}
Let $ \mathcal{P}_{\pm 1} $ denote the projections onto the spin-up and spin-down components in $ \mathbb{C}^2 $, that is,
\begin{equation*}
	\mathcal{P}_{1} =	\begin{pmatrix}
							1 & 0 \\
							0 & 0
						\end{pmatrix}, 
	\qquad 
	\mathcal{P}_{-1} =	\begin{pmatrix}
							0 & 0 \\
							0 & 1
						\end{pmatrix}.
\end{equation*}
We recall that the phase space is $ \Omega = \mathbb{R}^3 \times \mathbb{R} \times \mathbb{N}_0 \times \Set{\pm 1} $, and that we use the notational convention \eqref{eq:xi}.
We define $ k $-particle semi-classical measures as follows.
\begin{de}
	For $ \Psi_N \in \bigwedge^N L^2 \myp{\mathbb{R}^3 ; \mathbb{C}^2} $ normalized and $ 1 \leq k \leq N $, the $ k $-particle semi-classical measure on $ \Omega^k $ is the measure with density
	\begin{equation}
	\label{eq:semimeas}
		m_{f,\Psi_N}^{\myp{k}} \myp{\xi}
		= \frac{N!}{\myp{N-k}!} \innerp[\Big]{ \Psi_N }{ \myp[\Big]{ \bigotimes_{\ell=1}^k P_{u_{\ell},p_{\ell},j_{\ell}}^{\hbar,b} \mathcal{P}_{s_{\ell}} } \otimes \mathds{1}_{N-k}  \Psi_N }_{L^2 \myp{ \mathbb{R}^{3N};\mathbb{C}^{2^N} } },
	\end{equation}
	where $ \mathds{1}_{N-k} $ is the identity acting on the last $ N-k $ components of $ \Psi_N $.
\end{de}
\begin{rem}
	In \cite{FouLewSol-18}, the semi-classical measures are defined using creation and annihilation operators, and a formula similar to \eqref{eq:semimeas} is a consequence.
	By expanding the coherent state operators $ P_{u,p,j}^{\hbar,b} \mathcal{P}_s $ in a basis of $ L^2 \myp{\mathbb{R}^3; \mathbb{C}^2} $, a similar approach is also possible in our present case with a magnetic field.
	This provides a slightly different proof of \cref{lem:spinmeasureprop} below.
\end{rem}
The semi-classical measures have the following basic properties.
The upper bound in \eqref{eq:spinmeasurebound} below is a manifestation of the Pauli exclusion principle.
\begin{lem}
\label{lem:spinmeasureprop}
	The function $ m_{f,\Psi_N}^{\myp{k}} $ is symmetric on $ \Omega^k $ and satisfies
	\begin{equation}
	\label{eq:spinmeasurebound}
		0 \leq m_{f,\Psi_N}^{\myp{k}} \leq 1,
	\end{equation}
	\begin{equation}
	\label{eq:spincompatibility1}
		\frac{b^k}{\myp{2 \pi \hbar}^{2k}} \int_{\Omega^k} m_{f,\Psi_N}^{\myp{k}} \myp{\xi} \id \xi
		= \frac{N!}{\myp{N-k}!},
	\end{equation}
	and for $ k \geq 2 $,
	\begin{align}
		\frac{b}{\myp{2 \pi \hbar}^2} \int_{\Omega} m_{f,\Psi_N}^{\myp{k}} \myp{\xi_1, \dotsc, \xi_k} \id \xi_k 
		= \myp{N-k+1} m_{f,\Psi_N}^{\myp{k-1}} \myp{\xi_1, \dotsc, \xi_{k-1}}.
	\label{eq:spincompatibility2}
	\end{align}
\end{lem}
\begin{proof}
	We will start out by proving \eqref{eq:spinmeasurebound}, and we will concentrate on the case $ k=1 $, since the proof easily generalizes to $ k \geq 2 $.
	Note that $ 0 \leq \smash{m_{f,\Psi_N}^{\myp{k}}} $ obviously holds, as the $ \smash{P_{u,p,j}^{\hbar,b}} $'s are positive operators.
	Since $ \smash{P_{u,p,j}^{\hbar,b}} $ is trace class, we may write $ P_{u,p,j}^{\hbar,b} = \sum\nolimits_{k} \lambda_k \ketbra{\psi_k}{\psi_k} $, where the $ \psi_k $ constitute an orthonormal basis of $ L^2 \myp{\mathbb{R}^3} $, and $ \sum\nolimits_k \lambda_k = \mathrm{Tr} \myp{ P_{u,p,j}^{\hbar,b}} =1 $.
	Note that for any $ \psi \in L^2 \myp{\mathbb{R}^3} $ we can rewrite, as operators acting on $ \bigwedge^N L^2 \myp{\mathbb{R}^3;\mathbb{C}^2} $,
	\begin{equation*}
		N \myp{\ketbra{\psi}{\psi} \mathcal{P}_s \otimes \mathds{1}_{N-1}}
		= \sum\limits_{k=1}^N \mathds{1}_{k-1} \otimes \ketbra{\psi}{\psi}\mathcal{P}_s \otimes \mathds{1}_{N-k},
	\end{equation*}
	where
	\begin{align*}
		\MoveEqLeft[3]	\myp[\Big]{ \sum\limits_{k=1}^N \mathds{1}_{k-1} \otimes \ketbra{\psi}{\psi}\mathcal{P}_s \otimes \mathds{1}_{N-k} }^2
		= \sum\limits_{k=1}^N \mathds{1}_{k-1} \otimes \ketbra{\psi}{\psi}\mathcal{P}_s \otimes \mathds{1}_{N-k} \\
		&+ 2 \sum\limits_{1\leq k < \ell \leq N} \mathds{1}_{k-1} \otimes \ketbra{\psi}{\psi}\mathcal{P}_s \otimes \mathds{1}_{\ell-k-1} \otimes \ketbra{\psi}{\psi}\mathcal{P}_s \otimes \mathds{1}_{N-\ell}.
	\end{align*}
	Each term in the last sum acts as zero on anti-symmetric functions, implying for any $ \psi \in L^2 \myp{\mathbb{R}^3} $ that $ N \ketbra{\psi}{\psi}\mathcal{P}_s \otimes \mathds{1}_{N-1} $ is an orthogonal projection on $ \bigwedge^N L^2 \myp{\mathbb{R}^3;\mathbb{C}^2} $.
	We arrive at the conclusion that
	\begin{equation*}
		m_{f,\Psi_N}^{\myp{1}} \myp{u,p,j,s}
		= \sum\limits_{k=0}^{\infty} \lambda_k N \innerp[\big]{ \Psi_N }{ \myp{\ketbra{\psi_k}{\psi_k} \mathcal{P}_s \otimes \mathds{1}_{N-1}} \Psi_N }
		\leq 1.
	\end{equation*}
	The result for general $ k $ follows by applying what we have just shown $ k $ times, so \eqref{eq:spinmeasurebound} holds.
	
	The compatibility relation \eqref{eq:spincompatibility2} follows by applying \cref{lem:pres}:
	\begin{align*}
		\MoveEqLeft[3]	\frac{b}{\myp{2 \pi \hbar}^2} \int_{\Omega}	m_{f,\Psi_N}^{\myp{k}} \myp{\xi_1, \dotsc, \xi_k} \id \xi_k \\
		&= \frac{N!}{\myp{N-k}!} \sum\limits_{s_k = \pm 1} \innerp[\Big]{ \Psi_N }{ \myp[\Big]{ \bigotimes_{\ell=1}^{k-1} P_{u_{\ell},p_{\ell},j_{\ell}}^{\hbar,b} \mathcal{P}_{s_{\ell}} } \otimes \mathcal{P}_{s_k} \otimes \mathds{1}_{N-k} \Psi_N } \\
		&= \myp{N-k+1} m_{f,\Psi_N}^{\myp{k-1}} \myp{\xi_1, \dotsc, \xi_{k-1}}.
	\end{align*}
	Finally, \eqref{eq:spincompatibility1} is obtained by repeating this $ k-1 $ more times.
\end{proof}
The next two lemmas assert some particularly nice properties of the semi-classical measures, which will prove to be of great importance
later.
The first one states that the position densities of the measures are close to the position densities \eqref{eq:kpartdensityspin} of the wave function $ \Psi_N $.
\begin{lem}[Position densities]
\label{lem:spinmomentint}
	Let $ \Psi \in \bigwedge^N L^2 \myp{\mathbb{R}^3 ; \mathbb{C}^2} $ be any normalized wave function, and suppose that that $ f $ is a real, $ L^2 $-normalized and even function.
	We have for $ 1 \leq k \leq N $ that
	\begin{equation}
		\frac{b^k}{\myp{2 \pi \hbar}^{2k}} \sum\limits_{j \in \myp{\mathbb{N}_0}^k} \int_{\mathbb{R}^k} m_{f,\Psi}^{\myp{k}} \myp{u,p, j, s} \id p
		= k! \myp[\big]{ \widetilde{\rho}_{\Psi}^{\, \myp{k}} \ast \myp{ \abs{ f^{\hbar} }^2 }^{\otimes k} } \myp{u,s},
	\end{equation}
	where the convolution in the right hand side is the ordinary position space convolution in each spin component of $ \widetilde{\rho}_{\Psi}^{\, \myp{k}} $.
\end{lem}
\begin{proof}
	For notational convenience we introduce an arbitrary $ \Phi \in L^2 \myp{\mathbb{R}^{3N}} $.
	Think of $ \Phi $ as being one of the spin components of $ \Psi $.
	Note first that
	\begin{equation*}
		P_{u_1,p_1,j_1}^{\hbar,b} \otimes \cdots \otimes P_{u_k,p_k,j_k}^{\hbar,b}
		= \int_{\mathbb{R}^{2k}} \ketbra{\otimes_{\ell=1}^k f_{x_\ell,u_\ell,p_\ell,j_\ell}^{\hbar,b}}{\otimes_{\ell=1}^k f_{x_\ell,u_\ell,p_\ell,j_\ell}^{\hbar,b}} \id x,
	\end{equation*}
	and that for each fixed $ y \in \mathbb{R}^{3 \myp{N-k}} $ we have as in \eqref{eq:fourierg} that
	\begin{align*}
		\MoveEqLeft[3]	\innerp[\big]{ \otimes_{\ell=1}^k f_{x_\ell,u_\ell,p_\ell,j_\ell}^{\hbar,b} }{ \Phi \myp{\nonarg,y} }_{L^2 \myp{\mathbb{R}^{3k}}} \\
		&= \myp{2 \pi}^{\frac{k}{2}} \innerp[\big]{ \otimes_{\ell=1}^k \varphi_{x_\ell,j_\ell}^{\hbar,b} }{ \mathcal{F}_3^{\otimes k} \myb[\big]{ \myp{f^{\hbar} }^{\otimes k} \myp{\nonarg- u} \Phi \myp{\nonarg, y} } \myp{ \nonarg, \hbar^{-\frac{1}{2}} p } }_{L^2 \myp{\mathbb{R}^{2k}}}.
	\end{align*}
	Combining these observations and using \cref{lem:phires}, we get
	\begin{align*}
		&\frac{1}{\myp{2 \pi}^k} \sum\limits_{j \in \myp{\mathbb{N}_0}^k} \int_{\mathbb{R}^k} \innerp[\big]{ \Phi }{ \myp{ P_{u_1,p_1,j_1}^{\hbar,b} \otimes \cdots \otimes P_{u_k,p_k,j_k}^{\hbar,b}} \otimes \mathds{1}_{N-k} \Phi } \id p \\
		&= \frac{1}{\myp{2 \pi}^k} \sum\limits_{j \in \myp{\mathbb{N}_0}^k} \int_{\mathbb{R}^k} \int_{\mathbb{R}^{3\myp{N-k}}} \int_{\mathbb{R}^{2k}} \abs[\big]{ \innerp[\big]{ \otimes_{\ell=1}^k f_{x_\ell,u_\ell,p_\ell,j_\ell}^{\hbar,b} }{ \Phi \myp{\cdot,y} } }^2 \id x \id y \id p \\
		&= \frac{\myp{2\pi \hbar}^k}{b^k} \int_{\mathbb{R}^k} \int_{\mathbb{R}^{3\myp{N-k}}} \normt[\big]{\mathcal{F}_3^{\otimes k} \myb[\big]{ \myp{f^{\hbar} }^{\otimes k} \myp{\nonarg- u} \Phi \myp{\nonarg, y} } \myp{ \nonarg, \hbar^{-\frac{1}{2}} p } }{L^2 \myp{\mathbb{R}^{2k}}}^2 \id y \id p \\
		&= \frac{\myp{2 \pi}^k \hbar^{2k}}{b^k} \int_{\mathbb{R}^{3\myp{N-k}}} \normt[\big]{ \myp{f^{\hbar} }^{\otimes k} \myp{\nonarg- u} \Phi \myp{\nonarg, y}}{L^2 \myp{\mathbb{R}^{3k}}}^2 \id y.
	\end{align*}
	Applying this to $ \Psi $ and using that $ f $ is even, we obtain
	\begin{align*}
		\MoveEqLeft[1]	\sum\limits_{j \in \myp{\mathbb{N}_0}^k} \int_{\mathbb{R}^k} m_{f,\Psi}^{\myp{k}} \myp{ u, p, j, s} \id p \\
		& = \sum\limits_{r \in \Set{\pm 1}^N} \sum\limits_{j \in \myp{\mathbb{N}_0}^k} \int_{\mathbb{R}^k} \frac{N!}{\myp{N-k}!} \innerp[\Big]{ \Psi \myp{\nonarg ; r} }{ \myp[\Big]{ \bigotimes_{\ell=1}^{k-1} P_{u_{\ell},p_{\ell},j_{\ell}}^{\hbar,b} \mathcal{P}_{s_{\ell}} } \Psi \myp{\nonarg ; r} } \id p \\
		&= \frac{\myp{2 \pi \hbar}^{2k} N!}{b^k \myp{N-k}!} \sum_{r \in \Set{\pm 1}^{N-k}} \int_{\mathbb{R}^{3\myp{N-k}}} \normt[\big]{ \myp{f^{\hbar} }^{\otimes k} \myp{\nonarg - u} \Psi \myp{ \nonarg, y ; s,r}}{L^2 \myp{\mathbb{R}^{3k}}}^2 \id y \\
		& = \frac{\myp{2 \pi \hbar}^{2k}}{b^k} k! \myp[\big]{ \widetilde{\rho}_{\Psi}^{\, \myp{k}} \ast \myp{ \abs{ f^{\hbar} }^2 }^{\otimes k} } \myp{u,s},
	\end{align*}
	concluding the proof.
\end{proof}

\begin{lem}[Kinetic energy]
\label{lem:spinkinenergy}
	Let $ \Psi \in \bigwedge^N L^2 \myp{\mathbb{R}^3 ; \mathbb{C}^2} $ be normalized, and suppose that $ f \in C_c^{\infty} \myp{\mathbb{R}^3} $ is real-valued, $ L^2 $-normalized and even.
	Then we have
	\begin{align}
	\label{eq:spinkinenergy}
		\MoveEqLeft[3]	\innerp[\Big]{ \Psi }{ \sum\limits_{j=1}^N \myp{ \bm{\sigma} \cdot \myp{-i\hbar \nabla_j + bA \myp{x_j}}}^2 \Psi }
		=- \hbar N \int_{\mathbb{R}^3} \myp{ \nabla f \myp{u} }^2 \id u \nonumber \\
		&+ \frac{b}{\myp{2 \pi \hbar}^2} \sum\limits_{s= \pm 1} \sum\limits_{j=0}^{\infty} \int_{\mathbb{R}} \int_{\mathbb{R}^3} \myp{p^2 + \hbar b \myp{2j+1+s}} m_{f,\Psi}^{\myp{1}} \myp{u,p,j,s} \id u \id p.
	\end{align}
\end{lem}
\begin{proof}
	The assumption that $ f $ is both smooth and compactly supported is far from optimal, but it will be sufficient for our purposes.
	The assertion of the lemma will hold as long as $ f^{\hbar} $ satisfies the following version of the IMS localization formula \cite[equation $ \myp{3.18} $]{LieSolYng-94}
	\begin{equation}
	\label{eq:imsloc}
		\innerp[\big]{ \psi }{ f^{\hbar} \myp{-i \hbar \nabla + b A}^2 f^{\hbar} \psi }
		= \innerp[\big]{ \psi }{ \myp{f^{\hbar} }^2 \myp{-i \hbar \nabla + b A}^2 \psi } + \hbar^2 \innerp[\big]{ \psi }{ \myp{ \nabla f^{\hbar} }^2 \psi }
	\end{equation}
	for any $ \psi $ in the domain of $ \myp{-i \hbar \nabla + b A}^2 $.
	Since $ f $ is normalized, the IMS formula yields
	\begin{align}
		\innerp[\big]{ \psi }{ \myp{-i \hbar \nabla + b A}^2 \psi }
		= {}& \int_{\mathbb{R}^3} \innerp[\big]{ \psi }{ f^{\hbar} \myp{\nonarg-u}^2 \myp{-i \hbar \nabla + b A}^2 \psi } \id u \nonumber \\
		= {}& \int_{\mathbb{R}^3} \innerp[\big]{ \psi }{ f^{\hbar} \myp{\nonarg-u} \myp{-i \hbar \nabla + b A}^2 f^{\hbar} \myp{\cdot-u} \psi } \id u \nonumber \\
		&- \hbar \normt{\psi}{2}^2 \int_{\mathbb{R}^3} \myp{\nabla f \myp{u}}^2 \id u.
	\label{eq:localization}
	\end{align}
	
	Returning to the semi-classical measures, note by \eqref{eq:fourierg} and \eqref{eq:gproj} that
	\begin{align*}
		\MoveEqLeft[3]	\frac{b}{\myp{2 \pi \hbar}^2} \sum\limits_{j=0}^{\infty} \int_{\mathbb{R}} \int_{\mathbb{R}^3} p^2 \innerp[\big]{ \psi }{ P_{u,p,j}^{\hbar,b} \psi } \id u \id p
		= \frac{1}{2 \pi \hbar} \int_{\mathbb{R}} \int_{\mathbb{R}^3} p^2 \innerp[\big]{ g_{u,p}^{\hbar} }{ g_{u,p}^{\hbar} } \id u \id p \\
		&= \frac{1}{\hbar} \int_{\mathbb{R}} \int_{\mathbb{R}^3} \int_{\mathbb{R}^2} p^2 \abs[\big]{ \mathcal{F}_3 \myb[\big]{ f^{\hbar} \myp{\nonarg-u} \psi } \myp{y, \hbar^{-1} p } }^2 \id y \id u \id p \\
		&= \hbar^2 \int_{\mathbb{R}} \int_{\mathbb{R}^3} \int_{\mathbb{R}^2} \abs[\big]{ \mathcal{F}_3 \myb[\big]{ \partial_3^2 \myp{ f^{\hbar} \myp{\nonarg-u} \psi } } \myp{y,p} }^2 \id y \id u \id p \\
		&= \int_{\mathbb{R}^3} \innerp[\big]{ f^{\hbar} \myp{\nonarg-u} \psi }{ -\hbar^2 \partial_3^2 \myp{ f^{\hbar} \myp{\nonarg-u} \psi } } \id u.
	\end{align*}
	Similarly, also using \eqref{eq:fourierg} and \eqref{eq:gproj}, and recalling \eqref{eq:HA}, we get
	\begin{align*}
		\MoveEqLeft[3]	\frac{b}{\myp{2 \pi \hbar}^2} \sum\limits_{j=0}^{\infty} \int_{\mathbb{R}} \int_{\mathbb{R}^3} \hbar b \myp{2j+1} \innerp[\big]{ \psi }{ P_{u,p,j}^{\hbar,b} \psi } \id u \id p \\
		&= \frac{\hbar}{2 \pi} \sum\limits_{j=0}^{\infty} \int_{\mathbb{R}} \int_{\mathbb{R}^3} \innerp[\big]{ g_{u,p}^{\hbar} }{ H_{\hbar^{-1}bA^{\perp}} \Pi_j^{\myp{2}} g_{u,p}^{\hbar} } \id u \id p \\
		&= \frac{\hbar}{2 \pi} \int_{\mathbb{R}} \int_{\mathbb{R}^3} \innerp[\big]{ g_{u,p}^{\hbar} }{ H_{\hbar^{-1}bA^{\perp}} g_{u,p}^{\hbar} } \id u \id p \\
		&= \int_{\mathbb{R}^3} \innerp[\big]{ f^{\hbar} \myp{\nonarg-u} \psi }{ \hbar^2 \myp{ H_{\hbar^{-1}bA^{\perp}} \otimes \mathds{1}_{L^2 \myp{\mathbb{R}}} } \myp{ f^{\hbar} \myp{\nonarg-u} \psi } } \id u.
	\end{align*}
	Since $ \myp{-i \hbar \nabla + bA}^2 = \hbar^2 \myp{ H_{\hbar^{-1}bA^{\perp}} - \partial_3^2} $, combining these with \eqref{eq:localization} yields
	\begin{align*}
		\MoveEqLeft[3]	\innerp[\big]{ \psi }{ (-i \hbar \nabla + b A)^2 \psi }
		= - \hbar \normt{\psi}{2}^2 \int_{\mathbb{R}^3} \myp{\nabla f \myp{u}}^2 \id u \\
		&+ \frac{b}{\myp{2 \pi \hbar}^2} \sum\limits_{j=0}^{\infty} \int_{\mathbb{R}} \int_{\mathbb{R}^3} \myp{ p^2 + \hbar b \myp{2j+1} } \innerp[\big]{ \psi }{ P_{u,p,j}^{\hbar,b} \psi } \id u \id p.
	\end{align*}
	Now, using \eqref{eq:paulieq}, we get for $ \Phi \in L^2 \myp{\mathbb{R}^3; \mathbb{C}^2} $,
	\begin{align*}
		\MoveEqLeft[3]	\innerp[\big]{ \Phi }{ \myp{\bm{\sigma} \cdot \myp{-i \hbar \nabla + bA}}^2 \Phi } \\ 
		={}& \begin{aligned}[t]
			\frac{b}{\myp{2 \pi \hbar}^2} \sum\limits_{j=0}^{\infty} \int_{\mathbb{R}} \int_{\mathbb{R}^3} \myp[\Big]{ &(p^2 + \hbar b \myp{2 j+1} ) \innerp[\big]{ \Phi }{ P_{u,p,j}^{\hbar,b} \mathds{1}_{\mathbb{C}^2} \Phi } \\
			&+ \hbar b \innerp[\big]{ \Phi }{ P_{u,p,j}^{\hbar,b} \sigma_3 \Phi } } \id u \id p 
		\end{aligned} \\
		&- \hbar \normt{\Phi}{L^2 \myp{\mathbb{R}^3; \mathbb{C}^2}}^2 \int_{\mathbb{R}^3} \myp{\nabla f \myp{u}}^2 \id u \\
		={}&  \frac{b}{\myp{2 \pi \hbar}^2} \sum\limits_{s=\pm 1} \sum\limits_{j=0}^{\infty} \int_{\mathbb{R}} \int_{\mathbb{R}^3} (p^2 + \hbar b \myp{2j+1+s} ) \innerp[\big]{ \Phi }{ P_{u,p,j}^{\hbar,b} \mathcal{P}_s \Phi } \id u \id p \\
		&- \hbar \normt{\Phi}{L^2 \myp{\mathbb{R}^3; \mathbb{C}^2}}^2 \int_{\mathbb{R}^3} \myp{\nabla f \myp{u}}^2 \id u.
	\end{align*}
	Applying this to the first component of the wave function $ \Psi $ while keeping all other variables fixed, we finally obtain
	\begin{align*}
		\MoveEqLeft[3]	\innerp[\Big]{ \Psi }{ \sum\limits_{j=1}^N \myp{ \bm{\sigma} \cdot \myp{-i\hbar \nabla_j + bA \myp{x_j}}}^2 \Psi } \\
		={}& N \int_{\myp{\mathbb{R}^3 \times \Set{\pm 1}}^{N-1}} \innerp[\big]{ \Psi \myp{\nonarg ,z} }{ \myp{\bm{\sigma} \cdot \myp{ -i \hbar \nabla + bA}}^2 \Psi \myp{\nonarg ,z} }_{L^2 \myp{\mathbb{R}^3; \mathbb{C}^2}} \id z \\
		={}& \frac{b}{\myp{2 \pi \hbar}^2} \sum\limits_{s=\pm 1} \sum\limits_{j=0}^{\infty} \int_{\mathbb{R}} \int_{\mathbb{R}^3} \myp{p^2 + \hbar b \myp{2j+1+s}} N \innerp[\big]{ \Psi }{ P_{u,p,j}^{\hbar,b} \mathcal{P}_{s} \Psi } \id u \id p \\
		& - \hbar N \int_{\mathbb{R}^3} \myp{\nabla f \myp{u}}^2 \id u,
	\end{align*}
	finishing the proof.
\end{proof}

\subsection{Limiting measures, strong magnetic fields}
\label{sec:limitmeasures}
We now fix a real-valued, even and normalized function $ f \in L^2 \myp{\mathbb{R}^3} $, along with a sequence $ \myp{\Psi_N}_{N\geq1} $ of normalized functions with $ \Psi_N \in \bigwedge^N L^2 \myp{\mathbb{R}^3; \mathbb{C}^2} $ for
each $ N $.
We investigate the measures $ \smash{m_{f,\Psi_N}^{\myp{k}}} $ in the limit as $ N $ tends to infinity, when $ \beta_N $ is a sequence with $ \beta_N \to \beta $, $ 0 < \beta \leq \infty $.
This corresponds to the regime where the distance between the Landau bands of the Pauli operator remains bounded from below.
\begin{lem}
\label{lem:measurelim}
	For each $ k \geq 1 $ there is a symmetric function $ m_f^{\myp{k}} \in L^1 \myp{\Omega^k} \cap L^{\infty} \myp{\Omega^k} $ with $ 0\leq \smash{m_f^{\myp{k}}} \leq 1 $ such that, along a common (not displayed) subsequence in $ N $,
	\begin{equation}
	\label{eq:measurelim}
		\int_{\Omega^k} m_{f,\Psi_N}^{\myp{k}} \myp{\xi} \varphi \myp{\xi} \id \xi
		\longrightarrow \int_{\Omega^k} m_{f}^{\myp{k}} \myp{\xi} \varphi \myp{\xi} \id \xi
	\end{equation}
	for all $ \varphi \in L^1 \myp{\Omega^k} + L_{\varepsilon}^{\infty} \myp{\Omega^k} $, as $ N $ tends to infinity.
\end{lem}
The proof of this lemma is a standard exercise in functional analysis, using the boundedness of the sequence $ \myp{m_{f,\Psi_N}^{\myp{k}} }_{N\geq k} $ both in $ L^1 \myp{\Omega^k} $ and in $ L^{\infty} \myp{\Omega^k} $, and we leave the details to the reader.

If the sequence of measures $ \myp{ m_{f,\Psi_N}^{\myp{k}} }_{N\geq k} $ is \emph{tight}, that is, if
\begin{equation*}
	\lim\limits_{R \to \infty} \limsup\limits_{N\to \infty} \int_{\abs{\xi_1}+ \cdots + \abs{\xi_k}\geq R} m_{f,\Psi_N}^{\myp{k}} \myp{\xi} \id \xi =0.
\end{equation*}
then all the properties of the measures in \cref{lem:spinmeasureprop} carry over to the limit, and the weak convergence in \cref{lem:measurelim} is strengthened.
We collect these observations in the lemma below, but the proof (which is elementary) will be omitted.
The key ingredient for the proof is the fact that
\begin{equation*}
	\int_{\abs{\xi} \leq R} m_{f,\Psi_N}^{\myp{k}} \myp{\xi} \id \xi
	\xrightarrow{R \to \infty} \int_{\Omega^k} m_{f,\Psi_N}^{\myp{k}} \myp{\xi} \id \xi
	= \frac{\myp{2 \pi \hbar}^{2k}}{b^k} \frac{N!}{\myp{N-k}!}
\end{equation*}
uniformly in $ N $ as $ R $ tends to infinity, whenever $ \myp{ m_{f,\Psi_N}^{\myp{k}} }_{N\geq k} $ is tight.
\begin{lem}
\label{lem:limmeasures}
	Suppose that $ \myp{ m_{f,\Psi_N}^{\myp{1}} }_{N \in \mathbb{N}} $ is a tight sequence.
	Then we have
	\begin{enumerate}
		\item $ \myp{ m_{f,\Psi_N}^{\myp{k}} }_{N\geq k} $ is also tight for each $ k \geq 1 $.
		
		\item The limit measures $ m_f^{\myp{k}} $ are probability measures.
		More precisely,
		\begin{equation}
		\label{eq:measurelimprob}
			\frac{1}{\myp{2 \pi}^{2k}} \frac{\beta^k}{\myp{1+\beta}^k} \int_{\Omega^k} m_f^{\myp{k}} \myp{\xi} \id \xi = 1.
		\end{equation}
		
		\item The compatibility relation \eqref{eq:spincompatibility2} is preserved in the limit, that is, for $ k \geq 2 $ and almost every $ \xi \in \Omega^{k-1} $,
		\begin{equation}
		\label{eq:measurelimcomp}
			\frac{1}{\myp{2 \pi}^2} \frac{\beta}{1+\beta} \int_{\Omega} m_f^{\myp{k}} \myp{\xi,\xi_k} \id \xi_k
			= m_f^{\myp{k-1}} \myp{\xi}.
		\end{equation}
		
		\item The convergence in \eqref{eq:measurelim} holds on all of $ L^1 \myp{\Omega^k} + L^{\infty} \myp{\Omega^k} $.
	\end{enumerate}
\end{lem}

We now formulate the de Finetti theorem which serves as the main abstract tool in our proof of the lower bound of the energy in \cref{thm:energylowbound}.
The version of the theorem below is essentially \cite[Theorem $ 2.6 $]{FouLewSol-18}.
For some additional details, see e.g. \cite{MadThese}.
\begin{theo}[de Finetti]
\label{thm:finetti}
	Let $ M \subseteq \Omega $ be a locally compact subset, and $ m^{\myp{k}} \in L^1 \myp{M^k} $ a family of symmetric positive densities satisfying for some $ c > 0 $ and all $ k \geq 1 $ that $ 0 \leq m^{\myp{k}} \leq 1 $, and
	\begin{equation*}
		c \int_M m^{\myp{k}} \myp{\xi_1, \dotsc, \xi_k} \id \xi_k = m^{\myp{k-1}} \myp{\xi_1, \dotsc, \xi_{k-1}}
	\end{equation*}
	with $ m^{\myp{0}} = 1 $.
	Then there exists a unique Borel probability measure $ \mathrm{P} $ on the set
	\begin{equation*}
		\mathcal{S} = \Set[\big]{\mu \in L^1 \myp{M} \given 0 \leq \mu \leq 1, \quad c \int_M \mu \myp{\xi} \id \xi = 1}
	\end{equation*}
	such that for all $ k \geq 1 $, in the sense of measures,
	\begin{equation}
	\label{eq:definetti}
		m^{\myp{k}} = \int_{\mathcal{S}} \mu^{\otimes k} \id \mathrm{P} \myp{\mu}.
	\end{equation}
\end{theo}

%%%%%%%%%%%%%%%%%%%%%%%%%%%%%%%%%%%%%%%%%%%%%%%%%%%%%
%%%%%%%%%%%%%%%%%%%%%%%%%%%%%%%%%%%%%%%%%%%%%%%%%%%%%
\section{Lower energy bounds, strong fields}
\label{sec:lowerbound}
%%%%%%%%%%%%%%%%%%%%%%%%%%%%%%%%%%%%%%%%%%%%%%%%%%%%%
%%%%%%%%%%%%%%%%%%%%%%%%%%%%%%%%%%%%%%%%%%%%%%%%%%%%%
Throughout this section we suppose that the potentials $ V $ and $ w $ satisfy the assumptions of \cref{thm:energylowbound}, and that $ \myp{\beta_N} $ is a sequence satisfying $ \beta_N \to \beta $ with $ 0 < \beta \leq \infty $ and \eqref{eq:scaling2}.
We further assume that the auxiliary function $ f $ is smooth and compactly supported.
\begin{lem}
\label{lem:tightness}
	Suppose that $ \Psi_N \in \bigwedge^N L^2 \myp{\mathbb{R}^3; \mathbb{C}^2} $ is a sequence satisfying the energy bound $ \innerp{ \Psi_N }{ H_{N,\beta_N} \Psi_N } \leq C N $.
	Then the corresponding semi-classical measures $ \myp{ \smash{m_{f,\Psi_N}^{\myp{k}}} }_{N\geq k} $ are tight.
\end{lem}
\begin{proof}
	The proof is straightforward, and we only outline it.
	One first uses the energy bound combined with \eqref{eq:scaling2} and the Lieb-Thirring bound \eqref{eq:kinbound} to conclude that $ \frac{1}{N} \smash{\widetilde{\rho}_{\Psi_N}^{\, \myp{1}}} $ is a tight sequence.
	It is essential at this point that $ V $ is a confining potential.
	Then, applying \cref{lem:spinmomentint} and the fact that $ f $ is well localized, it follows that $ \myp{ m_{f,\Psi_N}^{\myp{1}} }_{N\geq 1} $ is tight in the position variable.
	
	On the other hand, using \eqref{eq:kinbound} to bound the kinetic energy and then combining with the expression for the kinetic energy from \cref{lem:spinkinenergy}, it follows that $ \myp{ m_{f,\Psi_N}^{\myp{1}}}_{N\geq 1} $ is also tight in the momentum variables $ \myp{p,j} \in \mathbb{R} \times \mathbb{N}_0 $.
	Now by \cref{lem:limmeasures}, the sequences $ \myp{ \smash{m_{f,\Psi_N}^{\myp{k}}} }_{N\geq k} $ are all tight for $ k \geq 1 $.
\end{proof}
We once again remind the reader of the notational convention \eqref{eq:xi}.
\begin{prop}[Convergence of states]
\label{lem:stateconv}
	Let $ \Psi_N \in \bigwedge^N L^2 \myp{\mathbb{R}^3; \mathbb{C}^2} $ be a sequence satisfying the energy bound $ \innerp{ \Psi_N }{ H_{N,\beta_N} \Psi_N } \leq C N $.
	Then there exist a subsequence $ \myp{N_\ell} \subseteq \mathbb{N} $ and a unique Borel probability measure $ \mathrm{P} $ on the set
	\begin{equation*}
		\mathcal{S} = \Set[\big]{\mu \in L^1 \myp{\Omega} \given 0 \leq \mu \leq 1, \ \frac{1}{\myp{2 \pi}^2} \frac{\beta}{1+\beta} \int_{\Omega} \mu \myp{\xi} \id \xi = 1},
	\end{equation*}
	such that for each $ k \geq 1 $ the following holds:
	\begin{enumerate}
		\item For all
		$ \varphi \in L^1 \myp{\Omega^k} + L^{\infty} \myp{\Omega^k} $,
		\begin{equation}
		\label{eq:stateconv1}
			\int_{\Omega^k} m_{f,\Psi_{N_\ell}}^{\myp{k}} \myp{\xi} \varphi \myp{\xi} \id \xi
			\longrightarrow \int_{\mathcal{S}} \myp[\Big]{ \int_{\Omega^k} \mu^{\otimes k} \myp{\xi} \varphi \myp{\xi} \id \xi} \id \mathrm{P} \myp{\mu}.
		\end{equation}
		as $ \ell $ tends to infinity.
		
		\item For $ U \in L^{5/2} \myp{\mathbb{R}^3 \times \Set{\pm 1}} + L^{\infty} \myp{\mathbb{R}^3 \times \Set{\pm 1}} $ if $ k=1 $, and for any bounded and uniformly continuous function $ U $ on $ \myp{ \mathbb{R}^3 \times \Set{\pm 1} }^{k} $ if $ k \geq 2 $, as $ \ell $ tends to infinity,
		\begin{align}
		\label{eq:stateconv2}
			\MoveEqLeft[3]	\frac{k!}{N_\ell^k} \sum\limits_{s \in \Set{\pm 1}^k} \int_{\mathbb{R}^{3k}} \widetilde{\rho}_{\Psi_{N_\ell}}^{\, \myp{k}} \myp{x,s} U \myp{x,s} \id x \nonumber \\
			&\longrightarrow \int_{\mathcal{S}} \myp[\Big]{ \sum\limits_{s \in \Set{\pm 1}^k} \int_{\mathbb{R}^{3k}} \rho_{\mu}^{\otimes k} \myp{x,s} U \myp{x,s} \id x } \id \mathrm{P} \myp{\mu},
		\end{align}
		where $ \rho_{\mu} $ is the position density
		\begin{equation*}
			\rho_{\mu} \myp{x,s} = \frac{1}{\myp{2 \pi}^2} \frac{\beta}{1+\beta} \sum\limits_{j=0}^{\infty} \int_{\mathbb{R}} \mu \myp{x,p,j,s} \id p.
		\end{equation*}
	\end{enumerate}
\end{prop}
\begin{proof}
	Consider the subsequence $ \myp{N_\ell} $ along with the limit measures $ \myp{ m_f^{\myp{k}} } $ from \cref{lem:measurelim}.
	Throughout the proof, we will suppress the subsequence from the notation.
	By \cref{lem:tightness}, the measures $ \myp{ m_{f,\Psi_{N}}^{\myp{k}} }_{N \geq k} $ are tight, the limit measures are ensured by \cref{lem:limmeasures} to satisfy the compatibility relation
	\begin{equation*}
		c \int_{\Omega} m_f^{\myp{k}} \myp{\xi_1, \dotsc, \xi_k} \id \xi_k = m_f^{\myp{k-1}} \myp{\xi_1, \dotsc, \xi_{k-1}},
	\end{equation*}
	with $ c = \frac{1}{\myp{2 \pi}^2} \frac{\beta}{1+\beta} $.
	Hence by the de Finetti \cref{thm:finetti} we have a unique Borel probability measure $ \mathrm{P} $ on $ \mathcal{S} $ such that
	\begin{equation*}
		m_f^{\myp{k}} = \int_{\mathcal{S}} \mu^{\otimes k} \id \mathrm{P} \myp{\mu}.
	\end{equation*}
	It follows that \eqref{eq:stateconv1} holds, since
	\begin{equation*}
	\label{eq:pintegration}
		\int_{\Omega^k} m_f^{\myp{k}} \myp{\xi} \varphi \myp{\xi} \id \xi 
		= \int_{\mathcal{S}} \myp[\Big]{ \int_{\Omega^k} \mu^{\otimes k} \myp{\xi} \varphi \myp{\xi} \id \xi } \id \mathrm{P} \myp{\mu}
	\end{equation*}
	for each $ \varphi \in L^1 \myp{\Omega^k} + L^{\infty} \myp{\Omega^k} $, by definition of the measure $ \int_{\mathcal{S}} \mu^{\otimes k} \id \mathrm{P} \myp{\mu} $.
	
	Now, if $ U $ is a bounded function on $ \myp{\mathbb{R}^3 \times \Set{\pm 1}}^{k} $, we define $ \varphi \in L^{\infty} \myp{\Omega^k} $ by $ \varphi \myp{x,p,j,s} := U \myp{x,s} $.
	Then by \eqref{eq:stateconv1} we have as $ N $ tends to infinity,
	\begin{align}
		\MoveEqLeft[3]	\sum\limits_{s \in \Set{\pm 1}^k} \int_{\mathbb{R}^{3k}} \rho_{m_{f,\Psi_{N}}^{\myp{k}}} \myp{x,s} U \myp{x,s} \id x \nonumber \\
		&\longrightarrow \frac{1}{\myp{2 \pi}^{2k}} \frac{\beta^k}{\myp{1+\beta}^k} \int_{\mathcal{S}} \myp[\Big]{ \int_{\Omega^k} \mu^{\otimes k} \myp{\xi} \varphi \myp{\xi} \id \xi} \id \mathrm{P} \myp{\mu} \nonumber \\
		&= \int_{\mathcal{S}} \myp[\Big]{ \sum\limits_{s \in \Set{\pm 1}^k} \int_{\mathbb{R}^{3k}} \rho_{\mu}^{\otimes k} \myp{x,s} U \myp{x,s} \id x } \id \mathrm{P} \myp{\mu},
	\end{align}
	so in order to show \eqref{eq:stateconv2} it suffices to see that $ \frac{k!}{N^k} \widetilde{\rho}_{\Psi_{N}}^{\, \myp{k}} $ has the same weak limit as $ \rho_{m_{f,\Psi_{N}}^{\myp{k}}} $ on the set of bounded, uniformly continuous functions on $ \myp{ \mathbb{R}^3 \times \Set{\pm 1} }^{k} $.
	However, by \cref{lem:spinmomentint},
	\begin{equation*}
		\rho_{m_{f,\Psi_{N}}^{\myp{k}}} \myp{x,s} 
		= \frac{\beta^k}{\myp{1+\beta}^k} \frac{\hbar^{2k} N^k}{b^k} \frac{k!}{N^k} \myp[\big]{ \widetilde{\rho}_{\Psi_N}^{\, \myp{k}} \ast \myp{ \abs{ f^{\hbar} }^2 }^{\otimes k} } \myp{x,s},
	\end{equation*}
	where $ \frac{\beta^k}{\myp{1+\beta}^k} \frac{\hbar^{2k} N^k}{b^k} \to 1 $ when $ N \to \infty $, so it suffices to show that $ \frac{k!}{N^k} \widetilde{\rho}_{\Psi_{N}}^{\, \myp{k}} $ and $ \frac{k!}{N^k} \widetilde{\rho}_{\Psi_N}^{\, \myp{k}} \ast \myp{ \abs{f^{\hbar}}^2 }^{\otimes k} $ have the same weak limit.
	However, this follows easily from the boundedness of $ \frac{k!}{N^k} \widetilde{\rho}_{\Psi_{N}}^{\, \myp{k}} $ in
	$ L^1 \myp{( \mathbb{R}^3 \times \Set{\pm 1} )^{k}} $, and from the fact that $ \lim_{\hbar \to 0} \normt{ U - U \ast \myp{ \abs{ f^{\hbar} }^2 }^{\otimes k} }{\infty} = 0 $ whenever $ U $ is a uniformly continuous and bounded function on $ ( \mathbb{R}^3 \times \Set{\pm 1} )^{k} $.
	
	For $ k=1 $ we appeal to \cref{lem:densitybound} and the tightness of $ \frac{1}{N} \widetilde{\rho}_{\Psi_N}^{\, \myp{1}} $ to obtain convergence for test functions $ U \in L^{5/2} \myp{\mathbb{R}^3 \times \Set{\pm 1}} + L^{\infty} \myp{\mathbb{R}^3 \times \Set{\pm 1}} $.
\end{proof}
In the case where $ \beta_N \to \infty $ we can further refine the assertions of \cref{lem:stateconv} above.
\begin{cor}[Convergence of states, strong field regime]
\label{cor:stateconvstrong}
	Suppose that $ \myp{\beta_N} $ satisfies $ \beta_N \to \infty $ and \eqref{eq:scaling2}, and that $ \Psi_N \in \bigwedge^N L^2 \myp{\mathbb{R}^3; \mathbb{C}^2} $ is a sequence satisfying the energy bound $ \innerp{ \Psi_N }{ H_{N,\beta_N} \Psi_N } \leq C N $.
	Then the measure $ \mathrm{P} $ from \cref{lem:stateconv} is supported on the set
	\begin{equation*}
		\widetilde{\mathcal{S}} = \Set[\big]{\mu \in L^1 \myp{\mathbb{R}^3 \times \mathbb{R}} \given 0 \leq \mu \leq 1, \ \frac{1}{\myp{2 \pi}^2} \iint_{\mathbb{R}^3 \times \mathbb{R}} \mu \myp{x,p} \id x \id p = 1},
	\end{equation*}
	where each $ \widetilde{\mu} \in \widetilde{\mathcal{S}} $ is identified with a density $ \mu \in \mathcal{S} $ by
	\begin{equation*}
		\mu \myp{x,p,j,s} =  \begin{cases} \widetilde{\mu} \myp{x,p}, & \text{if } j=0 \text{ and } s=-1, \\ 0, & \text{otherwise,} \end{cases}
	\end{equation*}
	and for each $ k \geq 1 $ the following holds:
	\begin{enumerate}
		\item For all $ \varphi \in L^1 \myp{\Omega^{k}} + L^{\infty} \myp{\Omega^{k}} $, as $ \ell $ tends to infinity,
		\begin{align}
			\MoveEqLeft[3]	\int_{\Omega^k} m_{f,\Psi_{N_\ell}}^{\myp{k}} \myp{\xi} \varphi \myp{\xi} \id \xi \nonumber \\
			&\longrightarrow \int_{\widetilde{\mathcal{S}}} \myp[\Big]{ \int_{\mathbb{R}^{4k}} \mu^{\otimes k} \myp{x,p} \varphi \myp{ x,p,0^{\times k},\myp{-1}^{\times k} } \id x \id p} \id \mathrm{P} \myp{\mu},
		\label{eq:stateconvstrong1}
		\end{align}
		where $ 0^{\times k} $ and $ \myp{-1}^{\times k} $ are the $ k $-dimensional vectors whose entries are all equal to $ 0 $ and $ -1 $, respectively.
		
		\item For $ U \in L^{5/2} \myp{\mathbb{R}^3 \times \Set{\pm 1}} + L^{\infty} \myp{\mathbb{R}^3 \times \Set{\pm 1}} $ if $ k=1 $, and for any bounded and uniformly continuous function $ U $ on $ \myp{\mathbb{R}^3 \times \Set{\pm 1}}^{k} $ if $ k \geq 2 $, as $ \ell $ tends to infinity,
		\begin{align}
			\MoveEqLeft[3]	\frac{k!}{N_\ell^k} \sum\limits_{s \in \Set{\pm 1}^k} \int_{\mathbb{R}^{3k}} \widetilde{\rho}_{\Psi_{N_\ell}}^{\, \myp{k}} \myp{x,s} U \myp{x,s} \id x \nonumber \\
			&\longrightarrow \int_{\widetilde{\mathcal{S}}} \myp[\Big]{ \int_{\mathbb{R}^{3k}} \rho_{\mu}^{\otimes k} \myp{x} U \myp{ x,\myp{-1}^{\times k} } \id x} \id \mathrm{P} \myp{\mu},
		\label{eq:stateconvstrong2}
		\end{align}
		where
		\begin{equation*}
			\rho_{\mu} \myp{x} = \frac{1}{\myp{2 \pi}^2} \int_{\mathbb{R}} \mu \myp{x,p} \id p.
		\end{equation*}
	\end{enumerate}
\end{cor}
\begin{proof}
	Using that $ \hbar b \to \infty $ since $ \beta_N \to \infty $, along with the expression for the kinetic energy in \cref{lem:spinkinenergy}, the Lieb-Thirring bound \eqref{eq:kinbound}, and the energy bound from the assumptions, we obtain in particular for any $ n \in \mathbb{N} $ that
	\begin{align*}
		\MoveEqLeft[3]	\sum\limits_{s= \pm 1} \sum\limits_{j = 0}^{n} \int_{\mathbb{R}^3} \int_{\mathbb{R}} \myp{j+ 1+s} m_{f}^{\myp{1}} \myp{x,p,j,s} \id p \id x \\
		&=\lim_{N \to \infty} \sum\limits_{s= \pm 1} \sum\limits_{j = 0}^{n} \int_{\mathbb{R}^3} \int_{\mathbb{R}} \myp{j+ 1+s} m_{f,\Psi_N}^{\myp{1}} \myp{x,p,j,s} \id p \id x =0
	\end{align*}
	implying that $ m_f^{\myp{1}} \myp{x,p,j,s} = 0 $ unless $ j=0 $ and $ s=-1 $.  It follows that
	\begin{align*}
		\MoveEqLeft[3]	\int_{\mathcal{S}} \myp[\Big]{ \int_{\mathbb{R}^3 \times \mathbb{R} \times \Set{0} \times \Set{-1}} \mu \myp{\xi} \id \xi } \id \mathrm{P} \myp{\mu} \\
		&= \int_{\Omega} m_f^{\myp{1}} \myp{\xi} \id \xi 
		= \int_{\mathcal{S}} \myp[\Big]{ \int_{\Omega} \mu \myp{\xi} \id \xi} \id \mathrm{P} \myp{\mu},
	\end{align*}
	so for $ \mathrm{P} $-almost every $ \mu \in \mathcal{S} $, we have
	\begin{equation*}
		\int_{\mathbb{R}^3 \times \mathbb{R} \times \Set{0} \times \Set{-1}} \mu \myp{\xi} \id \xi
		=\int_{\Omega} \mu \myp{\xi} \id \xi,
	\end{equation*}
	and hence $ \mathrm{P} $ is supported on $ \widetilde{\mathcal{S}} $.
	The rest of the corollary follows directly from Lemma \ref{lem:stateconv}.
\end{proof}

We now finally have the tools to give a proof of the lower bounds in \cref{thm:energylowbound} in the case when $ \beta_N \to \beta \in \left(0,\infty \right] $.
The proof will be split into a few lemmas, each giving a lower bound on part of the energy.
Note that if $ \Psi_N \in \bigwedge^N L^2 \myp{\mathbb{R}^3; \mathbb{C}^2} $ is a sequence of fermionic wave functions satisfying $ \innerp{ \Psi_N }{ H_{N,\beta_N} \Psi_N } = E \myp{N,\beta_N} + o\myp{N} $, then by the upper energy bound of \cref{prop:energyupbound} we have $ \innerp{ \Psi_N }{ H_{N, \beta_N} \Psi_N } \leq CN $, so that \cref{lem:stateconv} and \cref{cor:stateconvstrong} are applicable.
\begin{lem}
\label{lem:energykinlowbound}
	Suppose that the assumptions in \cref{thm:energylowbound} are satisfied, and that we have a sequence $ \Psi_N \in \bigwedge^N L^2 \myp{\mathbb{R}^3; \mathbb{C}^2} $ with $ \innerp{ \Psi_N }{ H_{N,\beta_N} \Psi_N } = E \myp{N,\beta_N} + o\myp{N} $.
	\begin{enumerate}
		\item If $ \beta_N \to \beta \in \myp{0,\infty} $, then, with $ \mathcal{S} $ as in \cref{lem:stateconv},
		\begin{align}
		\label{eq:energykinlowbound}
			C &\geq \liminf_{N \to \infty} \innerp[\Big]{ \Psi_N }{ \frac{1}{N} \sum\limits_{j=1}^N \myp{\bm{\sigma} \cdot \myp{-i \hbar \nabla_j + b A \myp{x_j}}}^2 \Psi_N } \nonumber \\
			&\geq \frac{1}{\myp{2 \pi}^2} \frac{\beta}{1 + \beta} \int_{\mathcal{S}} \myp[\Big]{ \int_{\Omega} \myp{p^2+ k_{\beta} \myp{2j+ 1+s} } \mu \myp{\xi} \id \xi } \id \mathrm{P} \myp{\mu}.
		\end{align}
		
		\item If $ \beta_N \to \infty $, then, with $ \widetilde{\mathcal{S}} $ as in \cref{cor:stateconvstrong},
		\begin{align}
		\label{eq:energykinlowboundstrong}
			C &\geq \liminf_{N \to \infty} \innerp[\Big]{ \Psi_N }{ \frac{1}{N} \sum\limits_{j=1}^N \myp{\bm{\sigma} \cdot \myp{-i \hbar \nabla_j + b A \myp{x_j}}}^2 \Psi_N } \nonumber \\
			&\geq \frac{1}{\myp{2 \pi}^2} \int_{\widetilde{\mathcal{S}}} \myp[\Big]{ \int_{\mathbb{R}^3 \times \mathbb{R}} p^2 \mu \myp{x,p} \id p \id x } \id \mathrm{P} \myp{\mu}.
		\end{align}
	\end{enumerate}
\end{lem}
\begin{proof}
	Suppose first that $ \beta_N \to \beta < \infty $.
	By Lemma \ref{lem:ltbounds} the kinetic energy per particle is bounded, so applying \cref{lem:spinkinenergy,lem:stateconv} we obtain for any positive $ R $,
	\begin{align*}
		C &\geq \liminf_{N \to \infty} \innerp[\Big]{ \Psi_N }{ \frac{1}{N} \sum\limits_{j=1}^N \myp{\bm{\sigma} \cdot \myp{-i \hbar \nabla_j + b A \myp{x_j}}}^2 \Psi_N } \\
		&\geq \liminf_{N \to \infty} \frac{1}{\myp{2 \pi}^2} \frac{b}{\hbar^2 N} \int_{\abs{u}+\abs{p}+j \leq R} \myp{p^2+ \hbar b \myp{2j+ 1+s}} m_{f,\Psi_N}^{\myp{1}} \myp{\xi} \id \xi \\
		&= \frac{1}{\myp{2 \pi}^2} \frac{\beta}{1 + \beta} \int_{\mathcal{S}} \myp[\Big]{ \int_{\abs{u}+\abs{p}+j \leq R} \myp{p^2+ k_{\beta} \myp{2j+ 1+s}} \mu \myp{\xi} \id \xi } \id \mathrm{P} \myp{\mu}.
	\end{align*}
	Taking $ R \to \infty $, monotone convergence implies \eqref{eq:energykinlowbound}.
	
	The bound \eqref{eq:energykinlowboundstrong} follows in exactly the same way by simply discarding the term $ \hbar b \myp{2j + 1+s} $ in the integrand above, and applying \cref{cor:stateconvstrong}.
\end{proof}
\begin{lem}
\label{lem:energypotlowbound}
	Suppose that $ \myp{\beta_N} $ satisfies \eqref{eq:scaling2} and $ \beta_N \to \beta \in \left(0,\infty \right] $.
	With the assumptions in \cref{thm:energylowbound} and a sequence $ \Psi_N \in \bigwedge^N L^2 \myp{\mathbb{R}^3; \mathbb{C}^2} $ satisfying $ \innerp{ \Psi_N }{ H_{N,\beta_N} \Psi_N } = E \myp{N,\beta_N} + o\myp{N} $, we have
	\begin{equation}
	\label{eq:energypotlowbound}
		\liminf_{N \to \infty} \innerp[\Big]{ \Psi_N }{ \frac{1}{N} \sum\limits_{j=1}^{N} V \myp{x_j} \Psi_N } 
		\geq \int_{\mathcal{S}} \myp[\Big]{ \int_{\mathbb{R}^3} V \myp{x} \rho_{\mu} \myp{x} \id x} \id \mathrm{P} \myp{\mu},
	\end{equation}
	and
	\begin{align}
		\MoveEqLeft[3]	\lim_{N \to \infty} \innerp[\Big]{ \Psi_N }{ \frac{1}{N^2} \sum\limits_{1 \leq j < k \leq N} w \myp{x_j-x_k} \Psi_N } \nonumber \\
		&= \frac{1}{2} \int_{\mathcal{S}} \myp[\Big]{ \int_{\mathbb{R}^6} w \myp{x-y} \rho_{\mu} \myp{x} \rho_{\mu} \myp{y} \id x \id y} \id \mathrm{P} \myp{\mu}.
	\label{eq:energyinterlowbound}
	\end{align}
\end{lem}
\begin{proof}
	By \cref{lem:ltbounds} the potential energy per particle is bounded, so for $ R > 0 $ large enough we have by the weak convergence of $ \frac{1}{N} \rho_{\Psi_N}^{\myp{1}} $ in $ L^{5/2} \myp{\mathbb{R}^3} $ that
	\begin{align*}
		C &\geq \liminf_{N \to \infty} \innerp[\Big]{ \Psi_N }{ \frac{1}{N} \sum\limits_{j=1}^{N} V \myp{x_j} \Psi_N } \\
		&\geq \liminf_{N \to \infty} \frac{1}{N} \int_{\abs{x} \leq R} V \myp{x} \rho_{\Psi_N}^{\myp{1}} \myp{x} \id x 
		= \int_{\mathcal{S}} \myp[\Big]{ \int_{\abs{x} \leq R} V \myp{x} \rho_{\mu} \myp{x} \id x } \id \mathrm{P} \myp{\mu}.
	\end{align*}
	Taking $ R \to \infty $ yields \eqref{eq:energypotlowbound} by the monotone convergence theorem.
	
	For the interaction part, write $ w = w_1+w_2 \in L^{3/2} \myp{\mathbb{R}^3} \cap L^{5/2} \myp{\mathbb{R}^3} + L_{\varepsilon}^{\infty} \myp{\mathbb{R}^3} $ and approximate $ w_1 $ in $ L^{3/2} \myp{\mathbb{R}^3} $ and $ L^{5/2} \myp{\mathbb{R}^3} $ by some $ w_0 \in C_c \myp{\mathbb{R}^3} $.
	By the Lieb-Thirring estimate \eqref{eq:ltdensity}, we have
	\begin{align}
		\MoveEqLeft[3]	\abs[\Big]{ \innerp[\Big]{ \Psi_N }{ \frac{1}{N^2} \sum\limits_{j < k}^{N} \big( w-w_0 \big) \myp{x_j-x_k} \Psi_N } } \nonumber \\
		&= \frac{1}{N^2} \abs[\Big]{\iint_{\mathbb{R}^6} \myp[\big]{ w-w_0 } \myp{x-y} \rho_{\Psi_N}^{\myp{2}} \myp{x,y} \id x \id y} \nonumber \\
		&\leq C \myp{ \normt{w_1-w_0}{\frac{3}{2}} + \normt{w_1-w_0}{\frac{5}{2}} + \normt{w_2}{\infty} }.
	\label{eq:interactionconv}
	\end{align}
	Note that by the bathtub principle (\cref{lem:functionalminima,lem:functionalminimastrong}) and the upper bound on kinetic energy \cref{lem:energykinlowbound} it follows that either
	\begin{align*}
		\MoveEqLeft[3]	\int_{\mathcal{S}} \myp[\Big]{ \int_{\mathbb{R}^3} \tau_{\beta} \myp{\rho_{\mu} \myp{x}} \id x} \id \mathrm{P} \myp{\mu} \\
		&\leq \frac{1}{\myp{2 \pi}^2} \frac{\beta}{1 + \beta} \int_{\mathcal{S}} \myp[\Big]{ \int_{\Omega} \myp{p^2+ k_{\beta} \myp{2j+ 1+s} } \mu \myp{\xi} \id \xi } \id \mathrm{P} \myp{\mu} 
		< \infty,
	\end{align*}
	or
	\begin{align*}
		\MoveEqLeft[3]	\frac{4 \pi^4}{3} \int_{\widetilde{\mathcal{S}}} \myp[\Big]{ \int_{\mathbb{R}^3} \rho_{\mu} \myp{x}^3 \id x} \id \mathrm{P} \myp{\mu} \\
		&\leq \frac{1}{\myp{2 \pi}^2} \int_{\widetilde{\mathcal{S}}} \myp[\Big]{ \iint_{\mathbb{R}^3 \times \mathbb{R}} p^2 \mu\myp{x,p} \id p \id x} \id \mathrm{P} \myp{\mu} 
		< \infty,
	\end{align*}
	depending on the sequence $ \myp{\beta_N} $.
	Applying either the bound \eqref{eq:rhoest2} or Markov's inequality leads to the conclusion that
	\begin{equation*}
		\int_{\mathcal{S}} \normt{\rho_{\mu}}{\frac{5}{3}} \id \mathrm{P} \myp{\mu} 
		< \infty.
	\end{equation*}
	Hence we can use Young's inequality to obtain
	\begin{align*}
		\MoveEqLeft[3]	\abs[\Big]{\int_{\mathcal{S}} \myp[\Big]{\int_{\mathbb{R}^6} \myp[\big]{ w-w_0 } \myp{x-y} \rho_{\mu} \myp{x} \rho_{\mu} \myp{y} \id x \id y} \id \mathrm{P} \myp{\mu}} \\
		&\leq \int_{\mathcal{S}} \normt{w_1-w_0}{\frac{5}{2}} \normt{\rho_{\mu}}{\frac{5}{3}} + \normt{w_2}{\infty} \id \mathrm{P} \myp{\mu} 
		< \infty.
	\end{align*}
	This bound together with \eqref{eq:interactionconv} implies that it suffices to show \eqref{eq:energyinterlowbound} for $ w \in C_c \myp{\mathbb{R}^3} $.
	However, the convergence holds in this case by \cref{lem:stateconv,cor:stateconvstrong}, since the function $ \myp{x,y} \mapsto w \myp{x-y} $ is bounded and uniformly continuous on $ \mathbb{R}^3 \times \mathbb{R}^3 $.
\end{proof}
\begin{proof}[Proof of \cref{thm:energylowbound} (and \cref{thm:convstates}) for strong fields]
	Assume first that $ \beta_N \to \beta \in \myp{0,\infty} $.
	It follows from \cref{lem:energykinlowbound,lem:energypotlowbound} that for any sequence $ \Psi_N \in \bigwedge^N L^2 \myp{\mathbb{R}^3; \mathbb{C}^2} $ satisfying $ \innerp{ \Psi_N }{ H_{N,\beta_N} \Psi_N } = E \myp{N,\beta_N} + o\myp{N} $, then along the subsequence $ N_{\ell} $ from \cref{lem:stateconv},
	\begin{equation*}
		\liminf_{\ell \to \infty} \frac{\innerp{ \Psi_{N_\ell} }{ H_{N_\ell,\beta_{N_\ell}} \Psi_{N_\ell} }}{N_\ell}
		\geq \int_{\mathcal{S}} \mathcal{E}_{\beta}^{\sssup{Vla}} \myp{\mu} \id \mathrm{P} \myp{\mu}
		\geq E^{\sssup{MTF}} \myp{\beta},
	\end{equation*}
	where the last inequality follows from \cref{lem:functionalminima} and the fact that $ \mathrm{P} $ is a probability measure.
	Assume for the sake of contradiction that $ E^{\sssup{MTF}} \myp{\beta} > \liminf_{N} \frac{E \myp{N,\beta_N}}{N} $, and take a sequence $ M_k \in \mathbb{N} $ satisfying
	\begin{equation*}
		\lim_{k \to \infty} \frac{E \myp{M_k,\beta_{M_k}}}{M_k} = \liminf_N \frac{E \myp{N,\beta_N}}{N}.
	\end{equation*}
	Since we might as well have proven \cref{lem:measurelim} and \cref{lem:stateconv} starting from this sequence, we may assume that $ N_\ell $ is a subsequence of $ M_k $.
	Hence
	\begin{equation*}
		E^{\sssup{MTF}} \myp{\beta}
		> \liminf_{N \to \infty} \frac{E \myp{N,\beta_N}}{N}
		= \liminf_{\ell \to \infty} \frac{\innerp{ \Psi_{N_\ell} }{ H_{N_\ell,\beta_{N_\ell}} \Psi_{N_\ell} }}{N_\ell}
		\geq E^{\sssup{MTF}} \myp{\beta},
	\end{equation*}
	which is absurd, so we must have equality everywhere (using the already proven upper energy bound in \cref{prop:energyupbound}), concluding the proof of Theorem \ref{thm:energylowbound} for $ 0 < \beta < \infty $.
	In particular, we also have
	\begin{equation*}
		\int_{\mathcal{S}} \mathcal{E}_{\beta}^{\sssup{Vla}} \myp{\mu} - E^{\sssup{MTF}} \myp{\beta} \id \mathrm{P} \myp{\mu} =0,
	\end{equation*}
	so $ \mathrm{P} $ is supported on the set of minimizers of the Vlasov energy functional.
	Hence $ \mathrm{P} $ induces a probability measure on the set of minimizers of the magnetic Thomas-Fermi functional, completing the proof of the first part of
	\cref{thm:convstates}.
	
	In the case where $ \myp{\beta_N} $ satisfies $ \beta_N \to \infty $ and \eqref{eq:scaling2}, we apply the same argument, obtaining
	\begin{equation*}
		E^{\sssup{STF}}
		\geq \liminf_{N \to \infty} \frac{\innerp{ \Psi_N }{ H_{N,\beta_N} \Psi_N }}{N}
		\geq \int_{\widetilde{\mathcal{S}}} \mathcal{E}_{\infty}^{\sssup{Vla}} \myp{\mu} \id \mathrm{P} \myp{\mu}
		\geq E^{\sssup{STF}}.
	\end{equation*}
	In this case $ \mathrm{P} $ induces a measure on the set of minimizers of the strong Thomas-Fermi functional, completing the proof of \cref{thm:convstates}, except for the case when $ \beta_N \to 0 $.
\end{proof}

%%%%%%%%%%%%%%%%%%%%%%%%%%%%%%%%%%%%%%%%%%%%%%%%%%%%%
%%%%%%%%%%%%%%%%%%%%%%%%%%%%%%%%%%%%%%%%%%%%%%%%%%%%%
\section{Lower energy bounds, weak fields}
\label{sec:lowerbound2}
%%%%%%%%%%%%%%%%%%%%%%%%%%%%%%%%%%%%%%%%%%%%%%%%%%%%%
%%%%%%%%%%%%%%%%%%%%%%%%%%%%%%%%%%%%%%%%%%%%%%%%%%%%%
Here we consider the case where $ \beta_N \to 0 $ as $ N \to \infty $.
Again, suppose that $ V $ and $ w $ satisfy the assumptions of \cref{thm:energylowbound}, and let $ f \in C_c^{\infty} \myp{\mathbb{R}^3} $ be a real-valued, even and $ L^2 $-normalized function.
Since the distance between the Landau bands of $ \myp{\bm{\sigma} \cdot \myp{- i \hbar \nabla + b A \myp{x}}}^2 $ is $ 2 \hbar b $, and $ \beta_N \to 0 $ is equivalent to $ \hbar b = \beta_N \myp{1+\beta_N}^{-2/5} \to 0 $, we can argue without diagonalising the magnetic Laplacian as in the beginning of \cref{sec:semimeasures}.
In other words, we get the usual phase space $ \mathbb{R}^3 \times \mathbb{R}^3 \times \Set{\pm 1} $.

This means that we can follow \cite{FouLewSol-18} in our construction of the semi-classical measures, but we do, however, need a slight rescaling (exclusively to control an error term in the case where $ b \to \infty $).
Thus, in addition to $ \hbar > 0 $, we also introduce an auxiliary parameter $ \alpha > 0 $ and put
\begin{equation}
\label{eq:coherentstate}
	f_{x,p}^{\hbar,\alpha} \myp{y} = \myp{\hbar \alpha}^{-\frac{3}{4}} f \myp[\Big]{\frac{y-x}{\sqrt{\hbar \alpha}}} e^{i \frac{p \cdot y}{\hbar}},
\end{equation}
and we further define $ f^{\hbar\alpha} = f_{0,0}^{\hbar,\alpha} $ and $ g^{\hbar,\alpha} = \mathcal{F}_{\hbar} \myb{f^{\hbar \alpha}} $.
Then we have a resolution of the identity
\begin{equation*}
	\frac{1}{\myp{2 \pi \hbar}^3} \iint_{\mathbb{R}^3 \times \mathbb{R}^3} \ketbra{f_{x,p}^{\hbar,\alpha}}{f_{x,p}^{\hbar,\alpha}} \id x \id p 
	= \mathds{1}_{L^2 \myp{\mathbb{R}^3}}.
\end{equation*}
Now, denoting $ P_{x,p}^{\hbar,\alpha} = \ketbra{f_{x,p}^{\hbar,\alpha}}{f_{x,p}^{\hbar,\alpha}} \mathds{1}_{\mathbb{C}^2} $, we define the $ k $-particle semi-classical measures
\begin{equation*}
	m_{f,\Psi_N}^{\myp{k}} \myp{x,p,s} = \frac{N!}{\myp{N-k}!} \innerp[\Big]{ \Psi_N }{ \myp[\Big]{ \bigotimes_{\ell=1}^k P_{x_{\ell},p_{\ell}}^{\hbar,\alpha} \mathcal{P}_{s_{\ell}} } \otimes \mathds{1}_{N-k} \Psi_N }_{L^2 \myp{ \mathbb{R}^{3N};\mathbb{C}^{2^N} } },
\end{equation*}
where $ \Psi_N \in \bigwedge^N L^2 \myp{\mathbb{R}^3 ; \mathbb{C}^2} $ is any wave function.  The parameter $ \alpha $ is arbitrary for now, but later we will settle on a specific choice, see \eqref{eq:alphascale} below.
Going through the proofs of Lemmas $ 2.2 $, $ 2.4 $, and $ 2.5 $ in \cite{FouLewSol-18} we get the same properties of the semi-classical measures as before:
\begin{lem}
\label{lem:measureprop}
	The function $ m_{f,\Psi_N}^{\myp{k}} $ is symmetric on $ \myp{\mathbb{R}^3 \times \mathbb{R}^3 \times \Set{\pm 1}}^k $ and satisfies
	\begin{equation}
	\label{eq:measurebound}
		0 \leq m_{f,\Psi_N}^{\myp{k}} \leq 1,
	\end{equation}
	\begin{equation}
	\label{eq:compatibility1}
		\frac{1}{\myp{2 \pi \hbar}^{3k}} \sum_{s \in \Set{\pm 1}^k} \iint_{\mathbb{R}^{3k} \times \mathbb{R}^{3k}} m_{f,\Psi_N}^{\myp{k}} \myp{x, p,s} \id x \id p
		= \frac{N!}{\myp{N-k}!},
	\end{equation}
	and for $ k \geq 2 $,
	\begin{align}
		\MoveEqLeft[3]	\frac{1}{\myp{2 \pi \hbar}^3} \sum_{s_k = \pm 1} \iint_{\mathbb{R}^3 \times \mathbb{R}^3} m_{f,\Psi_N}^{\myp{k}} \myp{x_1,p_1,s_1, \dotsc, x_k,p_k,s_k} \id x_k \id p_k \nonumber \\
		&= \myp{N-k+1} m_{f,\Psi_N}^{\myp{k-1}} \myp{x_1,p_1,s_1, \dotsc, x_{k-1},p_{k-1},s_{k-1}}.
	\label{eq:compatibility2}
	\end{align}
\end{lem}
\begin{lem}[Position densities]
\label{lem:momentint}
	Supposing that $ f $ is real, $ L^2 $-normalized and even, we have for $ 1 \leq k \leq N $ and any normalized $ \Psi_N $ that
	\begin{equation}
		\frac{1}{\myp{2 \pi \hbar}^{3k}} \int_{\mathbb{R}^{3k}} m_{f,\Psi_N}^{\myp{k}} \myp{x,p,s} \id p
		= k! \myp{ \widetilde{\rho}_{\Psi_N}^{\, \myp{k}} \ast \myp{ \abs{ f^{\hbar\alpha} }^2 }^{\otimes k} } \myp{x,s}.
	\end{equation}
\end{lem}
\begin{lem}[Kinetic energy]
\label{lem:kinenergy}
	Suppose that $ \Psi_N \in \bigwedge^N H_{\hbar b^{-1}A}^1 \myp{\mathbb{R}^3} $ is normalized in $ L^2 $ and satisfies $ A \Psi_N \in L^2 \myp{\mathbb{R}^3;\mathbb{R}^3} $, and that $ f \in C_c^{\infty} \myp{\mathbb{R}^3} $ is real-valued, $ L^2 $-normalised and even.
	Then we have
	\begin{align}
		\MoveEqLeft[3]	\innerp[\Big]{ \Psi_N }{ \sum\limits_{j=1}^N \myp{ \bm{\sigma} \cdot \myp{-i\hbar \nabla_j + bA \myp{x_j}}}^2 \Psi_N } \nonumber \\
		={}& \frac{1}{\myp{2 \pi \hbar}^3} \sum_{s= \pm 1} \iint_{\mathbb{R}^6} \abs{p + bA\myp{x}}^2 m_{f,\Psi_N}^{\myp{1}} \myp{x,p,s} \id x \id p \nonumber \\
		&+ \hbar b N \innerp[\big]{ \Psi_N }{ \sigma_3 \Psi_N } - \frac{\hbar}{\alpha} N \int \abs{ \nabla f }^2 \nonumber \\
		&+ 2b \mathrm{Re} \innerp[\Big]{ \Psi_N }{ \sum\limits_{j=1}^N \myp[\big]{ A - A \ast \abs{ f^{\hbar \alpha} }^2 } \myp{x_j} \cdot \myp{-i \hbar \nabla_j} \Psi_N } \nonumber \\
		&+ b^2 \innerp[\Big]{ \Psi_N }{ \sum\limits_{j=1}^N \myp[\big]{ \abs{A}^2 - \abs{A}^2 \ast \abs{ f^{\hbar \alpha} }^2 } \myp{x_j} \Psi_N }.
	\label{eq:kinenergy}
	\end{align}
\end{lem}
\begin{lem}[Estimation of error terms]
\label{lem:errorterms}
	With our specific choice of magnetic potential, $ A \myp{x} = \frac{1}{2} \myp{-x_2,x_1,0} $, we have
	\begin{equation*}
		\innerp[\Big]{ \Psi_N }{ \sum\limits_{j=1}^N \myp[\big]{ A - A \ast \abs{ f^{\hbar \alpha} }^2 } \myp{x_j} \cdot \myp{-i \hbar \nabla_j} \Psi_N }
		=0,
	\end{equation*}
	and
	\begin{equation*}
		\abs[\Big]{ \innerp[\Big]{ \Psi_N }{ \sum\limits_{j=1}^N \myp[\big]{ \abs{A}^2 - \abs{A}^2 \ast \abs{ f^{\hbar \alpha} }^2 } \myp{x_j} \Psi_N } }
		\leq C N \hbar \alpha.
	\end{equation*}
\end{lem}
\begin{proof}
	By direct computation,
	\begin{equation*}
		\abs{A}^2 \myp{x} - \abs{A}^2 \myp{y} - \nabla \abs{A}^2 \myp{x} \cdot \myp{x-y} = - \frac{1}{4} \myp{x_{\perp} - y_{\perp}}^2,
	\end{equation*}
	implying for any $ x \in \mathbb{R}^3 $, since $ f $ is even, that
	\begin{align*}
		\MoveEqLeft[3]	\abs[\big]{\abs{A}^2 \myp{x} - \abs{A}^2 \ast \abs{ f^{\hbar \alpha} }^2 \myp{x}} \\
		&= \abs[\Big]{\int \myp{ \abs{A}^2 \myp{x} - \abs{A}^2 \myp{y} } \abs{ f^{\hbar \alpha} \myp{x-y} }^2 \id y} \\
		&= \abs[\Big]{\int \myp{ \nabla \abs{A}^2 \myp{x} \cdot \myp{x-y} - \frac{1}{4} \myp{x_{\perp} - y_{\perp}}^2 } \abs{ f^{\hbar \alpha} \myp{x-y} }^2 \id y} \\
		&= \frac{1}{4} \int y_{\perp}^2 \abs{ f^{\hbar \alpha} \myp{y} }^2 \id y 
		= C \hbar \alpha.
	\end{align*}
	On the other hand, since $ A $ is linear and $ f $ is even,
	\begin{align*}
		A \myp{x} - A \ast \abs{ f^{\hbar \alpha} }^2 \myp{x}
		&= \int \myp{A \myp{x} - A \myp{y}} \abs{ f^{\hbar \alpha} \myp{x-y} }^2 \id y \\
		&= \int A \myp{y} \abs{ f^{\hbar \alpha} \myp{y} }^2 \id y = 0,
	\end{align*}
	so the error term in \eqref{eq:kinenergy} involving $ A - A \ast \abs{f^{\hbar \alpha}}^2 $ is simply not present in our case.
\end{proof}
Because of the last error term in \eqref{eq:kinenergy}, we need to distinguish two cases, depending on how fast the parameter $ \beta_N $ tends to zero.
If $ b = N^{1/3} \beta_N \myp{1+\beta_N}^{-3/5} $ is bounded from above, we can take $ \alpha = 1 $.
If, on the other hand, $ \beta_N $ goes to zero slowly enough such that $ b \to \infty $, we instead take $ \alpha = b^{-1} $.
Then all the error terms in \eqref{eq:kinenergy} will be of order at most $ \hbar b $, and furthermore $ \hbar \alpha \to 0 $, so $ \abs{f^{\hbar \alpha}}^2 $ is still an approximate identity.
For the sake of brevity we will treat both cases simultaneously by choosing
\begin{equation}
\label{eq:alphascale}
	\alpha = \myp{1+b}^{-1}.
\end{equation}
By combining \cref{lem:kinenergy,lem:errorterms} we have, since also $ \hbar \alpha^{-1} \to 0 $,
\begin{align}
\label{eq:kinapprox}
	\MoveEqLeft[3]	\innerp[\Big]{ \Psi_N }{ \sum\limits_{j=1}^N \myp{ \bm{\sigma} \cdot \myp{-i\hbar \nabla_j + bA \myp{x_j}}}^2 \Psi_N } \nonumber \\
	&= \frac{1}{\myp{2 \pi \hbar}^3} \sum_{s= \pm 1} \iint_{\mathbb{R}^6} \abs{p + bA\myp{x}}^2 m_{f,\Psi_N}^{\myp{1}} \myp{x,p,s} \id x \id p + o \myp{N}.
\end{align}
When $ b $ is unbounded, a slight complication arises from the fact that we cannot obtain tightness in the momentum variables of the semi-classical measures, due to the presence of $ b A \myp{x} $ in the above approximation.
We can, however, circumvent this by doing a simple translation in the momentum variables.
Note that doing this will not change the position densities of the measures.

For $ x \in \mathbb{R}^{3k} $ we denote $ \widetilde{A}\myp{x} = \myp{A\myp{x_1}, \dotsc, A \myp{x_k}} $ and define
\begin{equation*}
	\widetilde{m}_N^{\myp{k}} \myp{x,p,s} = m_{f,\Psi_N}^{\myp{k}} \myp{ x, p-b \widetilde{A} \myp{x} ,s }.
\end{equation*}
Then the family of sequences $ \myp{ \widetilde{m}_N^{\myp{k}} }_{N\geq k} $ still satisfies \cref{lem:measureprop,lem:momentint}, and in exactly the same way as in \cref{lem:measurelim}, we obtain (symmetric) weak limits $ \widetilde{m}^{\myp{k}} \in L^1 \myp{ \myp{\mathbb{R}^{6} \times \Set{\pm 1}}^k } \cap L^{\infty} \myp{ \myp{\mathbb{R}^{6} \times \Set{\pm 1}}^k } $ with $ 0 \leq \widetilde{m}^{\myp{k}} \leq 1 $ such that
\begin{align}
\label{eq:measurelimweak}
	\MoveEqLeft[3]	\sum_{s \in \Set{\pm 1}^k} \int_{\mathbb{R}^{6k}} \widetilde{m}_N^{\myp{k}} \myp{x,p,s} \varphi \myp{x,p,s} \id x \id p \nonumber \\
	&\longrightarrow \sum_{s \in \Set{\pm 1}^k} \int_{\mathbb{R}^{6k}} \widetilde{m}^{\myp{k}} \myp{x,p,s} \varphi \myp{x,p,s} \id x \id p
\end{align}
for all $ \varphi \in L^1 \myp{ \myp{\mathbb{R}^{6} \times \Set{\pm 1}}^k} + L_{\varepsilon}^{\infty} \myp{ \myp{\mathbb{R}^{6} \times \Set{\pm 1}}^k } $, as $ N $ tends to infinity.
\begin{lem}
	Suppose that $ \Psi_N \in \bigwedge^N L^2 \myp{\mathbb{R}^3; \mathbb{C}^2} $ is a sequence satisfying the energy bound $ \innerp{ \Psi_N }{ H_{N,\beta_N} \Psi_N } \leq C N $.
	Then we have
	\begin{enumerate}
		\item The sequence $ \myp{ \widetilde{m}_N^{\myp{k}} }_{N\geq k} $ is tight for each $ k \geq 1 $.
		
		\item The limit measures $ \widetilde{m}^{\myp{k}} $ are probability measures, i.e.,
		\begin{equation}
			\frac{1}{\myp{2 \pi}^{3k}} \sum_{s \in \Set{ \pm 1}^k} \iint_{\mathbb{R}^{3k} \times \mathbb{R}^{3k}} \widetilde{m}^{\myp{k}} \myp{x,p,s} \id x \id p = 1.
		\end{equation}
		
		\item The compatibility relation \eqref{eq:compatibility2} is preserved in the limit, that is, for $ k \geq 2 $ and almost every $ \myp{x,p,s} \in \myp{\mathbb{R}^3 \times \mathbb{R}^3 \times \Set{\pm 1}}^{k-1} $,
		\begin{equation}
			\frac{1}{\myp{2 \pi}^3} \sum_{s_k = \pm 1} \iint_{\mathbb{R}^6} \widetilde{m}^{\myp{k}} \myp{x,x_k;p,p_k;s,s_k} \id x_k \id p_k
			= \widetilde{m}^{\myp{k-1}} \myp{x,p,s}.
		\end{equation}
		
		\item The convergence in \eqref{eq:measurelimweak} holds for any $ \varphi $ in $ L^1 \myp{ \myp{\mathbb{R}^{6} \times \Set{\pm 1}}^k } + L^{\infty} \myp{ \myp{\mathbb{R}^{6} \times \Set{\pm 1}}^k } $.
	\end{enumerate}
\end{lem}
\begin{proof}
	We will only prove that $ \myp{ \widetilde{m}_N^{\myp{1}} } $ is a tight sequence.
	The rest follows in exactly the same way as in \cref{lem:limmeasures}.  Supposing that $ f $ is supported on a ball with radius $ K $ centred at the origin, we have by \cref{lem:momentint} for $ \hbar \alpha $ small that
	\begin{align*}
		\MoveEqLeft[3]	\int_{\abs{x} \geq R} \int_{\mathbb{R}^3} \widetilde{m}_N^{\myp{1}} \myp{x,p,s} \id p \id x \\
		&= \myp{2 \pi \hbar}^3 \int_{\abs{x} \geq R} \int_{\mathrm{supp} f^{\hbar \alpha}} \widetilde{\rho}_{\Psi_N}^{\, \myp{1}} \myp{x-y,s} \abs{ f^{\hbar \alpha} \myp{y} }^2 \id y \id x \\
		&\leq C \int_{\abs{x} \geq R-K} \frac{1}{N} \widetilde{\rho}_{\Psi_N}^{\, \myp{1}} \myp{x,s} \id x.
	\end{align*}
	Now, combining \cref{lem:ltbounds} with the fact that $ V $ is a confining potential, it follows that the right hand side above tends to zero uniformly in $ N $ as $ R $ tends to infinity, implying that $ \myp{ \widetilde{m}_N^{\myp{1}} } $ is tight in the position variable.
	Using the kinetic energy bound \eqref{eq:kinbound} combined with \eqref{eq:kinapprox} we also obtain
	\begin{align*}
		\MoveEqLeft[3] \int_{\abs{p}\geq R} \int_{\mathbb{R}^3} \widetilde{m}_N^{\myp{1}} \myp{x,p,s} \id x \id p
		= \int_{\mathbb{R}^3} \int_{\abs{p+bA\myp{x}}\geq R} m_{f,\Psi_N}^{\myp{1}} \myp{x,p,s} \id p \id x \\
		&\leq \frac{1}{R^2} \iint_{\mathbb{R}^6} \abs{p+bA\myp{x}}^2 m_{f,\Psi_N}^{\myp{1}} \myp{x,p,s} \id x \id p \leq \frac{C}{R^2},
	\end{align*}
	showing that $ \myp{ \widetilde{m}_N^{\myp{1}} } $ is also tight in the momentum variable.
\end{proof}
\begin{prop}[Convergence of states]
	Suppose that $ \beta_N \to 0 $, and that $ \Psi_N \in \bigwedge^N L^2 \myp{\mathbb{R}^3; \mathbb{C}^2} $ is a sequence of normalized wave functions satisfying the energy bound $ \innerp{ \Psi_N }{ H_{N,\beta_N} \Psi_N } \leq C N $.
	Then there exist a subsequence $ \myp{N_\ell} \subseteq \mathbb{N} $ and a unique Borel probability measure $ \mathrm{P} $ on the set
	\begin{equation*}
		\mathcal{S} = \Set[\big]{\mu \in L^1 \myp{\mathbb{R}^{6} \times \Set{\pm 1}} \given 0 \leq \mu \leq 1, \ \frac{1}{\myp{2 \pi}^3} \normt{\mu}{1} = 1},
	\end{equation*}
	such that for each $ k \geq 1 $ the following holds:
	\begin{enumerate}
		\item For all $ \varphi \in L^1 \myp{ \myp{\mathbb{R}^{6} \times \Set{\pm 1}}^k } + L^{\infty} \myp{ \myp{\mathbb{R}^{6} \times \Set{\pm 1}}^k } $, as $ \ell $ tends to infinity,
		\begin{align*}
			\MoveEqLeft[3]	\sum_{s \in \Set{\pm 1}^k} \iint_{\mathbb{R}^{6k}} \widetilde{m}_{N_{\ell}}^{\myp{k}} \myp{x,p,s} \varphi \myp{x,p,s} \id x \id p \\
			&\longrightarrow \int_{\mathcal{S}} \myp[\Big]{ \sum_{s \in \Set{\pm 1}^k} \iint_{\mathbb{R}^{6k}} \mu^{\otimes k} \myp{x,p,s} \varphi \myp{x,p,s} \id x \id p } \id \mathrm{P} \myp{\mu}.
		\end{align*}
		\item For $ U \in L^{5/2} \myp{\mathbb{R}^3 \times \Set{\pm 1}} + L^{\infty} \myp{\mathbb{R}^3 \times \Set{\pm 1}} $ if $ k=1 $, and for any bounded and uniformly continuous function $ U $ on $ \myp{ \mathbb{R}^{3} \times \Set{\pm 1} }^k $ if $ k \geq 2 $, as $ \ell $ tends to infinity,
		\begin{align}
		\label{eq:stateconvweak2}
			\MoveEqLeft[3]	\frac{k!}{N_\ell^k} \sum_{s \in \Set{\pm 1}^k} \int_{\mathbb{R}^{3k}} \widetilde{\rho}_{\Psi_{N_\ell}}^{\, \myp{k}} \myp{x,s} U \myp{x,s} \id x \nonumber \\
			&\longrightarrow \int_{\mathcal{S}} \myp[\Big]{ \sum_{s \in \Set{\pm 1}^k} \int_{\mathbb{R}^{3k}} \widetilde{\rho}_{\mu}^{\, \otimes k} \myp{x,s} U \myp{x,s} \id x } \id \mathrm{P} \myp{\mu},
		\end{align}
		where
		\begin{equation*}
			\widetilde{\rho}_{\mu} \myp{x,s} = \frac{1}{\myp{2 \pi}^3} \int_{\mathbb{R}^{3}} \mu \myp{x,p,s} \id p.
		\end{equation*}
	\end{enumerate}
\end{prop}
\begin{lem}
	Suppose that we have a sequence $ \Psi_N \in \bigwedge^N L^2 \myp{\mathbb{R}^3; \mathbb{C}^2} $ satisfying $ \innerp{ \Psi_N }{ H_{N,\beta_N} \Psi_N } = E \myp{N,\beta_N} + o\myp{N} $.
	Then, denoting by $ \rho_{\mu} $ the spin-summed position density of $ \mu $, we have
	\begin{align*}
		\MoveEqLeft[3]	\liminf_{N \to \infty} \innerp[\Big]{ \Psi_N }{ \frac{1}{N} \sum\limits_{j=1}^N \myp{\bm{\sigma} \cdot \myp{-i \hbar \nabla + b A \myp{x_j}}}^2 \Psi_N } \\
		&\geq \frac{1}{\myp{2 \pi}^3} \int_{\mathcal{S}} \myp[\Big]{ \sum_{s= \pm 1} \iint_{\mathbb{R}^6} p^2 \mu \myp{x,p,s} \id p \id x } \id \mathrm{P} \myp{\mu},
	\end{align*}
	\begin{equation*}
		\liminf_{N \to \infty} \innerp[\Big]{ \Psi_N }{ \frac{1}{N} \sum\limits_{j=1}^{N} V \myp{x_j} \Psi_N } 
		\geq \int_{\mathcal{S}} \myp[\Big]{ \int_{\mathbb{R}^3} V \myp{x} \rho_{\mu} \myp{x} \id x} \id \mathrm{P} \myp{\mu},
	\end{equation*}
	and
	\begin{align*}
		\MoveEqLeft[3]	\lim_{N \to \infty} \innerp[\Big]{ \Psi_N }{ \frac{1}{N^2} \sum\limits_{1 \leq j < k \leq N} w \myp{x_j-x_k} \Psi_N } \\
		&= \frac{1}{2} \int_{\mathcal{S}} \myp[\Big]{ \int_{\mathbb{R}^6} w \myp{x-y} \rho_{\mu} \myp{x} \rho_{\mu} \myp{y} \id x \id y} \id \mathrm{P} \myp{\mu}.
	\end{align*}
\end{lem}
\begin{proof}
	For the kinetic energy term, simply use that
	\begin{align*}
		\MoveEqLeft[3]	\innerp[\Big]{ \Psi_N }{ \sum\limits_{j=1}^N \myp{ \bm{\sigma} \cdot \myp{-i\hbar \nabla_j + bA \myp{x_j}}}^2 \Psi_N } \\
		&= \frac{1}{\myp{2 \pi \hbar}^3} \sum_{s= \pm 1} \iint_{\mathbb{R}^6} p^2 \widetilde{m}_N^{\myp{1}} \myp{x,p,s} \id x \id p + o \myp{N},
	\end{align*}
	and proceed as in the proof of \cref{lem:energykinlowbound}.
	The convergence of the potential energy terms follows exactly as in \cref{lem:energypotlowbound}.
\end{proof}
\begin{proof}[Proof of \cref{thm:energylowbound} and \cref{thm:convstates}, weak fields]
	We find exactly as in the previous cases that
	\begin{equation*}
		E_0^{\sssup{Vla}}
		\geq \liminf_{N \to \infty} \frac{\innerp{ \Psi_N }{ H_{N,\beta_N} \Psi_N }}{N}
		\geq \int_{\mathcal{S}} \mathcal{E}_0^{\sssup{Vla}} \myp{\mu} \id \mathrm{P} \myp{\mu}
		\geq E_0^{\sssup{Vla}},
	\end{equation*}
	finishing the proof of \cref{thm:energylowbound}, and implying that the de Finetti measure $ \mathrm{P} $ is supported on the set of minimizers of $ \mathcal{E}_0^{\sssup{Vla}} $.
	Since these are independent of the spin variable by \cref{rem:vlamin}, the convergence of states \eqref{eq:stateconvweak2} becomes
	\begin{align*}
		\MoveEqLeft[3]	\frac{k!}{N_\ell^k} \sum_{s \in \Set{\pm 1}^k} \int_{\mathbb{R}^{3k}} \widetilde{\rho}_{\Psi_{N_\ell}}^{\, \myp{k}} \myp{x,s} U \myp{x,s} \id x \\
		&\longrightarrow \int_{\mathcal{S}} \myp[\Big]{ \frac{1}{2^k} \sum_{s \in \Set{\pm 1}^k} \int_{\mathbb{R}^{3k}} \rho_{\mu}^{\, \otimes k} \myp{x} U \myp{x,s} \id x } \id \mathrm{P} \myp{\mu}.
	\end{align*}
	Using $ \mathrm{P} $ to induce a measure on the set of minimizers of the Thomas-Fermi functional concludes the proof of \cref{thm:convstates}.
\end{proof}

\subsection*{\textbf{Acknowledgement}}
The authors were partially supported by the Sapere Aude grant DFF--4181-00221 from the Independent Research Fund Denmark.
Part of this work was carried out while both authors visited the Mittag-Leffler Institute in Stockholm, Sweden.

\bibliographystyle{siam} %style de biblio
\bibliography{biblio} %nom du fichier biblio.bib (dans le meme dossier que le fichier latex qu'on compile)
\end{document}